\documentclass[a4paper,10pt]{article}
\usepackage[utf8x]{inputenc}
\usepackage[affil-it]{authblk}
\usepackage{amsmath,amsthm,latexsym,amssymb,amsfonts}
\usepackage{units}
\usepackage{amsbsy}
\usepackage{amssymb}
\usepackage{mathrsfs}
\usepackage{bbm}
\usepackage{color}
\usepackage{graphicx}
\usepackage{hyperref}
\usepackage{cite}
 \hypersetup{colorlinks=true, linkcolor=black, urlcolor=black, citecolor=black, bookmarksopen}

\newcommand{\coh}{\nicefrac{1}{2}}

\makeatletter
\newtheorem*{rep@theorem}{\rep@title}
\newcommand{\newreptheorem}[2]{%
\newenvironment{rep#1}[1]{%
 \def\rep@title{#2 \ref{##1}}%
 \begin{rep@theorem}}%
 {\end{rep@theorem}}}
\makeatother

\newtheorem{lm}{Lemma}
\newreptheorem{lm}{Lemma}

\newtheorem{df}{Definition}

\DeclareMathOperator{\Tr}{Tr}
\DeclareMathOperator{\rd}{d}
\DeclareMathOperator{\sign}{sgn}
\DeclareMathOperator{\diver}{div}
\DeclareMathOperator{\Ran}{Ran}

\numberwithin{equation}{section}

\title{Coherent states, $6j$ symbols and properties of the next to leading order asymptotic expansions}
\author[1,2,3]{Wojciech Kami{\'n}ski%
\thanks{wkaminsk@fuw.edu.pl}}
\author[2,3]{Sebastian Steinhaus%
\thanks{steinhaus.sebastian@gmail.com}}
\affil[1]{Wydzia{\l} Fizyki, Uniwersytet Warszawski, Ho{\.z}a 69, 00-681, Warsaw, Poland}
\affil[2]{Perimeter Institute for Theoretical Physics,
31 Caroline Street North,
Waterloo, Ontario,
Canada N2L 2Y5}
\affil[3] {Max Planck Institute for Gravitational Physics,
Am M\"uhlenberg 1, D-14476 Potsdam, Germany
}
\date{\vspace{-5ex}}

\begin{document}

\maketitle

\begin{abstract}
We present the first complete derivation of the well-known asymptotic expansion of the $SU(2)$ $6j$
symbol using a coherent state approach, in particular we succeed in computing the determinant
of the
Hessian matrix. To do so, we smear the coherent states and perform a partial stationary point
analysis with respect to the smearing parameters. This allows us to transform the variables from
group elements to dihedral angles of a tetrahedron resulting in an effective action, which
coincides with the action of first order Regge
calculus associated to a tetrahedron. To perform the remaining stationary point analysis, we compute
its Hessian matrix and obtain the correct measure factor. Furthermore, we expand the
discussion of the asymptotic formula to next to leading order terms, prove some of their properties
and derive a recursion relation for the full $6j$ symbol.
\end{abstract}

\section{Introduction}

Spin foam models \cite{spinfoams,spinfoams2,spinfoams3,Baez} are candidate models for quantum gravity invented as a generalization of Feynman diagrams to higher dimensional objects.
Their popularity is rooted in the fact that they were well adapted to descibe $3$D Quantum Gravity
theories such as the Ponzano-Regge \cite{PR,pr-model} or the Turaev-Viro model
\cite{turaev-viro}. To examine whether these models are a quantum theory of $4$D General Relativity,
in
particular whether one obtains Gravity in a semi-classical limit is an active area of topical
research. One of the strongest positive implications comes from the asymptotic analysis of single
simplices in spin foam models: A first attempt to compute the asymptotic expansion of the amplitude associated to a $4$-simplex in the Barrett-Crane model \cite{Barrett-Crane} can be found in \cite{4dWilliams,baez10j}. This was continued for the square of (the Euclidean and Lorentzian) $6j$ and $10j$ symbols in \cite{Freidel}, whereas the most recent asymptotic results for modern spin foam models, i.e. the EPRL-model \cite{EPRL} or the FK-model \cite{FK}, were obtained using a coherent state approach \cite{Conrady-Freidel,Frank,Frank2,Frank3,FrankEPRL,Frank3D}:

The basic amplitudes of the spin foam model are their vertex amplitudes ($SU(2)$ $6j$
symbols in the $3$D Ponzano Regge model). They are defined in a representation theoretic
way and can be constructed from coherent states of the underlying Lie group \cite{Perelmov} as
a multidimensional integral to which the stationary point approximation is applicable
\cite{Hormander}.
This method has proven to be very efficient in determining the dominating phase in the
asymptotic formula as well as the
geometric interpretation of the contributions to the asymptotic expansion in spin foam models
\cite{Conrady-Freidel,Frank,Frank2,Frank3,FrankEPRL,Frank3D}. In $3$D, on the points of stationary phase, $6j$
symbols are geometrically interpreted as
tetrahedra, their dominating phase given by the Regge action \cite{regge,regge2}, a discrete
version of General Relativity on a triangulation. Similar results were proven by this method for the
$4$-simplex \cite{4dWilliams,Conrady-Freidel,Frank,Frank2,Frank3,FrankEPRL} in spin foam models. Until today, this is still one of the
most promising evidences that spin foam models are viable Quantum Gravity theories.

Despite this success, the coherent state approach fails to produce the full amplitude. It has not
yet been possible to compute the so-called measure factor, a proportionality constant
 (depending on the representation labels) in the
asymptotic expansion, which is given by the determinant of the matrix of second
derivatives, i.e. the Hessian matrix, evaluated on the stationary point. This failure even
applies to the simplest spin foam model in $3$D, the Ponzano-Regge model \cite{pr-model}, whose
vertex amplitude is the $SU(2)$ $6j$ symbol. To the authors'
best knowledge the Ponzano and Regge formula \cite{PR} has not yet been obtained this way;
we can only refer to numerical results in
\cite{Frank3D}. This is particularly troubling for the
coherent state approach, since the full asymptotic formula for $SU(2)$ $6j$ symbols  introduced in
\cite{PR} has been proven in many different ways, for example, by geometric quantization
\cite{Roberts}, Bohr-Sommerfeld approach \cite{Littlejohn}, Euler-MacLaurin approximation
\cite{Gurau} or the character integration method \cite{Freidel}.

The source of the problem is the size of the Hessian matrix and the lack of immediate geometric
formulas for its determinant. For the $6j$ symbol,
for example, this matrix is $9$ dimensional and its entries are basis dependent.
 This is a major drawback of the coherent state
approach, in particular, since the full expansion is necessary to discuss and examine the properties
of spin foams models of Quantum Gravity.
To obtain this  measure factor and
compare it to other approaches \cite{Dittrich:2011vz,Hamber:1997ut}, the complete asymptotic
expansion
is indispensable. This is an important open issue for $4$D spin foam models.

Our approach to overcome this problem can be seen (as we will show in Appendix
\ref{rel-islas}) as a combination of the coherent state approach \cite{Conrady-Freidel,Frank,Frank2,Frank3,FrankEPRL,Frank3D} and the propagator
kernel method \cite{islas}. It inherits nice geometric properties from the coherent state analysis
with a similar geometric interpretation of the points of stationary phase. Moreover, the
Hessian matrix is always described in terms of geometrical quantities and, most importantly, its
determinant can be computed for the $6j$ symbol.

In addition to the computation of the asymptotic formula of the $6j$ symbol \cite{PR}, our
approach allows us to propose a new way to compute higher order corrections to the asymptotic
expansion. These corrections have already been discussed in \cite{LD,LD2}: it was conjectured that the asymptotic expansion has an alternating form
\begin{equation} \label{eq:conjecture_intro}
 \{6j\}=A_0\cos\left(\sum
\left(j_i+\frac{1}{2}\right)\theta_i+\frac{\pi}{4}\right)+
A_1\sin\left(\sum
\left(j_i+\frac{1}{2}\right)\theta_i+\frac{\pi}{4}\right)+\ldots \quad ,
\end{equation}
where $A_n$ are consecutive higher order corrections and homogeneous functions in
$j+\frac{1}{2}$. Our method allows us to prove this conjecture to any order in the asymptotic
expansion.

\subsection{Coherent states and integration kernels}

The coherent state approach is based on {the following} principle: Invariants (under the action of the group) can be constructed by integration of a tensor product of vectors (living in the tensor product of
vector spaces of irreducible representations) over the group, i.e. group averaging. Since the  invariant subspace of the tensor product of three representations of $SU(2)$ is
one-dimensional,  the invariant is uniquely defined up to normalization. However, in order to
apply the stationary point analysis the vectors in the construction above cannot be chosen
arbitrarily. The choice, from which the method takes its name, is the coherent states class, which
consists of eigenvectors of the generators of
rotations with highest eigenvalues \cite{Perelmov}. Although these states are very effective in
obtaining the dominating phase of the amplitude, the associated Hessian matrix turns out to be very
complicated. This problem occurs since the action is not purely imaginary, which is also related to
the problem of choice of phase for the coherent states which has not yet been fully
understood.

Both latter problems disappear if, instead of eigenstates with maximal eigenvalues, we take
null
eigenvectors for a generator of rotations $L$. Since this vector is
trivially invariant with respect to rotations generated by $L$, the phase problem
disappears. Similarly the contraction of invariants can be expanded in terms of an action
that actually is purely imaginary. There is a trade-off, though: The quantity of stationary
points increases and their geometric interpretation becomes more complicated. Moreover,
frequently there exist no such eigenvectors for certain representations, (half-integer spins for
$SU(2)$) and their tensor
product gives thus vanishing invariants.

The solution to these issues comes from the simple observation that null eigenvectors
can be obtained by the integration of a coherent state, pointing in direction perpendicular to
the axis of $L$, over the rotations generated by $L$. Like that the geometric
interpretation usually obtained when using coherent states is restored. Furthermore, if we first
perform the partial stationary phase approximation with respect to
the additional circle variables, we obtain a purely imaginary action. In the special case of
the $6j$ symbol, our construction allows us to write the invariant purely in terms of edge lengths
and dihedral angles of a tetrahedron, in particular we perform a variable transformation from group
elements to dihedral angles of the tetrahedron. The resulting phase of the integral is given be the
first order Regge action \cite{Barrett:1994nn}.

\subsection{Relation to discrete Gravity}
\label{sec:rel-discrete}

Regge calculus \cite{regge,regge2} is a discrete version of General Relativity defined upon a
triangulation of the manifold. Influenced by Palatini's formulation, a first order
Regge calculus was derived in \cite{Barrett:1994nn}, in which both edge lengths and dihedral angles
are considered as independent variables and their respective equations of motion are first order
differential equations. Additional constraints on the angles have to be imposed in order to
reobtain their geometric interpretation\footnote{The vanishing of the
angle Gram matrix on flat spacetime implies the existence  of the flat $n$-simplex with the given
angles.} once the equations of motion for the angles have been
solved. Our derivation of the Ponzano-Regge formula shows astonishing similarity to this procedure.
Moreover, from our calculation one can deduce a suitable measure for first order (linearized) Quantum Regge
calculus, such that the expected Ponzano-Regge factor $\frac{1}{\sqrt{V}}$ appears, which naturally leads to a triangulation invariant measure \cite{Dittrich:2011vz}.

Another version of $4$D Regge calculus was explored in \cite{Dittrich:2008va} with
areas of triangles and (a class of) dihedral angles as fundamental and independent variables.
Several local constraints guarantee that the geometry of a $4$-simplex is uniquely determined.
These variables were chosen in the pursuit to better understand the relation between discrete
gravity and $4$D spin foam models. The latter are based on a similar paradigm as the Ponzano-Regge or
the Turaev-Viro models \cite{PR,turaev-viro} in $3$D, yet enhanced by the implementation of the
simplicity constraints from the Plebanski formulation of
General Relativity \cite{plebanski}. Area-angle variables as a discretization of Plebanski rather than
Einstein-Hilbert formulation were conjectured to be more suitable to describe the semi-classical limit of
those models.

Although it is known that the asymptotic limit of the
amplitude of a $4$-simplex for $4$D gravity models is proportional to the cosine of the Regge action
\cite{Frank,Frank2,Frank3}, the proportionality factor still
remains unknown. We hope that the method presented in this work can help in filling the gap.

\subsection{Problem of the next to leading order (NLO) and complete asymptotic expansion}

The asymptotic expansion for the $SU(2)$ $6j$ symbol, in particular for the next to leading
order (NLO), is still a scarcely examined issue, since it is very non-trivial to write the (NLO)
contributions in a compact form. Steps forward in this direction can be found in \cite{LD,LD2,Smerlak},
where the latter gives the complete expansion in the isosceles case of the $6j$ symbol.

The stationary point analysis applied in this work allows for a natural extension in a Feynman
diagrammatic approach. From this
approach the full expansion can be computed in principle, however in a very lengthy way. We derive a recursion relations of the Ward-Takesaki type, which is surprisingly similar to the one
invented in \cite{Bonzom,Bonzom2} however in very different context, that, basically can be used in the asymptotic expansion to derive the NLO in a more concise way. Moreover, we can show explicitly that the consecutive terms in the expansion
\eqref{eq:conjecture_intro} are of the conjectured `sin/cos' form.

\subsection{Organization of the paper}
\label{sec:organize}

This paper is organized as follows: In section \ref{sec:construction_invariants} we will
present our modified coherent states, how to use them to construct invariants and how to contract
these invariants to compute spin network amplitudes. The contracted invariants will be used to
define an action for the stationary point analysis, which will be examined whether it allows for the
same geometric interpretation on
its stationary points as other coherent state approaches \cite{Conrady-Freidel,Frank,Frank2,Frank3,FrankEPRL,Frank3D}. Its symmetries
as well as the group generated by the symmetry transformations will be
discussed. Section \ref{sec:evaluation-final} deals with the partial stationary point
analysis with respect to the introduced circle variables. This will allow us to write the amplitude,
after a variable transformation, purely in terms of angle variables, which will be identified as
exterior dihedral angles of a polyhedron. In section \ref{sec:first-order} we focus on the example
of the $6j$ symbol. After another variable
transformation, we obtain the action of first order Regge calculus and perform the remaining
stationary point analysis. Eventually we obtain the asymptotic formula from \cite{PR}. In
section
\ref{sec:NLO-DL} we prove the conjecture from \cite{LD,LD2} that the full asymptotic expansion is
of alternating form \eqref{eq:conjecture_intro} and derive the recursion relations for the full $6j$
symbol. We conclude with a discussion of the results and an outlook in section \ref{sec:discussion}.

We would like to point out that several results of this paper have been obtained by tedious
calculations which we did not include in its main part to improve readability. Interested readers
are welcome to look them up in the appendices.

\section{Modified coherent states, spin-network evaluations and symmetries}
\label{sec:construction_invariants}

In this section we are going to present the modified coherent states, how to construct the
spin-network evaluation
from them and that they allow for the same geometric interpretation in the stationary point analysis
as similar coherent state approaches. Furthermore the symmetries of the action will be investigated.

Consider a three-valent spin network, i.e. a graph with three-valent nodes carrying $SU(2)$
intertwiners and edges carrying irreducible representations of $SU(2)$. For each edge of the spin
network we introduce a (fiducial) orientation such that each node of the network can be denoted as
the `source' $s(e)$ or the `target' $t(e)$ of the edge $e$. Later in this work we intend to give a
geometrical meaning to the spin network, in terms of polyhedra, triangles, etc. so we denote the
set of nodes by $F$ and the set of edges by $E$, which will become the set of triangles / faces and set of edges of the
triangulation respectively. This dual identification is not always possible but we restrict our
attention to the case of planar (spherical) graphs, where such notions are natural.

\subsection{Intertwiners from modified coherent states}\label{sec:inv-def}

Intertwiners are invariant vectors (with respect to the action
of the group) in the tensor product of vector spaces
associated to irreducible representations of that group. In the case of $3$ irreducible
representations of $SU(2)$
the space of invariants is one dimensional  and, moreover, there is a unique choice for the invariant for a
given cyclic order of representations \cite{Penrose,Penrose2}.

Suppose $\xi\in V_{j_1}\otimes \cdots \otimes V_{j_n}$ is a vector in the tensor product of vector
spaces of
representations, then
\begin{equation}
 \int_{SU(2)}dU\ U\xi
\end{equation}
is  invariant under the action of $SU(2)$. If $\xi$  is chosen in a clever way, such
an invariant is  non-trivial. In the case of
three representations it must be proportional to the unique invariant defined in \cite{Penrose,Penrose2}.
In the following we present a choice which has the advantage that the method of stationary
phase can be directly applied.

For every face $f$, which is bounded by three edges, we choose a cyclic order of these edges
$(j_{fe_1},j_{fe_2},j_{fe_3})$, labelled by the carried representations.  These
choices influence the orientation of the spin network \cite{Penrose,Penrose2} and are used to define and
determine the sign of its amplitude, see also appendix \ref{sec:sign}. We introduce the
following
intertwiners for every face $f$:
\begin{equation} \label{eq:def_int}
 C_f=\int_{SU(2)}dU_f\ U_f\int\prod_{j} \frac{d\phi_{ji}}{2\pi}
f_f\left(\{\phi_{fe}\}_{e\in F}\right) \prod_{e\in F}
\left(O_{\phi_{fe}}|\coh\rangle\right)^{2j_e} \quad ,
\end{equation}
where $f_f$ is a function of the three angles $\phi_{fe}$, $e\subset f$, $|\coh\rangle$ is the basic state of the fundamental representation and $O_\phi$ is a rotation
matrix on $\mathbb{R}^2$:
\begin{equation}
 O_\phi=\begin{pmatrix}
         \cos\phi &\sin\phi\\-\sin\phi &\cos\phi
        \end{pmatrix} \quad .
\end{equation}
As mentioned above, \eqref{eq:def_int} is invariant under the action of $SU(2)$.

Before moving on, we would like to outline the key differences between the approach described above
and the usual coherent state approach \cite{Frank,Frank2,Frank3,Frank3D}.
\begin{itemize}
\item Coherent states of $SU(2)$ are labelled by vectors in $\mathbb{R}^3$. On
the stationary point with satisfied reality conditions, one obtains the geometric interpretation
that for every face these three vectors form the edge vectors of a triangle. Later on
we will prove
the same geometric interpretation for the invariant $C_f$.
\item Furthermore we smear the coherent state by a rotation, which is the key ingredient of our
approach. In addition to the stationary point analysis with respect to the $\{U_f\}$, we will also
perform a stationary point analysis for the smearing angles $\{\phi_{fe}\}$. Clearly, this will
result in more stationary points contributing to the final amplitude. To suppress their
contributions, we introduce
modifiers $f_f$ which will be described in the next section.
\end{itemize}

\subsubsection{Prescription of the modifiers}\label{modifier-sec}

In order to make \eqref{eq:def_int} complete, we have to describe the function $f_f$.

For every face $f$ we choose three vectors $v_e$ ($e\subset f$) on $\mathbb{R}^2$ with
norms $j_{e}$ such that $\sum_{e\subset f} v_e=0$, i.e. they form a triangle with
edge lengths $j_e$. The vectors are ordered anti-clockwise, their choice is unique up to
Euclidean transformations, i.e rotations and translations.

Let us denote the edges (in cyclic order) by $1,2,3$. The angles (counted
clockwise) between the vectors
$v_k$ and $v_j$ are denoted by $2(\psi_{kj}-\pi)$, where  $2(\psi_{kj} - \pi)$ is the $SO(3)$
angle taking values in $(0,\pi)$ for $(k,j)\in\{(2,1),(3,2),(1,3)\}$.
\begin{figure}[ht]
\begin{center}
\includegraphics[scale=0.5]{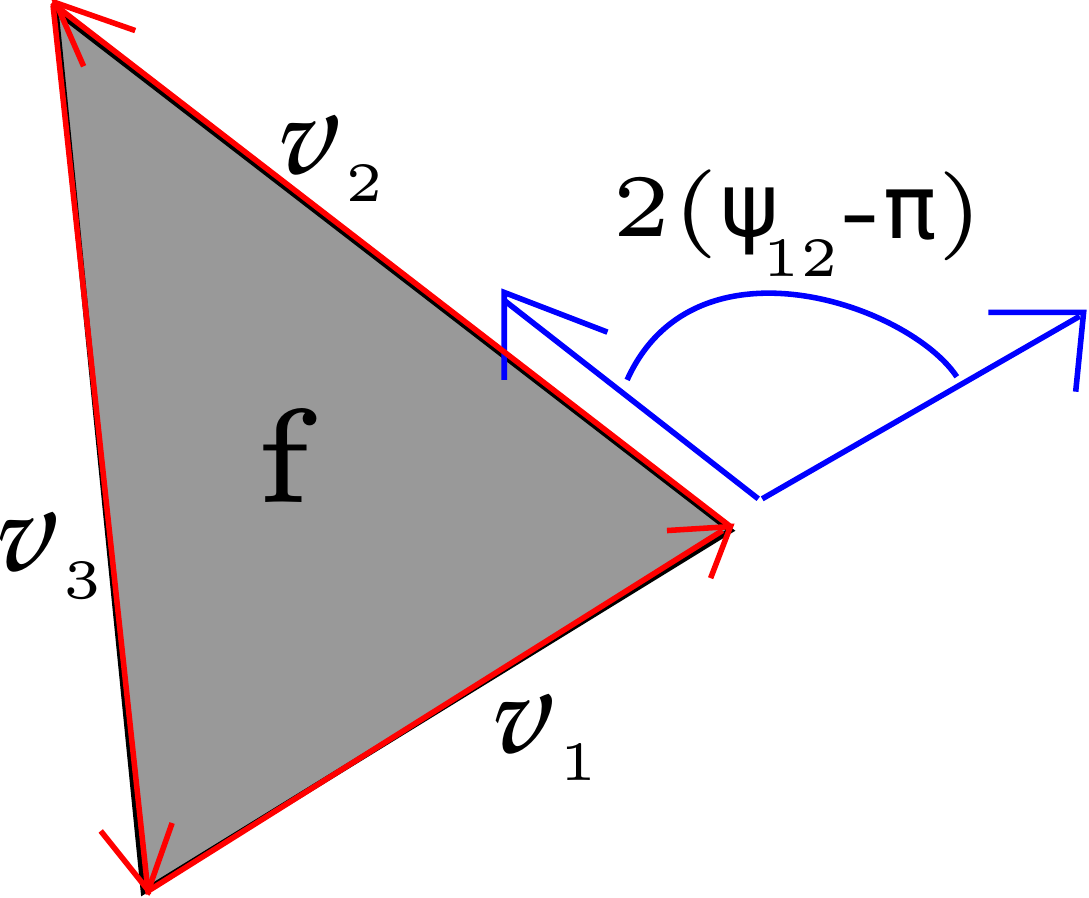}%
\caption{The choice of vectors $v_i$.}
\label{pic-inv}
\end{center}
\end{figure}
Due to the ordering,  the  $SU(2)$ angles $\psi_{21}$,
$\psi_{32}$ and $\psi_{13}$ are positive and smaller than $2\pi$, in fact, one can also
check that
$\psi_{kj}$ is  in $(\pi,2\pi)$. This
choice contributes an overall sign to the invariant,
 to be more precise, there are two different choices of cyclic order giving two invariants that may differ by a sign factor.  This will be discussed in more detail in appendix
\ref{sec:sign_intertwiner}. In particular  we compare them to the intertwiner introduced in
\cite{Penrose,Penrose2}. The angles $\psi_{kj}$ satisfy the relation
\begin{equation}
 \psi_{21}+\psi_{32}+\psi_{13}=4\pi\ .
\end{equation}
We introduce a function $f(x\ \text{mod}\ 2\pi,y\ \text{mod}\ 2\pi)$ such that
\begin{itemize}
 \item it is equal to $1$ in the neighbourhood of $x=\psi_{21}$, $y=\psi_{32}$,
\item it is equal to zero in the neighbourhood of points
\begin{equation}
 (x,y)=\pm(\psi_{21}+\pi,\psi_{32}),
\pm(\psi_{21},\psi_{32}+\pi),\pm(\psi_{21}+\pi,\psi_{32}+\pi),
(-\psi_{21},-\psi_{32}) \quad .
\end{equation}
\end{itemize}
Hence, we define
\begin{equation}
 f_f(\phi_{f1},\phi_{f2},\phi_{f3})=f(\phi_{f2}-\phi_{f1},\phi_{f3}-\phi_{f2}) \quad .
\end{equation}

\subsubsection{The spin network evaluation}
\label{sec:integral}

Given the definition of invariants in \eqref{eq:def_int} it is straightforward to define
the evaluation of a given spin network: The intertwiners are contracted with each other according to
the combinatorics of the network. The resulting amplitude has to be normalized, i.e. divided
by the product of norms of our intertwiners, see section \ref{sec:Theta}. It is, however, not
sufficient in
order to agree with the canonical definition \cite{Penrose,Penrose2}. The remaining sign
ambiguity will be resolved in Appendix \ref{sec:sign}.

As in the standard coherent state approach the amplitude (contraction of intertwiners) then
reads:
\begin{equation} \label{eq:full_amp}
 \begin{split}
 (-1)^s \int\prod_{f\in F}\rd U_f\prod_{e\subset f}\frac{\rd\phi_{fe}}{2\pi}&\prod_f
f_f\left(\{\phi_{fe}\}_{e\subset f}\right)\\
&\underbrace{\prod_{e\in E}
\epsilon\Big(
U_{s(e)}O_{\phi_{s(e)e}}|\coh\rangle\ ,\
U_{t(e)}O_{\phi_{t(e)e}}|\coh\rangle
\Big)^{2j_e}}_{e^S}
 \end{split} \quad ,
\end{equation}
where $s$ is the sign factor as prescribed in \cite{Penrose,Penrose2} (see Appendix
\ref{sec:sign}), and
$\epsilon(\cdot,\cdot)$ is an invariant bilinear form defined by
\begin{equation}
 \epsilon(|\coh\rangle,|\coh \rangle)=\epsilon(|-\coh\rangle,|-\coh \rangle)=0,\quad
\epsilon(|\coh\rangle,|-\coh \rangle)=-\epsilon(|-\coh\rangle,|\coh \rangle)=1\ .
\end{equation}
The choice of the orientation of edges, faces and the sign factor prescription will be described
in appendix
\ref{sec:prescription}.
To perform the stationary point analysis we rewrite (part of) the integral kernel as an
exponential function
and define the `action' $S$. From \eqref{eq:full_amp} one can deduce that
\begin{equation}
 S=\sum_{e}S_e \quad ,
\end{equation}
where the action $S_e$ (labelled by the edge $e$) is given by:
\begin{equation}
 S_e=2j_e\ln
\epsilon\left(U_{s(e)}O_{\phi_{s(e)e}}|\coh\rangle\ ,\
U_{t(e)}O_{\phi_{t(e)e}}|\coh\rangle\right) \quad .
\end{equation}

\subsection{The action}
\label{sec:action}

In order to examine the geometric meaning of the action on its  points of stationary phase, let
us
introduce the following geometric quantities. For each face $f \in F$ we introduce vectors $n_f$ (as
traceless Hermitian matrices, which can be naturally identified with vectors in
$\mathbb{R}^3$) defined by:
\begin{equation}
 n_f=U_fHU_f^{-1} \quad ,
\end{equation}
where
\begin{equation}
 H=\begin{pmatrix}
    0 & i\\ -i & 0
   \end{pmatrix} \quad .
\end{equation}
For each pair $\{f,e\}$ with $e\subset f$, we define vectors $B_{fe}$ (also as traceless
matrices):
\begin{equation}
\begin{split}
B_{fe}&=j_{e}(2U_{f}O_{\phi_{fe}}|\coh\rangle\langle\coh|O_{\phi_{fe
}}^{-1}U_{f}^{-1}-{\mathbb I})\\
&=U_{f}O_{\phi_{fe}}\left[j_e\begin{pmatrix}
                              1& 0\\ 0 & -1
                             \end{pmatrix}
\right]O_{\phi_{fe}}^{-1}U_{f}^{-1} \quad .
\end{split}
\end{equation}
Note that the length of $B_{fe}$ is
equal to $j_{e}$.

We can already deduce that $\Re S\leq 0$. The stationary point analysis contains the
conditions $\partial S=0$
and $\Re S=0$. These are as follows
\begin{itemize}
 \item The reality condition is satisfied if and only if
\begin{equation}
 U_{s(e)}O_{\phi_{s(e)e}}|\coh\rangle\perp
U_{t(e)}O_{\phi_{t(e)e}}|\coh\rangle \quad ,
\end{equation}
where $\perp$ means perpendicular in the $SU(2)$ invariant scalar product. This is equivalent to $B_{s(e)e}=-B_{t(e)e}$.
\item Using both the reality condition and the definition of $B_{fe}$  we obtain from the variation
of $S$ with respect to $U_f$:
\begin{equation}
 X\frac{\partial S_{e}}{\partial U_f}=\left\{\begin{array}{ll}
                                       \Tr X B_{fe} & e\subset f\\ 0 &
					e\not\subset f
                                      \end{array}\right. \quad ,
\end{equation}
where $X$ is a generator of the Lie algebra. Hence the action is stationary with respect to
$U_f$ if:
\begin{equation}
 \sum_{e\subset f} B_{fe}=0 \quad .
\end{equation}
\item Similarly we obtain for the variation of $S$ with respect to $\phi_{fe}$ (again using the
reality
condition):
\begin{equation}
 \frac{\partial S_{e}}{\partial \phi_{fe'}}=\left\{\begin{array}{ll}
                                       \Tr n_f B_{fe} & e=e'\subset f\\ 0 &
                                       \text{otherwise}
                                      \end{array}\right. \quad .
\end{equation}

So the condition from variation with respect to $\phi$ is
\begin{equation}
 \forall_{e\subset f}\ n_f\perp B_{fe} \quad .
\end{equation}
\end{itemize}
\begin{figure}[ht]
\begin{center}
\includegraphics[scale=0.5]{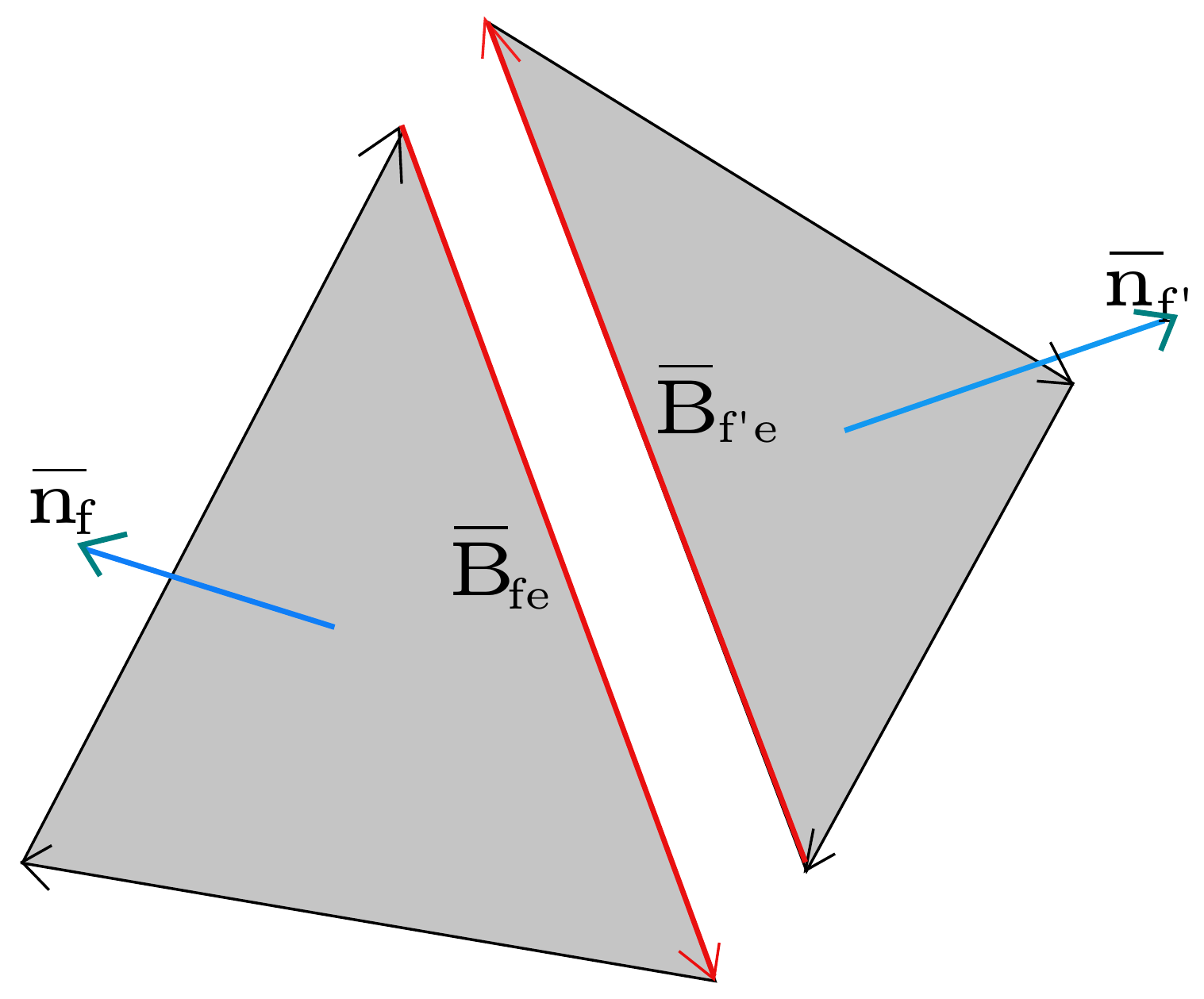}%
\caption{Stationary point condition.}
\label{pic-stationary1}
\end{center}
\end{figure}
Before we discuss the geometric meaning of the just derived conditions, we first have to
examine the
symmetries of the action to determine the amount of stationary points and their relations.

\subsection{Symmetry transformations of the action}
\label{sec:transformations}

There exist several variable transformations that only change $e^{S}$ by a sign
such that a stationary point is transformed into another stationary point.
Some of the transformations below are continuous so  the stationary points form submanifolds of
orbits under the action of  these symmetries.
We will explain the geometric interpretation of these orbits  in section \ref{sec:geom}, and show
that  these orbits are isolated for many spin networks, e.g. the $6j$ symbol.

The above mentioned transformations are as follows:
\begin{itemize}
 \item $u$-symmetry: \newline
 $\forall u\in SU(2)$,
\begin{equation}
 \forall_{f\in F},\ U_f\rightarrow uU_f
\end{equation}
applied to all $U_f$ simultaneously preserves $e^S$. This is the only symmetry which has to be
applied to all group elements simultaneously showing that one of the $SU(2)$
integrations in \eqref{eq:full_amp} is redundant (gauge).
\item $o_f$-rotation: \newline
For a chosen face $f$ and $\phi$,
\begin{equation}
 U_f\rightarrow U_fO_{\phi},\quad \forall_{e\subset f} \phi_{fe}\rightarrow
\phi_{fe}-\phi
\end{equation}
preserves $e^S$, in fact, each $e^{S_e}$ is preserved.
\item $-u_f$-symmetry: \newline
For any chosen face $f\in F$,
\begin{equation}
 U_f\rightarrow (-1)U_f
\end{equation}
preserves $e^S$ because for every face $\sum_{e\subset f} j_{e}$ is an integer.
\item $r_f$-reversal transformation: \newline
For any chosen face $f$,
\begin{equation}\label{D19}
 U_f\rightarrow U_f\underbrace{\begin{pmatrix}
                    i &0\\ 0& -i
                   \end{pmatrix}}_{D},\quad \forall_{e\subset f}\
\phi_{fe}\rightarrow -\phi_{fe}\ .
\end{equation}
Because
\begin{equation}
 D^{-1}O_{\phi}D=O_{-\phi},\quad D|\coh\rangle=i|\coh\rangle\ ,
\end{equation}
$e^S$ is multiplied by
\begin{equation}
i^{2\sum_{e\subset f}j_e}=(-1)^{\sum_{e\subset f}j_e} \quad .
\end{equation}
Let us notice that
$2j_{e}\in{\mathbb Z}$ and $\sum_{e\subset
f}j_e$ is an integer,
\item $-o_{fe}$ transformation: \newline
For any chosen pair $e\subset f$
\begin{equation}
 \phi_{fe}\rightarrow \phi_{fe}+\pi
\end{equation}
This multiplies the integrated term by $(-1)^{2j_e}$.
\end{itemize}

Note that the transformations $o_f$, $-u_f$ $-o_{fe}$, $r_f$ are restricted to variables
associated to one face. They transform the functions $f_f$ as follows:
\begin{itemize}
 \item $o_f$ shifts all angles $\phi_{f e} $ on $f$ by an angle $\phi$:
\begin{equation}
f_f'(\{\phi_{fe}\})=f_f(\{\phi_{fe}+\phi\})=f_f(\{\phi_{fe}\}) \quad ,
\end{equation}
since $f_f$ only depends on differences of angles.
\item $-o_{fe}$
\begin{equation}
 f_f'(\{\phi_{fe'}\})=f_f(\{\phi_{fe'}+\delta_{ee'}\pi\}) \quad .
\end{equation}
\item $r_f$
\begin{equation}
 f_f'(\{\phi_{fe}\})=f_f(\{-\phi_{fe}\}) \quad .
\end{equation}
\end{itemize}
To sum up, the functions $f_f$ are preserved by  $u$-, $-u_f$- and $o_f$-transformations, since the
first two do not affect the angles $\phi$ and the last one translates all angles by a constant.

In addition to that, let us also define an additional transformation $c$,
which we call parity transformation:
\begin{equation}
 \forall f: \quad U_f\rightarrow U_f \begin{pmatrix}
                     0 & 1\\ -1 &0
                    \end{pmatrix} \quad .
\end{equation}
It transforms the integral into its complex conjugate due to the fact that
\begin{equation}
 \bar{U}=\begin{pmatrix}
                     0 & 1\\ -1 &0
                    \end{pmatrix}^{-1}U\begin{pmatrix}
                     0 & 1\\ -1 &0
                    \end{pmatrix} \quad ,
\end{equation}
and the $f_f$, the matrix $O_\phi$ and the vectors $|\pm\coh\rangle$ are real.

In the next section we will examine which group is generated by the transformations, i.e. the
symmetry group of the action.

\subsection{Groups generated by symmetry transformations}

The transformations described in \ref{sec:transformations} generate a group $\tilde{G}$ with the following relations:
\begin{equation}
\begin{split}
&u_{(-1)}=\prod_f (-u_f),\\
&\forall_f,\quad r_f^2=(-u_f),\quad (-u_f)^2=1,\quad o_f(2\pi)=1,\\
&\forall_{e\subset f},\quad (-o_{fe})^2=1,\\
&\forall_f,\quad o_f(\pi)\prod_{e\subset f} \ (-o_{fe})=1\ .
\end{split}
\end{equation}
and all its elements commute besides $u$ (that form $SU(2)$) and
\begin{equation}
 \forall f,\ r_fo_f(\alpha)r_f^{-1}=o_f(-\alpha)\ .
\end{equation}
The group generated by all transformations except $u$ is denoted by $G$.

In $\tilde{G}$ (resp. $G$ ), there is a normal subgroup generated by the transformations $u$, $o_f$,
$-u_f$ (resp. $o_f$, $-u_f$), which preserves the modifiers $f_f$. We denote these subgroups
by $\tilde{H}$
(and $H$ respectively); their quotient groups are given by
\begin{equation}
 K=\tilde{G}/\tilde{H}=G/H \quad.
\end{equation}
This is an Abelian group generated by
\begin{equation}
 \forall_{e\subset f}\quad [r_f],\ [-o_{fe}]
\end{equation}
with relations
\begin{equation}
 \forall_f\ \prod_{e\subset f} [-o_{fe}]=1,\ [r_f]^2=[-o_{fe}]^2=1 \quad ,
\end{equation}
which show that $K$ is isomorphic to $\mathbb{Z}_2^{3|F|}$.

In the next two sections, we will discuss the geometric interpretation of the points of
stationary phase.

\subsection{Geometric lemma}
\label{sec:geom}

Our goal in this section is to describe the geometric interpretation of the stationary point
orbits introduced
in  section \ref{sec:action}. In particular, we will show how these points are related to the
standard
stationary point interpretation in the coherent state method.

\begin{lm}
For every set of vectors $B_{fe}$ of length $j_{e}$ satisfying
\begin{equation}
\label{Bcond}
\begin{split}
 &\forall_{e} B_{s(e)e}=-B_{t(e)e} \quad ,\\
&\forall_f\sum_{e\subset f} B_{fe}=0 \quad ,
\end{split}
\end{equation}
there exist $\phi_{fe}$ and $U_f$ being a point of stationary phase with vectors
$B_{fe}$. Moreover,
all these points are related via $G$
transformations.
\end{lm}

\begin{proof}
For every $f$ we can choose the unit vector $n_f$ perpendicular to all $B_{fe}$
(for all $e\subset f$). Such a normal is only determined up to a sign.
Let us choose $U_f$ such that
\begin{equation}
 n_f=U_fHU_f^{-1} \quad .
\end{equation}
Such a choice always exists, but it is not unique. $U_f$ is only determined up
to the transformation
\begin{equation}
U_f\rightarrow U_fDO_\phi
\end{equation}
since $D$, defined in \eqref{D19}, stabilizes $H$ up to a sign:
\begin{equation}
 D H D^{-1} = -H \quad .
\end{equation}
This is called the $D_\infty$ group.

The vectors $U_f^{-1}B_{fe}U_f$ are orthogonal to $H$. The
operators $U_f^{-1}B_{fe}U_f$ are thus real and we can choose their
 eigenvectors with positive eigenvalues as
\begin{equation}
 \binom{\cos\phi_{fe}}{\sin\phi_{fe}} \quad .
\end{equation}
Hence $\phi_{fe}$ is fixed (up to $\pi$).

It is straightforward to check that this construction gives a stationary point, in fact, each
stationary point with vectors $B_{fe}$ must be constructed in this way.
The ambiguities in the choices above are all related by $o_f$-, $-u_f$-, $-o_f$- and
$r_f$-transformations, i.e. $G$-transformations.
\end{proof}

For every face $f$ on the stationary point $\sum_{e\subset f} B_{fe}=0$ and
given the definition of $f_f$ in section \ref{modifier-sec}, there is a unique choice (up to
$o_f$ transformations) of the stationary point angles $\phi_{fe}$ such that $f_f$ is nonzero.
In the neighbourhood of those stationary point $f_f=1$, whereas around all remaining  ones at
least one of the functions $f_f$ is zero:
\begin{lm}
 For given vectors $B_{fe}$ satisfying \eqref{Bcond}, there exists only one orbit (orbit of
the action of the group $\tilde{H}$)
of stationary points of the action, such that
\begin{equation}
 \prod f_f(\{\phi_{fe}\})\not=0 \quad ,
\end{equation}
and in the neighbourhood of this orbit
\begin{equation}
 \prod f_f(\{\phi_{fe}\})=1 \quad .
\end{equation}
\end{lm}
Note that the normals to the faces change sign under $r_f$ transformations:
\begin{equation}
 n_f\rightarrow -n_f \quad .
\end{equation}
Under the $c$-transformation, the $B_{fe}$ are inverted, but the normals to the faces are not
affected, i.e. they behave as pseudovectors.
\begin{equation}
 \begin{split}
  &B_{fe}\rightarrow -B_{fe} \quad , \\
  &n_f\rightarrow n_f \quad .
 \end{split}
\end{equation}

In the next section we will specify the definition of the normals $n_f$.

\subsection{Normal vectors to the faces}
\label{sec:normals}

We will now give a precise geometric definition of $n_f$ (normal to the face). To
simplify
notation we will omit the subscript $\phi$ in $O_{\phi_{fe}}$. Note that
\begin{equation}
 n_f=U_fO_{fe}HO_{fe}^{-1}U_f^{-1}
\end{equation}
for any edge $e\subset f$, since $O_{fe}$ and $H$ commute.

Take two consecutive edges $e_1,e_2\subset f$ and their respective edge vectors $B_{fe_i}$:
\begin{align}
 B_{fe_1}&=U_{f} O_{fe_1} \left[j_{e_1}\begin{pmatrix}
                              1& 0\\ 0 & -1
                             \end{pmatrix}
\right]O_{fe_1}^{-1} U_{f}^{-1} \quad ,\\
B_{fe_2}&=U_{f} O_{fe_1}(O_{fe_1}^{-1}O_{fe_2})\left[j_{e_2}\begin{pmatrix}
                              1& 0\\ 0 & -1
                             \end{pmatrix}
\right](O_{fe_1}^{-1}O_{fe_2})^{-1} O_{fe_1}^{-1} U_{f}^{-1} \quad .
\end{align}
Rotating all three vectors by $U_f O_{fe_1}$ one obtains (the rotated vectors are denoted by
$B_{fe_i}'$, $n_f'$):
\begin{equation}
 n_f'=H,\quad B_{fe_1}'=j_{e_1}\begin{pmatrix}
                              1& 0\\ 0 & -1
                             \end{pmatrix},\quad
B_{fe_2}=O_{fe_1}^{-1}O_{fe_2}\left[j_{e_2}\begin{pmatrix}
                              1& 0\\ 0 & -1
                             \end{pmatrix}
\right](O_{fe_1}^{-1}O_{fe_2})^{-1} \quad .
\end{equation}
For a stationary point with non-vanishing modifier $f_f$, $O_{fe_1}^{-1}O_{fe_2}$ describes the
rotation by
the $SO(3)$ angle $0<2(\psi_{12}-\pi)<\pi$. We thus conclude:
\begin{equation} \label{eq:sgn_n}
 n_f\cdot(B_{fe_1}\times B_{fe_2})=n_f'\cdot(B_{fe_1}'\times B_{fe_2}')>0 \quad ,
\end{equation}
where we regard $n_f$ and $B_{fe_i}$ as vectors using the natural identification of hermitian matrices
with ${\mathbb R}^3$ (tracial scalar product). Condition \eqref{eq:sgn_n} fixes the sign of
$n_f$ and
also completes the geometric interpretation of the points of stationary phase.

\subsection{Interpretation of planar (spherical) spin-networks as polyhedra}
\label{sec:surface}

In the last section we obtained an interpretation of the stationary points in terms of
a set of vectors
$B_{fe}$ satisfying closure conditions for every face $f$
\begin{equation}
 \sum_e B_{fe}=0 \quad .
\end{equation}
However, these conditions do not specify a unique reconstruction of the according surface dual
to the spin network. In fact, already each triangle allows for two different configurations of
$B_{fe}$ vectors. Therefore, we will here describe a method to reconstruct the surface from $B_{fe}$
vectors for the spherical case:

Let us draw the graph on the sphere (on the plane) as described in appendix
\ref{sec:prescription}. From the possible ways of drawing it, which in the case of $2$-edge
irreducible spin networks is in one-to-one correspondence with the orientation of the spin network,
we have to choose
one. In the case of $2$-edge irreducible graphs
the polyhedra obtained from different choices only differ by orientation. In addition to nodes and edges, there
is also a natural notion of two-cells.
The latter are defined as areas bounded by loops of edges. We are mainly interested in the dual
picture that in this case is a triangulation of the sphere. Thus there is a unique identification
 of the vertices in the dual picture. A cyclic ordering of the edges for each $f$ is
inherited from the orientation of the sphere.

In the following, we will construct an immersion (not an embedding) of this triangulation of the
sphere into $\mathbb{R}^3$, such
that every edge $e$ is given by $B_{t(e)e}$ (with the right orientation).

Let us choose one vertex $v_0$. Every other vertex $v'$ can be connected to $v_0$ by a path
\begin{equation}
 v_0,e_0, v_1,e_1,\ldots, v'\ .
\end{equation}
Every edge $e_i$ in the sequence belongs to two faces. Exactly one of these faces is such that
$v_i,v_{i+1}$ are the consecutive vertices w.r.t the cyclic order of the face. We denote this
face by
$f_i$ (see figure \ref{pic-surface-built2}). We introduce the vector
\begin{equation}
 \tilde{v}'=\sum_i B_{f_ie_i} \quad .
\end{equation}
One can prove that this vector does not depend on the chosen path. To see this, let us
consider a basic
move that consists of replacing $v_i,e_i,v_{i+1}$ by $v_i,e,v,e',v_{i+1}$ where all three
vertices belong to the same face $f$. Using the property
\begin{equation}
 \sum_{e\subset f} B_{fe}=0
\end{equation}
and the proper orientation, one can show that the vector $\tilde{v}'$ is invariant with respect to
this move. In fact any two paths
can be transformed into one another by a sequence of these basic moves (or their inverses)
due to the fact that the graph is spherical. A different choice of $v_0$ gives a
translated surface.
\begin{figure}[ht]
\begin{center}
\includegraphics[scale=0.5]{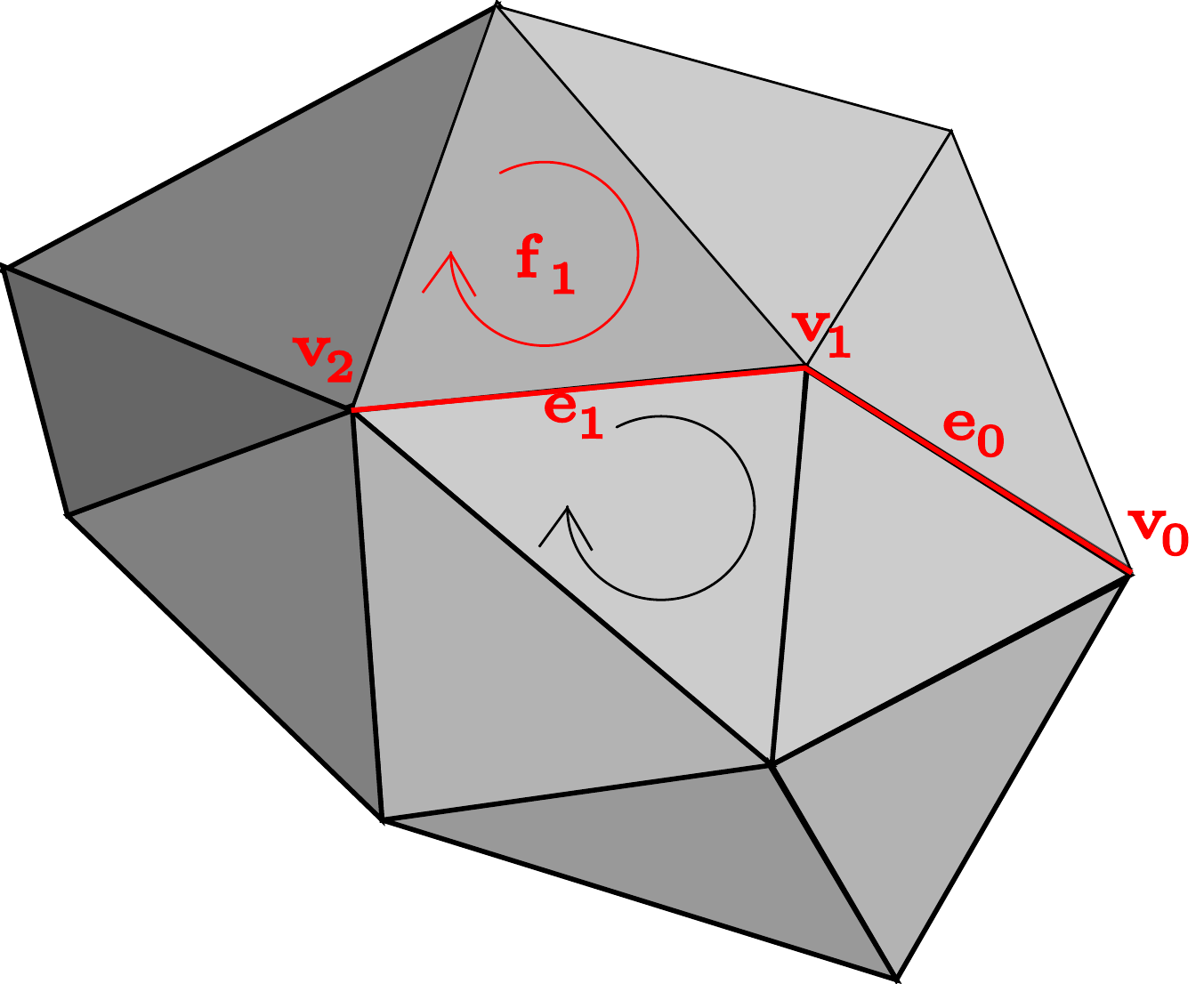}%
\caption{Reconstruction of the surface}
\label{pic-surface-built2}
\end{center}
\end{figure}
It is straightforward to check that
\begin{equation}
 \tilde{v_b}-\tilde{v}_a=B_{fe} \quad ,
\end{equation}
where $v_a$ and $v_b$ are vertices joint by the edge $e$ and $f$ is the face such that
$(v_a,v_b)$ is the pair of
consecutive vertices in the cyclic order of $f$.

Let us notice that from three vectors $B_{fe}$ satisfying the closure condition one can form
a triangle in two ways (see figure \ref{pic-clebsch-orient}), but only that one
depicted on the left appears in the
reconstruction discussed here.
\begin{figure}[ht]
\begin{center}
\includegraphics[scale=0.5]{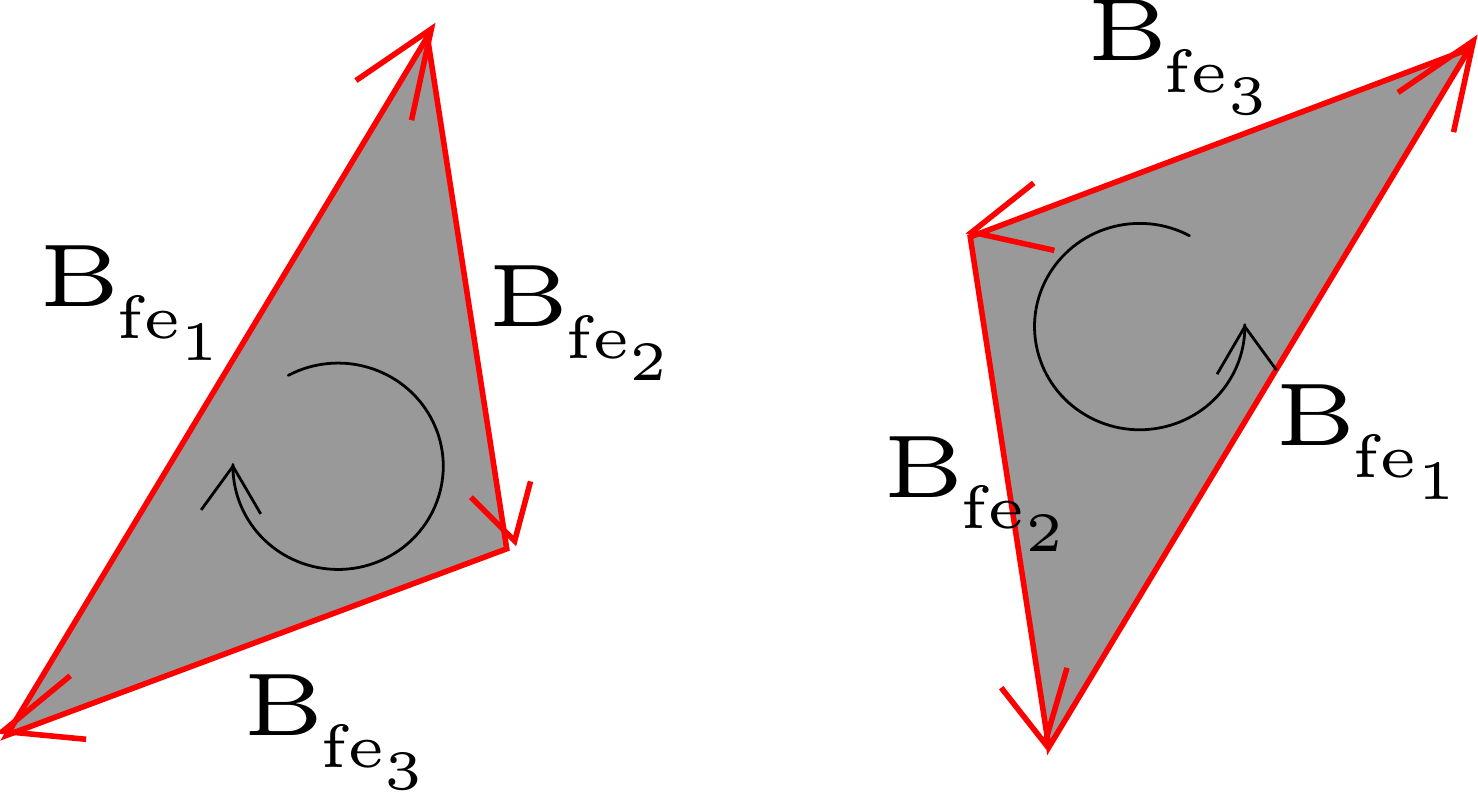}%
\caption{Two possible triangles formed by $B_{fe}$ satisfying closure condition.}
\label{pic-clebsch-orient}
\end{center}
\end{figure}
Moreover the direction of the normal to the face coincides with the orientation
inherited
by the face from the cyclic order of its edges.

For non-planar graphs, in general, we can only reconstruct the universal cover of the surface.

Before we continue with the stationary point analysis for the angles $\phi_{fe}$ in the next
section, let us briefly summarize the results of section \ref{sec:construction_invariants}: We have
introduced a class of modified coherent states for irreducible representations of $SU(2)$, which
contain an additional smearing parameter, and presented how to construct invariants from them. From
the contraction of these invariants (according to the spin network) an effective action has been
derived, whose points of stationary phase allow for the same geometrical interpretation as the
standard coherent states \cite{Perelmov,Frank3D}. The amount of stationary points is significantly
increased by the smearing parameters, yet they are all related by symmetry transformations of the
action; a certain set of them can be suppressed by the prescribed modifiers. Eventually,
we have depicted a way to reconstruct a triangulation from planar spin networks.

\section{Variable transformation and final form of the integral}
\label{sec:evaluation-final}

In this section we focus on the stationary point analysis with respect to the angles
$\phi_{fe}$, which is the key modification in comparison to previously used coherent state
approaches, see also section \ref{sec:inv-def}. This analysis allows us to obtain an
effective action for $S_e$ associated to the edge $e$ in terms of a phase, which we will
identify as
the angle
between the normals of the faces sharing the edge $e$. Furthermore we are able to expand
the effective action for $S_e$ in orders of $\frac{1}{j}$ and initiate the discussion
of
next-to-leading order
contributions.

\subsection{Partial integration over $\phi$ and the new action}

Suppose that we have a non-degenerate configuration, i.e.
\begin{equation}
 \forall_e n_{s(e)}\cdot n_{t(e)}\not=\pm 1 \quad .
\end{equation}
Then the partial stationary point analysis with respect to all $\phi_{fe}$ can be performed.
Its result will be the sum over the contribution from all stationary points with
respect to $\phi_{fe}$ for a given configuration of $B_{fe}$ vectors, but for fixed $U_f$ (so also
fixed $n_f$).

\subsubsection{Stationary points for $S_e$}
\label{sec:stat-phi}

In this section we will explain the contribution to the integral from the stationary point
of the action $S_e$ with respect
to $\phi_{s(e)e},\phi_{t(e)e}$. The $f_{s(e)},f_{t(e)}$
terms can be ignored, since they are equal to $1$ around the
stationary point.

We can separately consider terms corresponding to each edge
\begin{equation}
\frac{1}{4 \pi^2} \int d \phi_{s(e)e} d\phi_{t(e) e}  \;
\epsilon\left(U_{s(e)}O_{s(e)e}|\coh\rangle\ ,\
U_{t(e)}O_{t(e)e}|\coh\rangle\right)^{2j_e} \quad ,
\end{equation}
and perform the
stationary point analysis that gives the asymptotic result of the integration over
$\phi_{s(e)e},\phi_{t(e)e}$.
The stationary point with respect to $\phi_{t(e)e}$ and $\phi_{s(e)e}$ is given by the conditions
\begin{equation}
 U_{s(e)}O_{s(e)e}|\coh\rangle\perp U_{t(e)}O_{t(e)e}|\coh\rangle \quad ,
\end{equation}
which is equivalent to
\begin{equation}\label{eq:normal-form}
 U=O_{s(e)e}^{-1}U_{s(e)}^{-1}U_{t(e)}O_{t(e)e}= (-1)^{\tilde{s}} e^{-i\tilde{\theta}
\begin{pmatrix}
                                                            1 &0\\ 0&-1
                                                           \end{pmatrix}}
\begin{pmatrix}
 0 & -1\\ 1 & 0
\end{pmatrix} \quad ,
\end{equation}
where $\tilde{\theta}\in \left(-\frac{\pi}{2},\frac{\pi}{2}\right)$ and
$\tilde{s}\in \{0,1\}$ are uniquely determined by this equation.
In section \ref{sec:angle-inter} we will show that $2\tilde{\theta}$ can be interpreted as the
angle enclosed by the normal vectors $n_{s(e)}$ and $n_{t(e)}$ (w.r.t. the axis $B_{t(e)e}$). Hence,
$S_e$ on the stationary point
is of the following form:
\begin{equation}
 S_e=2j_{e} \ln \epsilon(\cdots)
=2j_{e}\tilde{\theta}+ i\ 2j_e\pi \tilde{s} \quad .
\end{equation}
As already discussed in section \ref{sec:geom}, each stationary point is characterized by the
existence of $B_{t(e)e}=-B_{s(e)e}$ orthogonal to both $n_{s(e)}$ and $n_{t(e)}$ (see also
stationary point conditions in section \ref{sec:action}). There exist two such configurations
that differ by a
sign of $B_{s(e)e}$.

For every configuration one has $4$ stationary points that can be obtained from one
another by $-o_{s(e)}$- and $-o_{t(e)}$-transformations. In case $j_e$ is an integer the
contributions from the two stationary points are equal, see also section \ref{sec:transformations}.

Contributions from $B_{fe}$ configurations with opposite signs are related by complex
conjugation.

\subsubsection{Geometric interpretation of the angle $\tilde{\theta}$}
\label{sec:angle-inter}

The missing piece of the description above is the exact value of the angle $\tilde{\theta}$.
Here we will provide a geometric interpretation
of this angle and its relation to the angle between faces. Let us recall:
\begin{align}
 B_{s(e)e}&=j_{e}U_{s(e)}O_{s(e)e}\begin{pmatrix}
                                                            1 &0\\ 0&-1
                                                           \end{pmatrix}
O_{s(e)e}^{-1}U_{s(e)}^{-1}\\
n_{s(e)}&=U_{s(e)}O_{s(e)e}HO_{s(e)e}^{-1}U_{s(e)}^{-1}\\
n_{t(e)}&=U_{t(e)}O_{t(e)e}HO_{t(e)e}^{-1}U_{t(e)}^{-1}
=e^{i\tilde{\theta} \frac{B_{t(e)e}}{|B_{t(e)e}|}}n_{s(e)}e^{-i\tilde{\theta}
\frac{B_{t(e)e}}{|B_{t(e)e}|}}
\end{align}
The angle $2\tilde{\theta}$ is the angle by which one needs to rotate $n_{s(e)}$ around the
axis $B_{t(e)e}$ to obtain $n_{t(e)}$.
We will denote this $SO(3)$ angle by
\begin{equation}
 \theta=2\tilde{\theta},\quad \theta \in (-\pi,\pi)\ .
\end{equation}
This remaining ambiguity of the sign factor $\tilde{s}$ will be resolved in appendix
\ref{sec:sign}.

\subsection{Partial integration over $\phi$}
\label{d-sec}

We introduce new variables
\begin{equation}
 \phi_1=\phi_{s(e)}-\phi_{s(e)}^0,\quad \phi_2=\phi_{t(e)}-\phi_{t(e)}^0\ ,
\end{equation}
where $\phi_{s(e)}^0$ and $\phi_{t(e)}^0$ denote the stationary points. Then using
\eqref{eq:normal-form}, we can write the action as:
\begin{equation} \label{eq:start}
\frac{1}{4 \pi^2} \int \rd\phi_1\rd \phi_2
(-1)^{2j_e\tilde{s}}
\left(e^{i \tilde{\theta}}
\cos \phi_1 \cos \phi_2 + e^{-i \tilde{\theta}} \sin \phi_1 \sin \phi_2
\right)^{2 j} \quad ,
\end{equation}
where we integrate over $\phi_i$. By splitting the terms in the bracket in real and imaginary part, we obtain:
\begin{equation}
\cos \tilde{\theta} \underbrace{\left(\cos \phi_1 \cos \phi_2 + \sin \phi_1 \sin
\phi_2
\right)}_{\cos(\phi_1 - \phi_2)} + i \sin \tilde{\theta} \underbrace{\left(\cos
\phi_1
\cos \phi_2 - \sin \phi_1 \sin \phi_2 \right)}_{\cos (\phi_1 + \phi_2)} \quad .
\end{equation}
We define new variables
\begin{equation}
\alpha := \phi_1 - \phi_2, \quad \beta:= \phi_1 + \phi_2 \quad,
\end{equation}
and the Jacobian for this transformation is given by:
\begin{equation}
\left| \frac{\partial \alpha \partial \beta}{\partial \phi_1 \partial \phi_2}
\right| =\left|
\begin{array}{c c}
1 & -1 \\
1 & 1
\end{array}
\right| = 2 \quad .
\end{equation}
Hence equation \eqref{eq:start} becomes:
\begin{align} \label{eq:edge_action}
& \frac{1}{8 \pi^2} {\int d\alpha d \beta\ } (-1)^{2j_e\tilde{s}}\left(\cos
\tilde{\theta} \cos \alpha + i \sin \tilde{\theta} \cos \beta \right)^{2 j_e} =
\nonumber \\
=&  \frac{1}{8 \pi^2} {\int d\alpha  d \beta\ } (-1)^{2j_e\tilde{s}}\exp\left\{
\underbrace{2j_e
\ln \left(\cos \tilde{\theta} \cos \alpha + i \sin \tilde{\theta} \cos \beta
\right)}_{=: S_e'}\right\} \quad ,
\end{align}
where $S_e=S_e'+i\ 2j_e\pi \tilde{s}$.

\subsubsection{Expansion around stationary points}

Given the definitions from the previous section, we compute the expansion of the following
expression:
\begin{equation}
\frac{1}{8\pi^2} \int\rd \alpha\rd\beta\ (-1)^{2j\tilde{s}} e^{S_e'} \quad .
\end{equation}
The stationary point is given by $\alpha=\beta=0$, which
corresponds to $\phi_i=0$, i.e. $\phi_{s(e)}=\phi_{s(e)}^0$,
$\phi_{t(e)}=\phi_{t(e)}^0$. In this point the action associated to the edge $e$ becomes:
\begin{equation}
 S_e'=2j_e\ln \left(e^{i\tilde{\theta}_e}\right)=i2j_e\tilde{\theta}_e
\end{equation}
In order to compute the first order contribution, one has to consider the matrix of second
derivatives (evaluated on the point of stationary phase):
\begin{align}
\frac{\partial^2 S_e'}{\partial \alpha^2} =& - 2j_e \frac{\cos \tilde{\theta} \cos
\alpha}{\cos \tilde{\theta} \cos \alpha + i \sin \tilde{\theta} \cos \beta} \quad , \\
\frac{\partial^2 S_e'}{\partial \alpha \partial \beta} =& 0 = \frac{\partial^2
S_e'}{\partial \beta \partial \alpha}  \quad ,\\
\frac{\partial^2 S_e'}{\partial \beta^2} = & - 2 i j_e \frac{\sin \tilde{\theta} \cos
\beta}{\cos \tilde{\theta} \cos \alpha + i \sin \tilde{\theta} \cos \beta} \quad .
\end{align}
Around the stationary point the action can be expanded (up to second order in the variables
$\alpha, \beta$):
\begin{equation} \label{eq:edge_expansions_stationary}
S_e'= i 2j_e  \tilde{\theta} +\frac{1}{2}\begin{pmatrix}
    \alpha & \beta
   \end{pmatrix}
\begin{pmatrix}
-2j_e \cos \tilde{\theta} e^{-i  \tilde{\theta}} & 0 \\
0 & -2 i j_e \sin \tilde{\theta} e^{-i \tilde{\theta}}
\end{pmatrix}
\begin{pmatrix}
\alpha \\
\beta
\end{pmatrix}
+\cdots \quad .
\end{equation}
In order to correctly perform the stationary phase approximation, it is indispensable to state
the right branch of the square root, here for $\tilde{\theta}\in
\left(-\frac{\pi}{2}, \frac{\pi}{2}\right)$:
\begin{equation}
\begin{split}
 \sqrt{\cos\tilde{\theta}e^{-i\tilde{\theta}}}&=\sqrt{|\cos\tilde{\theta}|}
e^{-i\frac{1}{2}\tilde{\theta}}\\
\sqrt{i\sin\tilde{\theta}e^{-i\tilde{\theta}}}&=\sqrt{|\sin\tilde{\theta}|}
e^{-i\frac{1}{2}\tilde{\theta}}
\left\{\begin{array}{ll}
        e^{{ - i\frac{\pi}{4}}} & \tilde{\theta}\in \left(-\frac{\pi}{2}, 0\right)\\
        e^{{i\frac{\pi}{4}}} & \tilde{\theta}\in \left(0, \frac{\pi}{2}\right)
       \end{array}\right. \quad .
\end{split}
\end{equation}
Let us notice that
\begin{equation}
 \sign\sin\theta=\sign\sin\tilde{\theta}
\quad \text{for}\quad \tilde{\theta}\in
\left(-\frac{\pi}{2}, \frac{\pi}{2}\right) \quad .
\end{equation}
Hence, the leading order contribution from the stationary point is:
\begin{equation}\label{eq:Legendre-preliminary}
\frac{1}{8 \pi^2} \frac{2 \pi (-1)^{2j_e\tilde{s}}}{\sqrt{2 j^2 \left| \sin 2
\tilde{\theta} \right|}}
e^{i 2 \tilde{\theta} \left(j + \frac{1}{2}\right) -i \frac{\pi}{4}
\text{sign}(\sin 2
\tilde{\theta})}
\left(1+O\left(\frac{1}{j}\right)\right) \quad .
\end{equation}
In the next section we will show an improvement of  this result.

\subsubsection{The total expansion of the edge integral}
\label{sec:total-edge}

Let us introduce a number (see appendix \ref{rel-islas} for a motivation of its origin)
\begin{equation}\label{eq:C_j-first}
 C_{j}=\frac{1}{4^j}\frac{\Gamma(2j+1)}{\Gamma(j+1)^2} \quad .
\end{equation}
We can multiply \eqref{eq:Legendre-preliminary} by $\frac{C_j}{C_j}=1$ and use the expansion
$\frac{1}{C_j} = \sqrt{\pi j} \left(1 + O(\frac{1}{j}) \right)$
derived in appendix \ref{sec:C_j} to write the result as
\begin{equation}\label{eq:first-exp}
C_{j_e} \frac{(-1)^{s_e}}{4 \sqrt{2 \pi j_e \left| \sin  {\theta} \right|}} e^{i(
\theta \left(j_e +
\frac{1}{2}\right) -  \frac{\pi}{4} \text{sign}(\sin {\theta}))}
\left(1+O\left(\frac{1}{j_e}\right)\right) \quad.
\end{equation}
By $s_e$ we denoted the sign factor
\begin{equation}
 s_e=\left\{\begin{array}{ll}
           0 & j_e\ \text{integer},\\
           0 & j_e\ \text{half-integer and}\ \tilde{s}=0\\
           1 & j_e\ \text{half-integer and}\ \tilde{s}=1\quad .
          \end{array}\right.
\end{equation}
We will determine the sign $s_e$ in appendix \ref{sec:sign}.

We introduce new `length' parameters
\begin{equation}
 l_e:=j_e+\frac{1}{2} \quad ,
\end{equation}
and using the fact that $\frac{(-1)^{s_e}}{4 \sqrt{2 \pi l_e \left| \sin  {\theta} \right|}}=
\frac{(-1)^{s_e}}{4 \sqrt{2 \pi j_e \left| \sin  {\theta} \right|}}\left(1+O(j^{-1})\right)$
we can express \eqref{eq:first-exp} in terms of $l_e$.
Before we move on, we would like to present a first glimpse at the next-to-leading order
contribution: As it will be shown in section \ref{sec:total-phi} by application of the stationary
point
analysis \eqref{eq:first-exp} and the recursion relation \eqref{eq:rec-rel}, the contribution
(including next-to-leading order (NLO)) from the
integral of $e^{S_e}$ over $\phi_{s(e)e},\phi_{t(e)e}$  is given by
\begin{equation}
 C_{j_e}\frac{(-1)^{s_e}}{4 \sqrt{2 \pi l_e \left| \sin \theta \right|}}
e^{i(  l_e\theta  -  \frac{\pi}{4} \text{sign}(\sin
\theta)-\frac{1}{8l_e}\cot\theta)}
\left(1+O\left(\frac{1}{l_e^2}\right)\right) \quad ,
\end{equation}
where $\theta\in(-\pi,\pi)$ is the angle by which one has to rotate $n_{s(e)}$ around
$B_{t(e)e}$
to obtain $n_{t(e)}$.

\subsection{New form of the action}
\label{sec:new-form-action1}

In the previous sections we have computed the contribution of one point of stationary phase
with respect to
the angles $\phi_{fe}$. From section \ref{d-sec} we can also conclude that having one stationary
point all others are obtained
by application of transformations from $\tilde{G}$ that keep $U_f$ fixed. These are given by
compositions of
\begin{equation}
 (-u_f)o_{f}(\pi),\ -o_{fe}\ \forall f\quad .
\end{equation}
However, only the orbit generated by the group of $(-u_f)o_{f}(\pi)$ from a non-trivial
stationary
point contributes, since all other stationary points are suppressed by the modifiers $f_f$.
Therefore it is sufficient to compute the number of these stationary points. The group
generated
by
$(-u_f)o_{f}(\pi)$ is equal to $Z_2^{|F|}$ and acts freely on the stationary points; the
countability of the orbit is thus $2^{|F|}$.

Around the stationary orbit, the integral is hence of the form:
\begin{equation}
(-1)^s 2^{|F|}\int \prod dU_f \prod_{e}(-1)^{s_e}C_{j_e}
\frac{1}{4\sqrt{2\pi\left(j_e +\frac{1}{2}\right)|\sin\theta_e|}}
e^{-i\frac{\pi}{4}\sign\sin {\theta_e}}\
e^{i\left(j_e+\frac{1}{2}\right)\theta_{e} - \frac{1}{8 (j_e + \frac{1}{2})} \cot \theta_e} \quad ,
\end{equation}
where $\theta_{e}$ is the angle between $n_{s(e)}$ and $n_{t(e)}$ with the sign determined by left
hand rule with respect to $B_{t(e)e}$. In the neighbourhood of the stationary point this definition is
meaningful.
The value of the product $\prod_e (-1)^{s_e}$ is discussed in appendix \ref{sec:sign1}. We use
new
`length' parameters introduced in section \ref{sec:total-edge}
\begin{equation}
 l_e:=j_e+\frac{1}{2}
\end{equation}
and perform a change of variables
\begin{equation}
U_f \rightarrow n_f \quad ,
\end{equation}
which is worked out in appendix \ref{variables-n}. The correct integral measure is given by:
\begin{equation}
\mu = \frac{1}{2 \pi} \delta(|n|^2 -1)\, dn_1 \, dn_2 \, dn_3 \quad .
\end{equation}
Thus, we can write the integral (integrating out $o_f$ and $-u_f$ gauges) as:
\begin{equation}
\begin{split}
 &\frac{(-1)^s(-1)^{\sum_e s_e}e^{-i\frac{\pi}{4}\sum_e
\sign_e}\prod_e C_{j_e}}{2^{\frac{5}{2}|E|}\pi^{|F|+\frac{1}{2}|E|}}\\
&\int\prod_{f\in F}
\delta(|n_f|^2-1)\ d^3n_f\frac{1}{\sqrt{\prod_{e}l_e|\sin\theta_e|}}
e^{i\sum_e\left(l_e\theta_e +S^{l_e}_1(\theta_e) \right)} \quad .
\end{split}
\end{equation}
where from \ref{d-sec} we know
\begin{equation}
S^{l_e}_1(\theta_e)= -\frac{1}{8l_e}\cot\theta_e+\ldots \quad .
\end{equation}
The only present symmetry that has to be discussed is a $u$-symmetry, which is implemented by
$SO(3)$ rotations:
\begin{equation}
 n_f\rightarrow un_fu^{-1}
\end{equation}
If the configurations of the vectors $B_{fe}$ is rigid then the stationary $u$-orbit is
isolated.\footnote{Rigid means that the only deformations of the configuration of the edges with
given lengths are rotations. For an isolated orbit, there exists a neighbourhood of the orbit that
does not intersect any other orbit.}.

\subsubsection{$c$ transformation as parity transformation}
\label{parity}

Furthermore, we would like to point out that given one orbit of stationary phase, we can
always
construct a different one via parity transformation of the $B_{fe}$ vectors (see also section
\ref{sec:geom} about $c$ transformations). After integrating out gauges these two points are
related by
\begin{equation}
 \begin{split}
  n_f'=n_f \quad ,\\
  B_{fe}'=-B_{fe} \quad ,
 \end{split}
\end{equation}
so also the angles are related by $\theta_e'=-\theta_e$
($n_f$ are preserved as pseudovectors).
Finally, we see that the asymptotic contribution from the parity related stationary orbits is just
the complex conjugate of the original one, such that the complete expansion is real.

In order to provide the correct expression of the action before performing the remaining stationary
point analysis, it is necessary to compute the normalization of the intertwiners, the so-called
`Theta' graph.

\subsection{Normalization - `Theta' graph}\label{sec:Theta}

We need to compute the self-contraction of the invariants $C_f$  using the (in this case)
symmetric bilinear form $\epsilon$ (as a generalization of the anti-symmetric form $\epsilon$ of
spin $\coh$ to arbitrary representations). Its special properties allow us to relate the
$\epsilon$ product
($(\cdot,\cdot)$) to the scalar product on $SU(2)$:
\begin{equation}
 (C_f,C_f)=\langle \overline{C_f}, \begin{pmatrix}
                                           0 &1\\ -1 & 0
                                         \end{pmatrix} C_f\rangle=\langle {C_f},C_f\rangle \quad ,
\end{equation}
since $C_f$ is real and $SU(2)$ invariant. The integral of the contraction of the intertwiner with itself is given by:
\begin{equation}
\begin{split}
 \int_{SU(2)^2\times
S_1^6}dU_1dU_2&\prod_i\frac{\rd\phi_{i1}}{2\pi}\prod_i\frac{\rd\phi_{i2}}{2\pi}\\
&f_1(\{\phi_{i1}\})f_2(\{\phi_{i2}\})
\prod_i\left(\coh|,O_{\phi_{i1}}^{-1}U_1^{-1}U_2O_{\phi_{i2}}|\coh\right)^{2j_i} \quad .
\end{split}
\end{equation}
Its stationary point conditions are:
\begin{itemize}
 \item $B_{i1}=-B_{i2}$ .
\item $\sum_{i}B_{i1}=\sum_i B_{i2}=0$ .
\end{itemize}
As the `Theta' graph itself is an evaluation of a spin network its effective action have the
same transformations on the action as described in \ref{sec:transformations}.

The $u$ symmetry can be ruled out just by dropping the integration over ${U}_1$. Then one is left
with the group $G$ generated by the transformations
\begin{equation}
 r_{f1},\ r_{f2},\ -u_2,\ o_{f1},\ o_{f2},\ -o_{f1,i},\ -o_{f2,i} \quad .
\end{equation}
On the stationary $H$ orbits, i.e. the normal subgroup of $G$ generated by $\{o_f,-u_f\}$,
these transformations act as the group $K=G/H$, which gives $\mathbb{Z}_2^3\times
\mathbb{Z}_2^3$.

This group acts freely on the stationary $H$ orbits and as before the modifiers suppress all but one
of the $H$ orbits.
If we take $f_1=f_2=1$ and restrict ourselves to the case where $\sum_i j_i$ is
even (all $j_i$ integer) then
the action is invariant with respect to all transformations, thus every
stationary orbit
contribute the same $\frac{1}{2^6}$ of the overall result.\footnote{In the case
when $\sum j_e$ is not even,
or some $j_e$ are not integer, this choice leads to vanishing invariant.}

The computation of the full expansion of the theta graph in the even
case also gives an expansion on the stationary orbit in the presence of $f_i$. This is
briefly discussed in the next section.

\subsubsection{Theta graph for integer spins and $\sum j$ even}

We will derive the complete expansion for $\sum j$ even.
We need to compute
\begin{equation}
 \prod_{i} C_{j_i}(C^{j_1j_2j_3}_{000})^2
\end{equation}
where $C_{j_i}$ (see also appendix \ref{sec:C_j}) is the normalization of the $|0\rangle$
vector.

In appendix \ref{theta:sec} we show (following \cite{Varshalovich}) that the theta
graph $(C^{j_1j_2j_3}_{000})^2$ is
equal to
\begin{equation}
 \frac{1}{2 \pi {S}}\left(1+O\left(\frac{1}{l^2}\right)\right) \quad ,
\end{equation}
where $S$ is the area of the triangle with edges $j_i+\frac{1}{2}$.

\subsection{Final formula}
\label{sec:final-form}

Let us state the final formula normalized by the square roots of the `Theta' diagrams. Those
are equal to:
\begin{equation}\label{eq:Clebsch-to-Penrose}
 (-1)^{s_f}2^{-7/2}\sqrt{\frac{\prod_{e\subset f}C_{j_e}}{\pi
S_f}}\left(1+O\left(\frac{1}{l^2}\right)\right) \quad ,
\end{equation}
where $s_f$ is a sign factor necessary to be consistent with \cite{Penrose,Penrose2} that will be derived in
\ref{sec:sign_intertwiner}.

To summarize the various calculations of this chapter, the contraction of normalized intertwiners
has the following asymptotic expansion after the stationary phase
approximation for the angles $\phi_{fe}$ has been performed and the asymptotic expansion from
\eqref{eq:Clebsch-to-Penrose} has been inserted:
\begin{equation} \label{eq:phi_int_performed}
\begin{split}
 &\frac{(-1)^{s+\sum_f s_f+\sum_e s_e}e^{-i\frac{\pi}{4}\sum_e
\sign_e}}{2^{\frac{5}{2}|E|-\frac{7}{2}|F|}\pi^{\frac{1}{2}|F|+\frac{1}{2}|E|}}
\frac{\prod_{f\in F}S_f^{\coh}}{\prod_{e\in E} l_e^{\coh}}\\
&\int\prod_{f\in F}
\delta(|n_f|^2-1)\ d^3n_f\frac{1}{\sqrt{\prod_{e}|\sin\theta_e|}}
e^{i\sum_e\left(l_e\theta_e - \frac{1}{8 l_e} \cot \theta_e \right)} \quad .
\end{split}
\end{equation}
As it will be shown in appendix \ref{sec:sign}, $s+\sum_f s_f+\sum_e s_e=0$ mod $2$ and thus
the term
\begin{equation}
 (-1)^{s+\sum_f s_f+\sum_e s_e}
\end{equation}
in the integral can be omitted.

This is the contribution up to next-to-leading order. It is straightforward to generalize it to higher
order due to the complete expansion of the edge amplitude (section
\ref{sec:DLphi})
and the expansion of `Theta' diagrams (appendix \ref{theta:sec}).

In the next section we will focus our attention on the specific example of the $6j$ symbol. After another variable transformation to the set of exterior dihedral
angles of the tetrahedron has been performed, we obtain the action of flat first order Regge
Calculus , i.e. Regge Calculus in which both edge lengths and dihedral angles are considered as independent variables. The stationary point conditions
(with respect to the dihedral angles) will reduce the action to ordinary Regge calculus, such
that the geometry is entirely described by the set of edge lengths, where angles on the
stationary point agree with the angles given for a tetrahedron built from the lengths.
 We will perform the
stationary point analysis, in particular compute the determinant of the Hessian matrix, and obtain
the correct asymptotic expression for the $SU(2)$ $6j$ symbol \cite{PR}.

\newpage

\section{Analysis of 6j symbol and first order Regge Calculus}
\label{sec:first-order}

In this section, we will perform the remaining integrations via stationary phase approximation
starting from \eqref{eq:phi_int_performed} in the case of the $6j$ symbol. As we are restricting the discussion to a specific spin network, we introduce the following
notations:

This spin network consists of $4$ faces $f$, which we will simply count by $i \in \{1,\dots,4\}$,
and $6$ edges $e$, which we will denote by $ij, i < j$, i.e. the faces sharing it. On the stationary
point with respect to $\{\phi_{fe}\}$, we have two configurations of $B_{fe}$, which we will label
accordingly as $B_{ij}$ and similarly $\theta_{ij}$ using the convention that $\theta_{ij}$ is the
angle at the edge $l_{ij}$.

In \cite{Frank3D} it has been shown that the $6j$ symbol can be interpreted as a tetrahedron on the
points of stationary phase (for non-degenerate configurations). In section \ref{sec:action} we
have shown
that our approach gives the same interpretation. Hence, we can assume that for one stationary point,
the normals to the faces $n_i$ of the tetrahedron are outward pointing and the $B_{ij}$ vectors are
oriented such that $\theta_{ij} \in (0,\pi)$. For the second stationary orbit, described by
$B'_{ij}=-B_{ij}$, the angles are negative, hence this contributes the complex conjugate.

In order to perform the remaining stationary point analysis, it is necessary to perform another
variable transformations from normals of faces $n_i$ to angles between these normals $\theta_{ij}$ followed by integrating out gauge degrees of freedom corresponding to $u$ transformations:
\begin{equation}
n_i \rightarrow \theta_{ij} \quad .
\end{equation}
This transformation is performed in appendix \ref{variables-theta} in great detail, and we
obtain the following relation:
\begin{equation} \label{eq:normals_to_theta}
\prod_{i} d^3 n_i \delta(|n_i|^2 -1) \rightarrow \prod_{ij} d \theta_{ij} \prod_{ij} |\sin
\theta_{ij}| \delta( \det \tilde{G}) \quad ,
\end{equation}
where $\tilde{G}$ denotes the angle Gram matrix (for exterior dihedral angles) of a tetrahedron with
components $G_{ij} = \cos(\theta_{ij})$, with $\theta_{ii}=0$. Using \eqref{eq:normals_to_theta} and
simplifying \eqref{eq:phi_int_performed} for the case of the $6j$ symbol, we obtain in the
neighbourhood of the stationary point:
\begin{equation}
\frac{e^{-i\frac{6}{4}\pi}\prod_i S_i^{\coh}}{2\pi^3\prod_{i<j}l_{ij}^{\coh}}
\int \prod_{i<j} \rd \theta_{ij}
\underbrace{\prod_{i<j} |\sin\theta_{ij}|\delta(\det\tilde{G})}_{\rm Jacobian}
\frac{1}{\prod_{i<j}\sqrt{|\sin\theta_{ij}|}}
e^{i\sum_{i<j}\left(l_{ij}\theta_{ij} -\frac{1}{8l_{ij}}
\cot\theta_{ij}\right)} \quad .
\end{equation}
Let us consider one of the stationary points for which $\sin\theta_{ij}>0$. The second one
contributes the complex conjugate of the first because two points (orbits) are related by $c$
(parity)
transformations:
\begin{equation} \label{eq:first_order_regge}
 \frac{ i}{4\pi^4}\frac{|l|\prod_i S_i^{\coh}}{\prod_{i<j}l_{ij}^{\coh}}
\int \rd\rho\prod_{i<j}
\rd \theta_{ij}
\prod_{i<j} \sqrt{\sin\theta_{ij}}
e^{i\left(\sum_{i<j}\left(l_{ij}\theta_{ij} -\frac{1}{8l_{ij}}\cot\theta_{ij}
\right)-|l|\rho\det\tilde{G}\right)} \quad ,
\end{equation}
where $|l|^2:=\sum_{i<j} l_{ij}^2$ and $\rho$ is a Lagrange multiplier.

It is worth to examine the action in \eqref{eq:first_order_regge} in more detail: This function of
edge lengths $l_{ij}$ and angles $\theta_{ij}$ is known as the action for `first order' Regge
Calculus \cite{Barrett:1994nn}. We will comment on this further in section
\ref{sec:1st-Regge6j}.

In the next section we will perform a stationary phase approximation for the integrations over the
angles $\theta_{ij}$. We
will use the improved action $\sum_{i<j} l_{ij}\theta_{ij}$, where we regard higher order
corrections as the vertices of a Feynman diagram expansion, and the resulting points of
stationary phase will correspond to perturbed stationary points obtained previously from the
stationary point analysis w.r.t. the $SU(2)$ group elements $U_f$ in section
\ref{sec:construction_invariants}.

\subsection{Stationary point analysis}

The stationary point conditions for the action \eqref{eq:first_order_regge} are:
\begin{itemize}
\item Derivative with respect to $\theta_{ij}$:
\begin{equation} \label{eq:stationary_theta}
l_{ij} - |l| \rho \frac{\partial \det \tilde{G}}{\partial \theta_{ij}} = 0 \quad .
\end{equation}
\item Derivative with respect to $\rho$:
\begin{equation} \label{eq:stationary_rho}
- |l| \det \tilde{G} = 0 \quad .
\end{equation}
\end{itemize}
Equations \eqref{eq:stationary_theta} and \eqref{eq:stationary_rho} are exactly those equations
stating that $\theta_{ij}$ are the exterior dihedral angles of a tetrahedron formed by edges of
length $l_{ij}$ (see appendix \ref{app:technical} and \cite{DF}). From the stationary point
analysis w.r.t. group elements $U_f$ we know that all normals $n_i$ to the faces are outward
directed\footnote{The point of stationary phase w.r.t. the angles $\theta_{ij}$ is only a small
perturbation in comparison to the stationary point w.r.t. group elements.}. The areas of the
respective face are denoted by $S_i$. For a flat tetrahedron, the following relation holds (see for
example \cite{DF,Kokkendorff}):
\begin{equation} \label{eq:relation_angle_length}
 l_{ij} =\frac{2}{3}\frac{1}{V}S_iS_j\sin\theta_{ij}\ .
\end{equation}

On the other hand $\det\tilde{G}=0$ holds, where a (single) null eigenvector of
$\tilde{G}$ is given by the vector of areas of the triangles $(S_1,\dots,S_4)$ (of the
tetrahedron)\footnote{This illustrates that $\det \tilde{G}=0$ imposes the closure of the flat
tetrahedron.}. Thus follows:
\begin{equation} \label{eq:derivative_gram}
 \frac{\partial\det\tilde{G}}{\partial\theta_{ij}}=-2 \frac{{\det}'\tilde{G}}{\sum_kS_{k}^2
}S_iS_j\sin\theta_{ij}\overset{\eqref{eq:relation_angle_length}}
{=} -3 \frac{V{\det}'\tilde{G}}{\sum_kS_{k}^2}\
l_{ij} \quad ,
\end{equation}
where $\det' \tilde{G} = \sum_i \tilde{G}_{ii}^*$ and $\tilde{G}^*_{ii}$ is the $(i,i)$th minor of
$\tilde{G}$. $\det'\tilde{G}$ is computed in appendix \ref{bunch41}:
\begin{equation}\label{eq:help1}
 {\det}'\tilde{G}=\frac{3^4}{2^2}(\sum_i S_i^2)\frac{
V^4}{\prod S_i^2} \quad .
\end{equation}
Using \eqref{eq:derivative_gram} and \eqref{eq:help1}, we solve \eqref{eq:stationary_theta}
for the Lagrange multiplier $\rho$:
\begin{equation}
 \rho=-\frac{2^2\prod S_i^2}{3^5 V^5|l|}\quad .
\end{equation}
The quadratic order in the expansion around the stationary point, which we also call the kinetic
term, i.e. the Hessian matrix of the action, is given by:
\begin{equation}
\mathcal{H} := -i|l|\left(\begin{array}{cc}
        0 &\frac{\partial\det\tilde{G}}{\partial\theta_{ij}}\\
    \frac{\partial\det\tilde{G}}{\partial\theta_{km}} & \rho
\frac{\partial\det\tilde{G}}{\partial\theta_{ij}\partial\theta_{km}}\\
       \end{array}
\right) \quad .
\end{equation}
To complete the stationary point analysis, we have to compute the determinant of its inverse
evaluated on the stationary point.

\subsection{Propagator and Hessian}

Let us introduce a function of lengths $l$:
\begin{equation}
 \lambda=|l|\rho=- \frac{2^2 \prod S_i^2}{3^5 V^5} \quad .
\end{equation}
It is of scaling dimension $1$ with respect to $l$.

\subsubsection{Propagator}

We will prove that the inverse of the kinetic term is equal to
\begin{equation}
 \mathcal{H}^{-1}=i\left(\begin{array}{cc}
        \frac{c}{|l|^2} &\frac{1}{|l|}\frac{\partial\lambda}{\partial l_{ij}}\\
   \frac{1}{|l|}\frac{\partial\lambda}{\partial l_{kl}} & \frac{\partial\theta_{ij}}{\partial
l_{kl}}\\
       \end{array}
\right) \quad ,
\end{equation}
where $c$ is a constant (defined in Lemma \ref{lm:bunch} in appendix \ref{app:technical}). Let
us compute
\begin{equation}
 i\left(\begin{array}{cc}
        \frac{c}{|l|^2} &\frac{1}{|l|}\frac{\partial\lambda}{\partial l_{ij}}\\
   \frac{1}{|l|}\frac{\partial\lambda}{\partial l_{mn}} & \frac{\partial\theta_{ij}}{\partial
l_{mn}}\\
       \end{array}
\right) (-i)|l|\left(\begin{array}{cc}
        0 &\frac{\partial\det\tilde{G}}{\partial\theta_{kl}}\\
    \frac{\partial\det\tilde{G}}{\partial\theta_{mn}} & \rho
\frac{\partial\det\tilde{G}}{\partial\theta_{kl}\partial\theta_{mn}}\\
       \end{array}
\right) \quad .
\end{equation}
This gives
\begin{equation} \label{eq:produc_hessian_inverse}
 |l|\left(\begin{array}{cc}
\frac{1}{|l|}\frac{\partial\lambda}{\partial
l_{mn}}\frac{\partial\det\tilde{G}}{\partial\theta_{mn}} &
\frac{\partial\det\tilde{G}}{\partial\theta_{ij}}\frac{c}{|l|^2}+\frac{1}{|l|}\frac{\partial\lambda}
{\partial l_{mn}}\rho
\frac{\partial\det\tilde{G}}{\partial\theta_{mn}\partial\theta_{ij}}\\
\frac{\partial\det\tilde{G}}{\partial\theta_{mn}}\frac{\partial\theta_{mn}}{\partial l_{ij}}
 &
\frac{1}{|l|}\frac{\partial\lambda}{\partial l_{ij}}
\frac{\partial\det\tilde{G}}{\partial\theta_{kl}}+
\frac{\partial\theta_{ij}}{\partial
l_{mn}}\rho\frac{\partial^2\det\tilde{G}}{\partial\theta_{mn}\partial\theta_{kl}}
       \end{array}
\right) \quad ,
\end{equation}
using the results of appendix \ref{bunch}, we see that \eqref{eq:produc_hessian_inverse} is
equal to
the identity.

\subsubsection{Hessian}

Similar to the angle Gram matrix discussed in the previous section, $\det \frac{\partial
\theta_{ij}}{\partial l_{kl}} = 0$ in the case of a flat tetrahedron. This is due to the fact that
given a set of dihedral angles of a flat tetrahedron, the tetrahedron is only defined up to
rotations and uniform scaling of its edge lengths. Hence, the null eigenvector of the matrix
$\frac{\partial
\theta_{ij}}{\partial l_{kl}}$ is given by the edge vector $\vec{l}:=(l_{12},\dots,l_{34})$\footnote{This is equivalent to the Schl\"afli identity in $3$D: $\sum_{ij} l_{ij} d\theta_{ij}=0$}. We rewrite the matrix $\mathcal{H}^ {-1}$ in the basis in which its second row is
parallel to $\vec{l}$ and the next ones are perpendicular to $\vec{l}$:

\begin{equation}
 i\left(\begin{array}{cccc}
        \frac{c}{|l|^2} &\frac{1}{|l|}\frac{\partial\lambda}{\partial l}&\cdots&\cdots\\
    \frac{1}{|l|}\frac{\partial\lambda}{\partial l}&0&0&0\\
    \vdots& 0&\frac{\partial\theta_{ij}}{\partial
l_{kl}}&\vdots\\
   \vdots&0&\vdots&\vdots
       \end{array}
\right) \quad .
\end{equation}
The determinant of $(-\mathcal{H}^{-1})$ is thus equal to
\begin{equation}
 \det (-\mathcal{H}^{-1})=-(-i)^7\Big(\underbrace{\frac{1}{|l|}\frac{\partial\lambda}{\partial
l}}_{=\frac{\lambda}{|l|^2}}\Big)^2\ {\det}'\frac{\partial\theta_{ij}}{\partial
l_{kl}} \quad .
\end{equation}
Since $\lambda$ is of scaling dimension $1$ (with respect to edge lengths), $l_{ij}
\frac{\partial \lambda}{\partial l_{ij}} = \lambda$. More
details and the tedious calculation of $\det' \frac{\partial \theta_{ij}}{\partial l_{kl}}$ can be
found in appendix \ref{app:technical}:
\begin{equation}
{\det}' \frac{\partial \theta_{ij}}{\partial l_{kl}} =  \frac{3^3}{2^5} \frac{|l|^2}{\prod
S_i^2} V^3 \quad .
\end{equation}
Combining all these results, we obtain:
\begin{equation}
\det (-\mathcal{H}^{-1}) =  -i\frac{1}{|l|^4}\left(-\frac{2^2 \prod S_i^2}{3^5 V^5}\right)^2
\frac{3^3}{2^5} \frac{|l|^2}{\prod
S_i^2}
V^3= -i \frac{1}{2 \; 3^7}\frac{\prod_i S_i^2}{|l|^2V^7} \quad ,
\end{equation}
and hence
\begin{equation}
 \sqrt{|\det \mathcal{H}^{-1}|}=\frac{1}{\sqrt{2} \; 3^{\frac{7}{2}}}\frac{\prod_i S_i}{|l|V^\frac{7}{2}} \quad .
\end{equation}
Since $\mathcal{H}^{-1}$ is antihermitian, it has only imaginary (and nonzero) eigenvalues.
Therefore it is important to count the number of $+i{\mathbb R}$ and $-i{\mathbb R}$
eigenvalues
in order to pick the right branch of $\sqrt{\det (-\mathcal{H}^{-1})}$. The number of positive and
negative imaginary eigenvalues is
constant on the connected components of parameter spaces. For oriented tetrahedra (one of the two
components) it can be computed in the equilateral case, i.e. all $l_{ij}$ are equal.
This was done in appendix \ref{iso:sec}, then $\mathcal{H}^{-1}$ has $4$ $i{\mathbb R}$
eigenvalues
and $3$ $-i{\mathbb R}$. Finally, we conclude:
\begin{equation}\label{eq:Hess}
 \frac{1}{\sqrt{{\det (-\mathcal
H)}}}=e^{-4i\frac{\pi}{4}}e^{3i\frac{\pi}{4}}\sqrt{|\det \mathcal{H}^{-1}|}=\frac{1}{\sqrt{2} \,
3^{\frac{7}{2}}}e^{-i\frac{\pi}{4}}
\frac{\prod S_i}{|l|V^\frac{7}{2}} \quad.
\end{equation}
The last step is to combine all the previous results to obtain the final formula for the
asymptotics of the $6j$ symbol.

\subsection{Final Result} \label{sec:final_result}

In this section, we will combine the results of the previous calculations step by step. First we perform the stationary point analysis for \eqref{eq:first_order_regge}:
\begin{equation} \label{eq:inter1}
 \begin{split}
  \frac{ i}{4\pi^4}\frac{|l|\prod_i S_i^{\coh}}{\prod_{i<j}l_{ij}^{\coh}}\prod_{i<j}
&\sqrt{\sin\theta_{ij}}
\frac{(2\pi)^\frac{7}{2}}{\sqrt{{\det (-{\mathcal H})}}}
e^{i\left(\sum_{ij}\left(l_{ij}\theta_{ij} -\frac{1}{8l_{ij}}
\cot \theta_{ij}\right)+\tilde{S}_1\right)}\\
&=i\frac{2^{\frac{3}{2}}}{\pi^\frac{1}{2}}\frac{|l| \prod_i S_i^{\coh} \prod_{i<j}
\sqrt{\sin\theta_{ij}}}{\sqrt{{\det (-{\mathcal H})}}\prod_{i<j}l_{ij}^{\coh}}
e^{i\left(\sum_{ij}l_{ij}\theta_{ij}+S_1\right)} \quad ,
 \end{split}
\end{equation}
where $S_1$ is the NLO contribution. As a next step, we substitute $\sin \theta_{ij}=
\frac{3}{2} \frac{l_{ij} V}{S_i S_j}$ (for $ \sin \theta_{ij}>0$) and \eqref{eq:Hess} in
\eqref{eq:inter1}:
\begin{equation} \label{eq:inter2}
\begin{split}
\frac{2^{\frac{3}{2}}}{\pi^\frac{1}{2}}
&\frac{|l| \prod S_i^{\coh}}{\prod l_{ij}^{\coh}} \left(\frac{3}{2}\right)^{3}
\frac{\prod l_{ij}^{\coh} V^3}{\prod S_i^{\frac{3}{2}}} e^{i\frac{\pi}{4}}\frac{\prod S_i}{ \sqrt{2} \,
3^{\frac{7}{2}} |l|V^{\frac{7}{2}}}
e^{i\left(\sum_{ij}l_{ij}\theta_{ij}+S_1\right)}
\\
&=\frac{1}{2}\frac{1}{\sqrt{12 \pi \, V}}
e^{i\frac{\pi}{4}}e^{i\left(\sum_{ij}l_{ij}\theta_{ij}+S_1\right)}
\end{split}
\end{equation}
As previously discussed, the full contribution comes from two stationary points, which are related by parity transformations. Eventually, we obtain:
\begin{equation} \label{eq:PR_final_formula}
\frac{1}{\sqrt{12 \pi V}}\left(
\cos\left(\sum_{ij}l_{ij}\theta_{ij}+\frac{\pi}{4}+S_1\right)+ O\left(|l|^{-2}\right)\right)
\quad ,
\end{equation}
as in \cite{PR}. In the formula above, we implicitly assumed that $S_1$ is real. This
property will be proven in section \ref{sec:NLO-DL}.

\subsection{First order Regge calculus}
\label{sec:1st-Regge6j}

A first order formulation of Regge Calculus \cite{Barrett:1994nn,regge_new_angle} is a
discretization of General Relativity defined on the triangulation of the manifold in which
both edge lengths and dihedral angles are considered as independent variables. Its
introduction was motivated by Palatini's formulation of Relativity where equations of motion are first order differential equations.
Its action in $3$D is given by
\begin{equation}
 S_R[l_e]=\sum_e l_e \epsilon_e, \quad \epsilon_e=2 \pi - \sum_{\tau \supset e} \theta^{(\tau)}_e
\end{equation}
where $l_e$ denotes the length of the edge $e$, $\theta^{(\tau)}_e$ denotes the dihedral
angle at edge $e$ in the tetrahedron $\tau$. By $\epsilon_e$ we denote the deficit angle at
edge $e$. For every tetrahedron an additional constraint is imposed, namely
\begin{equation}
 \det \tilde{G} = 0
\end{equation}
that enter the action via a Lagrange multiplier \cite{Barrett:1994nn}.
$\tilde{G}$ is the angle Gram matrix of the tetrahedron. One can eliminate the
$\theta_e^\tau$ variables by partially solving the equations of motion (given by variations
with respect to
$\theta_e^\tau$), then
\begin{equation}
 \theta_e^\tau=\theta_e^\tau(l)
\end{equation}
turns out to be the dihedral angle at the edge $e$ for a discrete geometry determined by the
edge lengths $\{l_e\}$.

Our derivation of the $6j$ symbol asymptotics follows the same idea. It also suggests a
suitable measure in the path integral quantization for (linearized) first order Regge calculus
in order
to reobtain the factor
$\frac{1}{\sqrt{V}}$ from Ponzano-Regge asymptotics. We also hope that our methods might be applied
in the $4$D case, where a similar action, motivated by the construction of modern spin foam models,
was proposed in \cite{Dittrich:2008va}. Furthermore, the present results could naturally
provide and motivate a triangulation independent measure for first order Regge calculus following
the approach
in \cite{Dittrich:2011vz}. Examining first order and area-angle (quantum) Regge calculus in $4$D
might also give new insights into possible measures for $4$D spin foam models.

\newpage

\section{Properties of the next to leading order and complete asymptotic expansion}
\label{sec:NLO-DL}

So far, we dealt with the asymptotic expansion of a spherical spin network evaluation in the
leading order approximation and managed to work out the example of the $6j$ symbol. However,
our method allows us to derive, in principle, the full asymptotic expansion of the evaluation by
the higher order stationary point analysis, e.g. we have already mentioned the next-to-leading
order (NLO) corrections to the contribution from edges of the spin network (on the
stationary points) in section \ref{d-sec}. Such corrections improve the asymptotic behaviour in
particular for
small spins. Therefore we will apply our formalism in this section to derive new insights on the NLO
corrections (to the $6j$ symbol).

NLO order corrections to the asymptotic formula of the $SU(2)$ $6j$ symbol have been thoroughly
discussed in \cite{LD,LD2}. In particular, the authors found evidence that the leading contributions in the
expansion in $\frac{1}{l}$ are purely real and oscillating as $\cos(S_R + \frac{\pi}{4})$, whereas
the next order term (also purely real) behaves like $\sin(S_R +
\frac{\pi}{4})$, where $S_R$ denotes the Regge action for the tetrahedron. Furthermore, this
behaviour is conjectured to be alternating for consecutive orders.

We will refer to this behaviour introduced in \cite{LD,LD2} as ``Dupuis-Livine''
(DL) property and we will show that it holds
for the full expansion of the asymptotics of any evaluation of spin networks, satisfying
certain generic conditions, for example the
$6j$ symbol in the non-degenerate case. Furthermore we will derive a new recursion relation for
the $6j$ symbol which can be applied to obtain a simpler form of the next to leading order correction to the Ponzano-Regge formula.

\subsection{Properties of the Dupuis-Livine form}

In this section we will give a definition to the Dupuis-Livine form and also discuss some of
its basic properties.

Consider an asymptotic expansion in the variables $\{j\}$  of the  following form
\begin{equation}
 \sum_{i} A_k(\{j\})e^{i\sum j_i\theta_i} \quad ,
\end{equation}
where $A_k$ is a homogeneous function in all variables $j$ of degree $k+\beta$. It can be
rewritten  in terms of the variables $\{l\}$ (with $l=j+\frac{1}{2}$):
\begin{equation}
 \sum_{i} \tilde{A}_k(\{l\})e^{i\sum l_i\theta_i} \quad ,
\end{equation}
where $\tilde{A}_0=e^{-\frac{i}{2}\sum \theta_i}A_0$.

We will say that it has the Dupuis-Livine (DL) property, if it can be written as
\begin{equation}
 \tilde{A}_0(\{l\})e^{i\sum l_i\theta_i}\sum_k B_k \quad ,
\end{equation}
where $i^k B_k$ is a real and homogeneous function of degree $k$. Note that if we write
this expansion in the form
\begin{equation}
 \tilde{A}_0(l)e^{i\sum l_i\theta_i+ S}
\end{equation}
then $S$ also has DL form (and starts with degree $1$). Furthermore, suppose that two asymptotic
series $f_1$ and $f_2$ have the DL property then also
\begin{equation}
 f_1 \, f_2 \quad ,\quad \frac{1}{f_{i}}
\end{equation}
have this property. In particular the last two relations are very useful for our discussion,
since they allow us to examine the full expansion of the evaluation of the spin network in
steps: first we examine the contributions from the edges, i.e. the partial integrations over the
$\phi_{fe}$, then the normalization factors until we eventually discuss the full expansion.

\subsection{Partial integration over $\phi$} \label{sec:DLphi}

In this subsection, we will examine whether the contributions from the partial integration over
$\phi$ have the DL property. We will prove it by using a recurrence relation
similar to Bonnet's formula for Legendre polynomials.
Therefore, it will be necessary to introduce some
technical definitions, from which we are able to derive recursion relations.

\subsubsection{Weak equivalence}

Let $\psi_i=f_ie^{S}$, where $S=k S_{-1}+\ldots$, $\Re S_{-1}\leq
0$ and $f_i$ grows at most polynomially in $k$ and admits a power series expansion in $k$.

\begin{df}
$\psi_1$ is weakly equivalent to $\psi_2$ around the
point $x_0$,
\begin{equation}
 \psi_1\equiv \psi_2 \quad ,
\end{equation}
if the expansion in $k$ of the integral of both around $x_0$ is the same.
\end{df}

If $\psi=fe^S$ then
\begin{equation}
 L^*\psi\equiv 0
\end{equation}
where
\begin{equation}
 L^*\psi=L\psi+(\diver L)\psi
\end{equation}
and $L$ is a vector field.

\subsubsection{Equivalences and recursion relations}
\label{equiv:sec}

Let us introduce
\begin{align}
 L_\pm&=\cos\tilde{\theta}\sin\alpha\frac{\partial}{\partial\alpha}
\pm i\sin\tilde{\theta}\sin\beta\frac{\partial}{\partial\beta} \quad ,\\
A_\pm&=\cos\tilde{\theta}\cos\alpha\pm i\sin\tilde{\theta}\cos\beta
\end{align}
that we regard as vector fields and functions of the variables $\alpha,\beta$. It is
straightforward to calculate
\begin{equation}
 {\diver L_\pm}=A_\pm \quad ,
\end{equation}
and
\begin{equation}
 \begin{split}
  L_+A_+=L_-A_-&=\frac{1}{2}A_+^2+\frac{1}{2}A_-^2-\cos 2\tilde{\theta}\\
  L_-A_+=L_+A_-&=A_+A_--1 \quad .
 \end{split}
\end{equation}
Starting from $L^*_{\pm} A_+^k \equiv 0$ and using the above identities, we derive the
following relation (see appendix \ref{weak:sec} for more details):
\begin{equation} \label{eq:rec-rel_A}
 -\frac{(k+2)^2}{k+1}(A_+)^{k+2}+2(k+1)\cos 2\tilde{\theta}\ A_+^k-(k-1)A_+^{k-2}\equiv 0 \quad .
\end{equation}
Therefore we introduce the following quantity:
\begin{equation} \label{eq:P_l-def}
 \tilde{P}_l=\frac{1}{C_j}A_+^{2j} \quad ,
\end{equation}
where $l=j+\frac{1}{2}$ and $C_j$ is given by \eqref{eq:C_j-first}:
\begin{equation} \label{eq:C_j-def}
 C_j = \frac{1}{4^j} \frac{\Gamma(2j+1)}{\Gamma(j+1)^2} \quad .
\end{equation}
Furthermore $C_j$ admits a complete expansion in $j$, see also appendix \ref{sec:C_j}:
\begin{equation}
 C_j= \frac{1}{\sqrt{\pi
j}}\left(1+O\left(
\frac{1}{j} \right)\right) \quad .
\end{equation}
Moreover, one can show that
\begin{equation} \label{eq:C_j-rel}
 C_{j+1}=\frac{2j+1}{2j+2}C_j\quad ,\quad C_{j-1}=\frac{2j}{2j-1}C_j \quad .
\end{equation}
Combining \eqref{eq:P_l-def}, \eqref{eq:C_j-def} and \eqref{eq:C_j-rel} with
\eqref{eq:rec-rel_A} and substituting $k=2j$ in \eqref{eq:rec-rel_A} we obtain:
\begin{equation}
\begin{split}
0\equiv &2 C_j\left[-\frac{(2j+2)(j+1)}{2j+1}\frac{2j+1}{2j+2}\
\tilde{P}_{l+1}+(2j+1)\cos 2\tilde{\theta}
\tilde{P}_l-\frac{2j-1}{2}\frac{2j}{2j-1}\tilde{P}_{l-1}\right]=\\
=&2 C_j \left[ -\left(l+\frac{1}{2}\right)\tilde{P}_{l+1}+2l\cos
2\tilde{\theta}
\tilde{P}_l-\left(l-\frac{1}{2}\right)\tilde{P}_{l-1}\right] \quad .
\end{split}
\end{equation}
But $C_j$ admits a nonzero asymptotic expansion, thus
\begin{equation} \label{eq:Pl_recursion}
  -\left(l+\frac{1}{2}\right)\tilde{P}_{l+1}+2l\cos 2\tilde{\theta}
\tilde{P}_l-\left(l-\frac{1}{2}\right)\tilde{P}_{l-1}\equiv 0
\end{equation}
around any stationary point. With the definitions given here, \eqref{eq:edge_action}, i.e.
the amplitude associated to one edge, becomes:
\begin{equation}
\frac{1}{8 \pi^2} \int d\alpha d\beta A_+^{2j} = \frac{C_j}{8 \pi^2} \int d\alpha d\beta \tilde{P}_l \quad ,
\end{equation}
which establishes the connection to our previous calculations.

Let us notice that \eqref{eq:Pl_recursion} is exactly Bonnet's recursion formula for Legendre
polynomials.

\subsubsection{Total expansion and DL property}

Over any stationary point we have shown that the integral of $\tilde{P}_l$ can be expanded as
\begin{equation}
(-1)^s \sum_{k\geq 0}
\frac{e^{il\theta}}{l^{k+\frac{1}{2}}}A_k(\theta)+O(l^{-\infty})
\quad ,
\end{equation}
where $\theta=2\tilde{\theta}$ is now the $SO(3)$ angle and $s$ is a sign factor
that comes from the $SU(2)$ angle\footnote{ Values of the integral for $\tilde{\theta}$ and $\tilde{\theta}+\pi$ differ by the
factor $(-1)^s$. This restricts $A_k$ to be of the form described above.}.
Moreover we know from the previous section that
\begin{equation} \label{eq:rec-Legendre}
  -\left(l+\frac{1}{2}\right)\tilde{P}_{l+1}+2l\cos \theta
\tilde{P}_l-\left(l-\frac{1}{2}\right)\tilde{P}_{l-1}\equiv 0 \quad .
\end{equation}
Applying the asymptotic form to the recursion relations, we obtain:

\begin{lm}\label{ergod3}
For every $m\geq 0$
\begin{equation} \label{eq:lemma1}
 \sum_{k \leq m}(2\beta^k_{m+1-k}+\beta^k_{m-k})A_k i^{m+1-k}\sin\left(\theta -
\frac{\pi}{2}(m-k)\right)=0 \quad ,
\end{equation}
where
\begin{equation}
 \beta^k_{m}=\frac{(-k- \frac{1}{2})_{m}}{ m!}\in {\mathbb R} \quad ,
\end{equation}
and
\begin{equation}
(a)_m = a \cdot (a-1) \cdot \ldots \cdot (a - m + 1),\quad (a)_0 = 1 \quad .
\end{equation}
\end{lm}

We will prove the lemma in appendix \ref{proof:sec}.

Consider the case where $k=m$ in \eqref{eq:lemma1}. For any $m \geq 0$ one
obtains that
\begin{equation}
 2\beta^m_{1} + \beta^m_0= -2 \left(m + \frac{1}{2}\right) + 1 = -2m \quad ,
\end{equation}
such that \eqref{eq:lemma1} can be rewritten in the following way:
\begin{equation} \label{eq:rewritten}
 2 m A_m i \sin(\theta)=\sum_{k<m} (2\beta^k_{m+1-k} + \beta^k_{m-k}) A_k
i^{m+1-k} \sin\left(\theta-\frac{\pi}{2}(m-k)\right) \quad .
\end{equation}
Let us introduce
\begin{equation}
B_m := A_m i^{-m} e^{i\frac{\pi}{4}{\rm sign}\sin\theta} \quad .
\end{equation}
From the asymptotics of the integrations over
$\tilde{P}_l$ follows that $B_0 \in
\mathbb{R}$ and \eqref{eq:rewritten} can be rewritten as
\begin{align}
2 m  \underbrace{A_m i^{-m} e^{i \frac{\pi}{4}{\rm sign \sin(\theta)}}}_{=B_m}
\sin(\theta) = & \sum_{k < m} \left(2 \beta^k_{m+1-k} + \beta^k_{m-k} \right)
\underbrace{A_k i^{-k} e^{i \frac{\pi}{4} {\rm sign \sin(\theta)} }}_{=B_k}
\sin\left( \theta - \frac{\pi}{2} (m-k) \right)\nonumber \\
\label{eq:rec-rel}
\iff 2 m B_m \sin(\theta) =& \sum_{k<m} \left(2 \beta^k_{m+1-k} +
\beta^k_{m-k}\right) B_k \sin\left(\theta-\frac{\pi}{2}(m-k)\right) \quad .
\end{align}
This implies that all $B_k\in {\mathbb R}$  and it proves that the asymptotic
terms (in the connected component expansion - $e^S$) are of the form
\begin{equation}
 \tilde{A_k}\in i^k{\mathbb R}\text{ for } k>0 \quad .
\end{equation}
This proves that the contributions from the integration  over $\phi$ evaluated on the points of
stationary phase are of DL form.

\subsubsection{The total expansion of the original integral}
\label{sec:total-phi}

We know that the total expansion of the original integral around the
stationary point is of the form given in \eqref{eq:first-exp}.
Using the recurrence relation \eqref{eq:rec-rel} we can compute its next-to-leading order:
\begin{equation}
 C_j\frac{(-1)^{s}}{4 \sqrt{2 \pi l \left| \sin \theta \right|}}
e^{i(  l\theta  -  \frac{\pi}{4} \text{sign}(\sin
\theta)-\frac{1}{8l}\cot\theta)}
\left(1+O\left(\frac{1}{l^2}\right)\right) \quad .
\end{equation}
As a next step, we will examine whether the normalization factors computed from the self-contraction
of intertwiners is of DL form as well.

\subsection{Different forms of intertwiners and DL property}

To examine whether the normalization factors satisfy the DL property, we will construct different
forms of invariants. Since the (three-valent) intertwiner is unique, all new constructions
are proportional to the original one.

Let $U_i$ be distinct group elements from a sufficiently small neighbourhood of the identity. Let
\begin{equation}
\begin{split}
 C_{U_i,f}=&\int dU\int \prod\frac{d\phi_i}{2\pi} f(\phi_1,\phi_2,\phi_3)\\
&UU_1O_{\phi_1}|\coh\rangle^{2j_{e_1}}\otimes
UU_2O_{\phi_2}|\coh\rangle^{2j_{e_2}}\otimes
UU_3O_{\phi_3}|\coh\rangle^{2j_{e_3}}
\end{split}
\end{equation}
be the new invariant, where $f$ is such a function that it is constant in the neighbourhood of the
angles, which satisfy the stationary phase conditions, i.e. where all
\begin{equation}
 B_i=j_{e_i}UU_1O_{\phi_1}\begin{pmatrix}
                    1 &0\\ 0 &-1
                   \end{pmatrix}(UU_1O_{\phi_1})^{-1}
\end{equation}
sum to zero. We will choose $U_i$ in such a way (described below) that such points are separated. In such a case we can choose $f$ to be nonzero around only one of them.

Let us now describe $U_i$. For given three vectors $B_i$ in the plane
perpendicular to $H=\left(\begin{array}{c c} 0 & i\\ -i & 0  \end{array}\right)$ (see also
section \ref{sec:construction_invariants} for more details) such that
\begin{equation}
 \sum B_i=0 \quad ,
\end{equation}
we choose $U_i$ in the neighborhood of identity such that $U_iHU_i^{-1}\perp
B_i$. There are many such choices which will be used in
the sequel.

Let us take contraction of such a $C_{U_i,f}$ with the intertwiner $C_{f'}$ obtained with the help of modifiers.
\begin{equation}
 \left(C_{U_i,f},C_{f'}\right)
\end{equation}
Due to the definition of $C_{U_i,f}$, there is only one $-u$ and $o_f$ orbit of
stationary points on which $f$ and $f'$ are nonzero.
These are given by the conditions
\begin{equation}
\begin{split}
 &B_{fe_i}=-B_i \quad ,\\
&n_f\perp B_{fe_i} \quad ,\\
&\underbrace{UU_iHU_i^{-1}U^{-1}}_{n_{f'}}\perp B_{fe_i} \quad .
\end{split}
\end{equation}
Hence, on the stationary point $U$ is of the form
\begin{equation}
 U=\begin{pmatrix}
    \cos\alpha &\sin\alpha\\ -\sin\alpha &\cos\alpha
   \end{pmatrix} \quad .
\end{equation}
If we choose $U_i$ in such a way that
\begin{equation}
 U_iHU_i^{-1}\cdot H\not=\pm 1 \quad,
\end{equation}
i.e. the two normal vectors are not (anti)parallel, then also
\begin{equation}
 UU_iHU_i^{-1}U^{-1}\cdot H\not=\pm 1 \quad .
\end{equation}
This guarantees that this configuration is non-degenerate, such that the partial integration
over $\phi$  (see section \ref{sec:DLphi}) and the fixing of the $o_f$
and $u_f$ symmetry can be performed. Following the same method as presented in section
\ref{sec:DLphi} we prove that the asymptotic expansion of $\left(C_{U_i,f},C_{f'}\right)$ has
the DL
property. Similar considerations apply to
\begin{equation}
 \left(C_{U_i^1,f^1},C_{U_i^2,f^2}\right)
\end{equation}
if $U_i^1H(U_i^1)^{-1}\cdot U_i^2H(U_i^2)^{-1}\not=\pm 1$.

Finally, using the uniqueness of the intertwiner, we obtain
\begin{equation} \label{eq:DLclebsch}
 \left(C_{f},C_{f}\right)=\pm\frac{\left(C_{f},C_{U_i^2,f^2}\right)\left(C_{
U_i^1,f^1},C_{f}\right)}{\left(C_{U_i^1,f^1},C_{U_i^2,f^2}\right)} \quad .
\end{equation}
As a product of functions whose asymptotic expansion is of DL form, it follows directly that
\eqref{eq:DLclebsch} is of DL form, too.

\subsection{Leading order expansion and a recursion relation for the $6j$ symbol} \label{sec:NLO_expansions}

In the two previous sections we have shown that both the contributions from partial integrations
over $\phi$ and the normalization factors satisfy the DL property. Hence using properties explained
in appendix \ref{DL-app} we have proven the conjecture from \cite{LD,LD2}.

In this section we will discuss the next-to-leading order expansion for the $6j$-symbol.
Therefore we do a brief recap of the results of section \ref{sec:first-order}.

From the stationary point (with outward pointing normals) we have contributions from the Hessian,
i.e. the kinetic term, and higher order terms, which are computed using a Feynman diagrammatic
approach:
\begin{equation}
 \propto\frac{1}{\sqrt{ \det({-\mathcal H})  }}e^{i\sum l_{ij}\theta_{ij}+S_1} \quad ,
\end{equation}
where $S_1$ are the evaluations of the connected Feynman diagrams of the expansion in
$\{\theta, \rho \}$ evaluated on the stationary point of the action $i\sum l_{ij}\theta_{ij}
$, using $-{\mathcal H}^{-1}$ as the propagator of this theory. We are interested only in
$|l|^{-1}$ contributions, the respective Feynman rules are briefly discussed in  appendix
\ref{sec:Feyn}.

The expansion up to the next to leading order is of the form
(see also section \ref{sec:final_result}):
\begin{equation}
 \frac{1}{2} \frac{1}{\sqrt{12 \pi V}}
e^{i\left(\sum_{ij}\left(l_{ij}\theta_{ij}-\frac{1}{8l_{ij}}
\cot \theta_{ij}\right)+\tilde{S}_1\right)}
= \frac{1}{2} \frac{1}{\sqrt{12 \pi V}}
e^{i\left(\sum_{ij}l_{ij}\theta_{ij}+S_1\right)}
\end{equation}
where $S_1$ is of order $|l|^{-1}$.
The full contribution comes from two stationary point that are related via parity
transformation, see also section \ref{parity}; their contributions are related by complex
conjugation. Hence, we obtain up to $|l|^{-1}$:
\begin{equation}
\frac{1}{\sqrt{12 \pi V}}\left(
\cos\left(\sum_{ij}l_{ij}\theta_{ij}+\frac{\pi}{4}+S_1\right)+O\left(|l|^{ -2}\right)\right)
\end{equation}
The next to leading order expansion is briefly described in appendix \ref{sec:Feyn}.
Although, this method is algorithmically more involved than the method proposed in
\cite{LD,LD2}, the final expression is also more geometric. We will now derive a recursion relation for the full $6j$ symbol using a similar idea as in \cite{Bonzom,Bonzom2} that, we hope, can serve to compute the NLO expansion in more concise way.

\subsubsection{Recursion relation for $6j$ symbols}

In this section we derive a recursion relation for the whole $6j$ symbol. First, let us introduce a multiplication operator
\begin{equation}
 N(l)=\sqrt{\prod_i \Theta_i(l)}
\end{equation}
where $\Theta_i$ is normalization (of a three-valent intertwiner) computed from the Theta graph. Furthermore we define
the operator $T^{v}_{ij}$ via its action on a function of edge lengths $\{l\}$:
\begin{equation}\label{eq:T_ij}
 T^v_{ij} \, C(l)=\left(1+v\frac{1}{2l_{ij}}\right)C(\{l_{km}+v\delta_{(ij)(km)}\}) \quad .
\end{equation}
We assume that $T^v_{ii}=1$.

As a next step, recall the definition of $\tilde{P}_l$ \eqref{eq:P_l-def} and its recursion relation
\eqref{eq:Pl_recursion}. The latter can be written as follows:
\begin{equation}
  \cos {\theta}
\tilde{P}_l\equiv
\left(\frac{1}{2}+\frac{1}{4l}\right)\tilde{P}_{l+1}+
\left(\frac{1}{2}-\frac{1}{4l}\right)
\tilde { P }_{l-1} \quad .
\end{equation}
and we can write the non-normalized $6j$ amplitude as
\begin{equation} \label{eq:Z_Pl}
 Z''(l)=\int \prod
\rd \theta_{ij}
\prod_{(ij)}\sin\theta_{ij}
\prod_{(ij)} \tilde{P}_{l_{ij}}(\theta_{ij})
\delta(\det\tilde{G}) \quad .
\end{equation}
In order to derive the recursion relation, we insert an additional $\det \tilde{G}$ into \eqref{eq:Z_Pl}:
\begin{equation} \label{eq:Z_Pl_detG}
 \int \prod
\rd \theta_{ij}\det\tilde{G}
\prod_{(ij)}\sin\theta_{ij}
\prod_{(ij)} \tilde{P}_{l_{ij}}(\theta_{ij})
\delta(\det\tilde{G})=0 ,
\end{equation}
since $\det \tilde{G}$ is constrained to vanish. Similar to \cite{Bonzom,Bonzom2}, $\det \tilde{G}$ can be expanded as a sum over perturbations:
\begin{equation} \label{eq:exp_detG}
\det \tilde{G} = \sum_{\sigma \in S_4} \text{sgn} \sigma \frac{1}{16} \sum_{\vec{v} \in \{-1,1\}^4} e^{i v_i \theta_{i \sigma_i}} \quad , 
\end{equation}
with the convention that $\theta_{ij} = \theta_{ji}$ and $\theta_{ii} =0$. Using \eqref{eq:exp_detG}, equation \eqref{eq:Z_Pl_detG} can be rewritten as:
\begin{equation} \label{eq:recursion_Z_Pl}
 \sum_{\sigma\in S_4} \sign\sigma
\frac{1}{16}\sum_{\vec{v}\in\{-1,1\}^4} \prod_i T^{v_i}_{i\sigma_i}Z''(l)=0 \quad .
\end{equation}
On the other hand, we know from previous calculations that
\begin{equation} \label{eq:6j_expansion}
 \{6j\}\equiv N^{-1}Z''(l)+c.c+O(l^{-\infty}) \quad ,
\end{equation}
such that we can summarize both \eqref{eq:recursion_Z_Pl} and \eqref{eq:6j_expansion} into the following recursion relation for the $6j$ symbol\footnote{The recursion relation has been verified numerically for several $6j$ symbols.}:
\begin{equation}
 \det\left[\frac{{T}^1_{ij}+{T}^{-1}_{ij}}{2}\right]N\{6j\}\equiv 0 \quad ,
\end{equation}
where $T^v_{ij}$ is defined as in \eqref{eq:T_ij}.

Another useful form is the following
\begin{equation}
 \sum_{\sigma\in S_4} \sign\sigma
\frac{1}{16}\sum_{\vec{v}\in\{-1,1\}^4}  \frac{N(l+v_{i\sigma_i})}{N(l)}
\left(\prod_iT^{v_i}_{i\sigma_i}\right)\{6j\}\equiv 0 \quad ,
\end{equation}
since the expansion of $\frac{N(l+v_{i\sigma_i})}{N(l)}$ is straightforward to compute. We have to point out though that the coefficients in this formula are not rational, yet they allow for nice a asymptotic expansion. Thus they should in principle allow for the computation of the
higher order expansions of the $6j$ symbol.

\section{Discussion and outlook} \label{sec:discussion}

Coherent state approaches are the only available tools so far to successfully compute the
asymptotic expansion of spin foam models \cite{Conrady-Freidel,Frank,Frank2,Frank3,FrankEPRL,Frank3D}, which gives us a first, and yet, very
incomplete understanding of the relation of spin foam models to gravity. The strength and beauty of
this approach is its clear geometrical interpretation and straightforward computation of the
dominating phase of the expansion, which is identified as the Regge action of the examined
triangulation. Despite these successes, the approach usually fails in the computation of the
determinant of the Hessian matrix, which provides the normalization to the path integral and, more
importantly, a measure on the space of geometries.

To overcome this drawback, we have introduced modified coherent states, i.e. states labelled by null
eigenvectors with respect to a generator of rotations, smeared perpendicular to the axis of
rotation. We have shown that these states allow for the same geometrical interpretation as the usual
$SU(2)$ coherent states and presented a method to deal with the (due to the smearing) increased
number of stationary points. This allowed us to
derive the well-known asymptotic expansion of the $SU(2)$ $6j$ symbol \cite{PR} entirely, by
computing its amplitude in the stationary phase approximation, first with respect to the smearing
parameters and second, after a variable transformation, with respect to the dihedral angles of the
tetrahedron.
In the process, we have discovered that the resulting amplitude is proportional to the action of the
first order formulation of Regge calculus, a result that
supports the conjecture given in \cite{Dittrich:2008va} that $4$D spin foam models can be better
described by angle and area variables instead of only edge lengths, the fundamental variables of
ordinary Regge calculus. This result could also stimulate new work following the ideas of
\cite{Dittrich:2011vz} to obtain an invariant  path integral measure (under Pachner moves \cite{Pachner1,Pachner2}) for
first order Regge
calculus and to compare it to spin foam models.

In addition to this result, we also extended the calculation to the next to leading order
correction for the $6j$ symbol. We have been able to prove the conjecture presented in \cite{LD,LD2}
that the higher order
corrections are alternatingly oscillating with the cosine or the sine of the Regge action, and
furthermore we can, in principle, calculate the asymptotic expansion up to arbitrary order. Despite
this success, we are not able to present the next-to-leading order in a short and concise way.
This is a nuisance of all known derivations of next-to-leading order expansion, see for
example \cite{LD,LD2}. However, we derived a recursion relation for the $6j$ symbol, very similar in nature to
the one in \cite{Bonzom,Bonzom2}, that can in principle be used to obtain more concise form of the next to leading order term.

The main goal of this work was not the derivation of known results, but to develop and advertise a
new coherent state method, which is capable of challenging the determination of the measure in spin
foam models \cite{Conrady-Freidel,Frank,Frank2,Frank3,FrankEPRL}. The computation of the full asymptotic expansion (even only up to leading
order) would not only increase the understanding of spin foam
models, but could also give a measure on the space of geometries, which could be compared to the
proposed measure in \cite{Dittrich:2011vz}. Given such a measure, one would be able to examine which
geometries dominate the spin foam transition amplitudes in the various models, which could also be
used to exclude some of them. Our successful and complete derivation of the asymptotic
expansion of the $SU(2)$ $6j$ symbol is a good start, however the method still has to prove itself
by tackling more complicated models. Therefore, two issues have to be overcome:

The first problem is to extend the presented coherent state approach to groups with non-unique
intertwiners. Our calculations are heavily based on the fact that the
intertwiner of three irreducible representations of $SU(2)$ is unique, which simplified the
construction of our model. The only $4$D spin foam model with unique intertwiners is the
Barrett-Crane model \cite{Barrett-Crane}, which has already been ruled out as a viable quantum gravity theory. Nevertheless, our calculations presented in this work can be applied and can lead to interesting new insights \cite{BC-paper}.

The second problem is common to all coherent state approaches to spin foam models so far; all the
known calculations are restricted to one simplex of the triangulation. To extract the asymptotic
expansion for larger triangulations and to examine possible invariances under (local) changes of the
triangulation like Pachner moves is still an open issue. In this work, before computing the
asymptotic expansion of the $6j$ symbol, we have kept the discussion as general as possible. It
would be interesting to examine, whether the relation to the first order formulation of Regge
calculus can also be found in larger triangulations or whether one obtains modifications, which
could be understood as quantum gravity effects.

At the end we would like to point out that the application of our method to the case of
the non-compact group $SL(2,\mathbb{R})$ is rather straightforward and we leave
the determination of the $3$D Lorentzian $6j$ symbol for future investigations.

\subsection*{Acknowledgement}
 The authors would like to thank Frank Hellmann, Jerzy Lewandowski and Krzysztof Meissner for
fruitful discussions,
and especially Bianca Dittrich for a lot of valuable comments
and for showing us the paper \cite{Freidel}. We would also like to acknowledge
Matteo Smerlak who focused our attention to the NLO expansion problem. W.K. acknowledges
partial
support by the grant ``Maestro'' of  Polish Narodowe Centrum Nauki nr 2011/02/A/ST2/00300 and the grant of Polish Narodowe Centrum Nauki number 501/11-02-00/66-4162. S.St. gratefully acknowledges a stipend by the DAAD (German Academic Exchange Service). This research was supported in part by Perimeter Institute for Theoretical Physics. Research at Perimeter Institute is supported by the Government of Canada through Industry Canada and by the Province of Ontario through the Ministry of Research and Innovation.


\appendix

\section{Spin network evaluation and sign convention}
\label{sec:sign}

This appendix is devoted to the sign issue. We will show how one can determine the total
sign of our formula using the prescription of \cite{Penrose,Penrose2}.

\subsection{Penrose prescription for spherical graph}
\label{sec:prescription}

In this section we will describe a canonical way to evaluate spherical (planar) spin
networks. Let us draw it on the $2$-sphere such that no edges intersect; if the spin network is 2-line irreducible
there are two distinct ways
to do so, which differ by orientation. The result of the evaluation does not
depend on this choice.
\begin{figure}[ht]
\begin{center}
\includegraphics[scale=0.5]{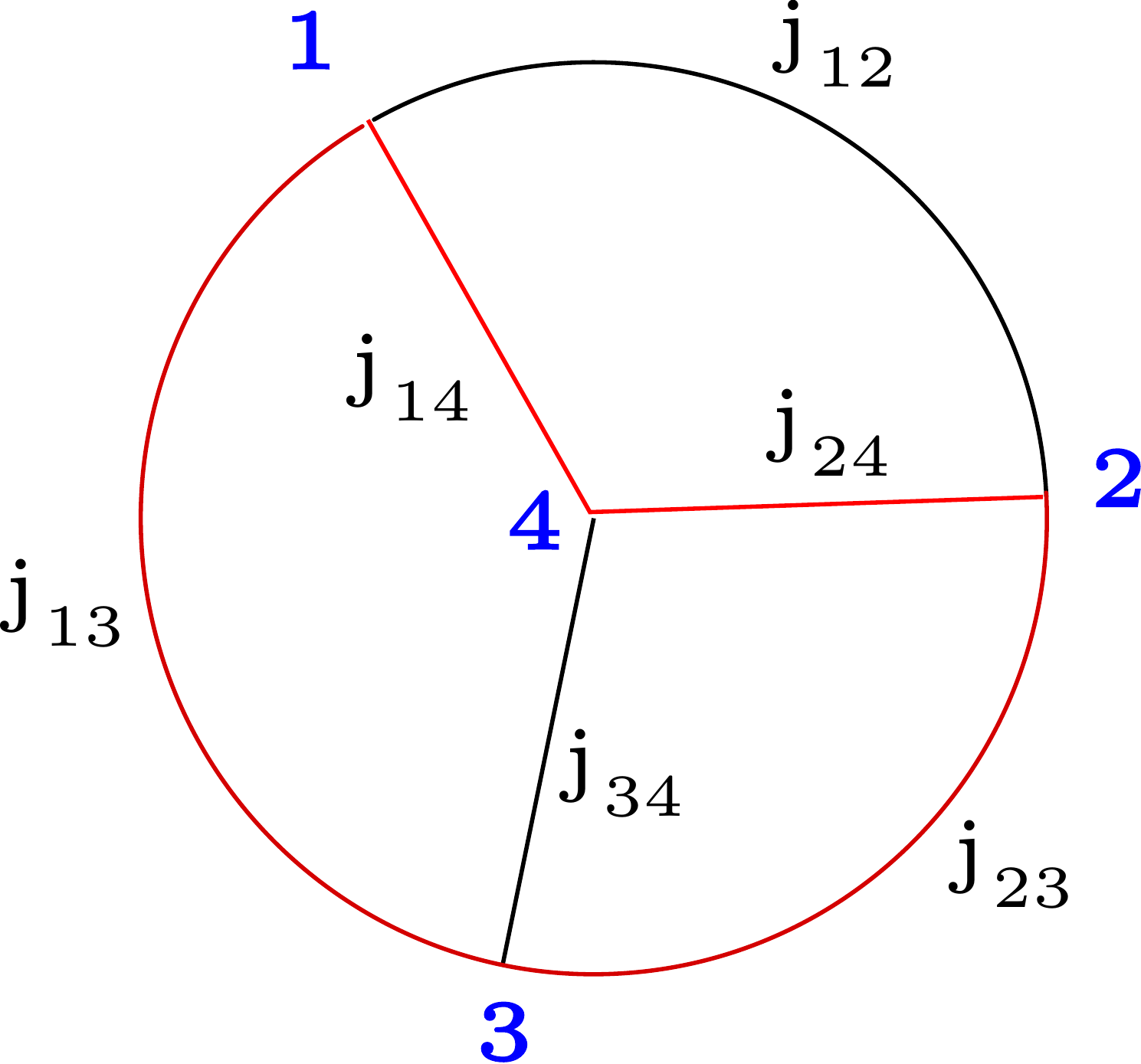}%
\caption{Orientation of intertwiners inherited from orientation of the sphere (plane).
Half-integer spins colored red.}
\label{pic-sphere-orient}
\end{center}
\end{figure}
For every node of the graph (a face in the dual picture) we have a natural cyclic order inherited from
the orientation of the sphere. In the second
step we choose any ordering of nodes (faces). This gives a natural orientation of the edges; they start in
nodes lower in the order and end in nodes higher in the order. We draw the graph on the plane as
shown on figure \ref{pic-sphere-orient}
such that the order of the nodes is preserved and the order of legs in every node is
consistent with the cyclic
order obtained above.

In the third step we count the number of crossings $s$ of half-integer edges with each other.
\begin{figure}[ht]
\begin{center}
\includegraphics[scale=0.5]{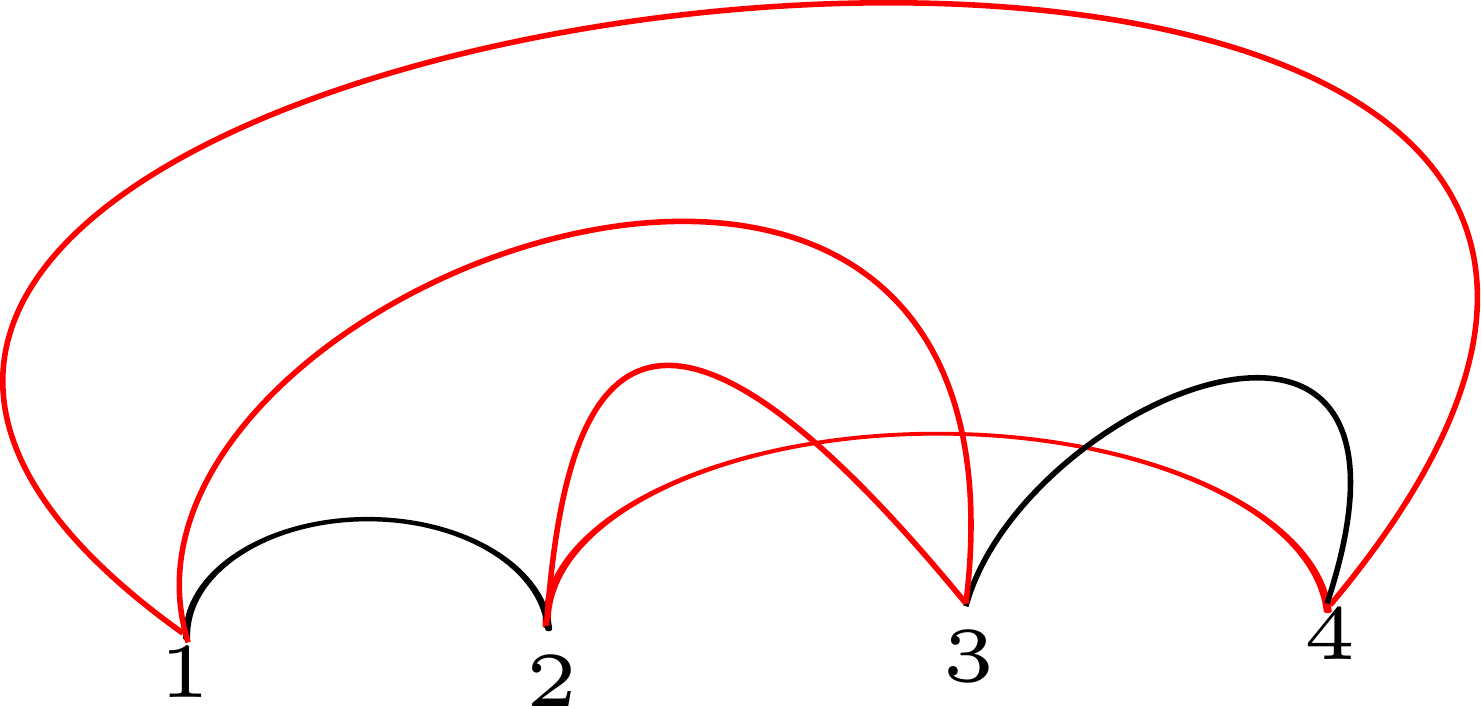}%
\caption{Nodes are in the right order and for each intertwiner the legs are in the right cyclic order. The number of
crossings for half-integer edges is $s=2$.}
\label{pic-evaluation}
\end{center}
\end{figure}
The spin network is evaluated by contracting invariants, given for every node, by using the
$\epsilon$ bilinear form oriented according to the edge orientation inherited from the nodes. The
ordering of the legs is as in figure \ref{pic-evaluation}:
\begin{equation}
 \prod_{e} \epsilon^{j_e}_{A_{s(e)e}A_{t(e)e}}\prod_{v} I_v^{A_{ve_1}A_{ve_2}A_{ve_3}} \quad .
\end{equation}
These invariants are described in \cite{Penrose,Penrose2} (see also section \ref{sec:Theta}).
One can show that the result does not depend on the made choices.

\subsection{Sign factors and spin structure}
\label{sec:sign-total}

In this section we will show how to compute the sign factor for spherical graphs.
First of all, let us notice that in the case that all $j$ are integers, the sign disappears
completely. We will prove now that this is also the case in general. Explicitly we will prove that
(see \ref{sec:final-form} for the definitions)
\begin{equation}
 s+\sum_f s_f+\sum s_e=0\ \text{mod}\ 2\ .
\end{equation}

\subsubsection{Sign factor in the intertwiner}
\label{sec:sign_intertwiner}

In this section we compute the sign $s_f$. In order to do this, we compare our invariant with the
one from \cite{Penrose,Penrose2} (given for a fixed order of
$j_1,j_2,j_3$). The dual  of the latter
is given on vectors $\xi_1^{2j_1}\otimes \xi_2^{2j_2}\otimes \xi_3^{2j_3}$ by the formula
\begin{equation}\label{eq:Penrose-inv}
 (-1)^{j_1+j_3-j_2}\,C\,\epsilon(\xi_1,\xi_2)^{j_1+j_2-j_3}
\epsilon(\xi_2,\xi_3)^{j_2+j_3-j_1}
\epsilon(\xi_3,\xi_1)^{j_1+j_3-j_2} 
\end{equation}
with normalization $C>0$ \cite{Penrose,Penrose2}.

The contraction of \eqref{eq:Penrose-inv} with our invariant is given by:
\begin{equation}\label{eq:contr-inv}
 (-1)^{j_1+j_3-j_2}C\int \frac{d\phi_1\ d\phi_2\ d\phi_3}{(2\pi)^3}
\epsilon(v_{\phi_1},v_{\phi_2})^{j_1+j_2-j_3}
\epsilon(v_{\phi_2},v_{\phi_3})^{j_2+j_3-j_1}
\epsilon(v_{\phi_3},v_{\phi_1})^{j_1+j_3-j_2} \quad ,
\end{equation}
where
\begin{equation}
 v_{\phi}=\begin{pmatrix}
           \cos\phi \\ \sin\phi
          \end{pmatrix}, \quad \epsilon(v_{\phi},v_{\phi'})=\sin(\phi'-\phi) \quad ,
\end{equation}
and we skipped the integration over $U$, since \eqref{eq:Penrose-inv} is invariant. Let us recall
our notation:
\begin{equation}
 \psi_{ij}=\phi_i-\phi_j\quad .
\end{equation}
After a change of variables
\begin{equation}
 (\phi_1,\phi_2,\phi_3)\rightarrow (\psi_{21},\psi_{32}, \phi_1)
\end{equation}
and performing one trivial integration over $\phi_1$, \eqref{eq:contr-inv} is equal to
\begin{equation}
 (-1)^{j_1+j_3-j_2}C\int \frac{d\psi_{21}\ d\psi_{32}}{(2\pi)^2}
(\sin\psi_{21})^{j_1+j_2-j_3}
(\sin\psi_{32})^{j_2+j_3-j_1}
(\sin\psi_{13})^{j_1+j_3-j_2} \quad ,
\end{equation}
with the constraint $\psi_{21}+\psi_{32}+\psi_{13}=0$.

As the expression is real (since $j_i+j_k-j_l$ is an integer), in the asymptotic limit it is
dominated by the
stationary point (maxima of the integral) of the action
\begin{equation}
(j_1+j_2-j_3)\ln|\sin\psi_{21}|+
(j_2+j_3-j_1)\ln|\sin\psi_{32}|+
(j_1+j_3-j_2)\ln|\sin\psi_{13}|+\rho(\psi_{21}+\psi_{32}+\psi_{13}) \quad ,
\end{equation}
where $\rho$ is a Lagrange multiplier and $\psi_{21},\psi_{32},\psi_{13}$ are treated as independent
variables. The stationary point condition reads
\begin{equation}
 (j_i+j_k-j_l)\cot\psi_{ij}=\rho\ .
\end{equation}
Now we can use the fact that
\begin{equation}
 \cot\psi_{32}\cot\psi_{21}+
 \cot\psi_{13}\cot\psi_{32}+
 \cot\psi_{21}\cot\psi_{13}=1
\end{equation}
to obtain
\begin{equation}
\rho^2=\frac{(j_1+j_2-j_3)(j_2+j_3-j_1)(j_1+j_3-j_2)}{j_1+j_2+j_3} \quad .
\end{equation}
Furthermore, we see that
\begin{equation}
 \cot^2\psi_{32}=\frac{(j_1+j_2-j_3)(j_1+j_3-j_2)}{(j_2+j_3-j_1)(j_1+j_2+j_3)}
=\frac{j_1^2-(j_2-j_3)^2}{(j_2+j_3)^2-j_1^2} \quad .
\end{equation}
Hence, we compute that
\begin{align}
 \cos2\psi_{32}&=\frac{\cot^2\psi_{32}-1}{\cot^2\psi_{32}+1}=\frac{j_1^2-j_2^2-j_3^2}{2j_2j_3} \quad ,\\
 \sin2\psi_{32}&=\frac{2\cot\psi_{32}}{\cot^2\psi_{32}+1} =
\pm\frac{A}{j_2j_3} \quad ,
\end{align}
where $A$ is the area of the triangle with edge lengths $j_1,j_2,j_3$. Thus $\pm 2\psi_{32}$ modulo
$2\pi$ is the angle in this triangle opposite to the edge $j_1$. Similar relations hold for
$\psi_{21}$ and $\psi_{13}$. Together with the relation
$\psi_{21}+\psi_{32}+\psi_{13}=0$, it gives the condition that
\begin{equation}
 2\psi_{21},\ 2\psi_{32},\ 2\psi_{13}\ \text{mod}\ 2\pi
\end{equation}
are oriented (i.e. incorporate sign) angles of the triangle on the plane with edges
$(j_1,j_2,j_3)$.

In the presence of a function $f_f$, only one of those stationary points contributes. Since the
Jacobian is real, the only contribution to the sign is given by the value of the integral in the
stationary point. We know that $\psi_{ij}\in (\pi,2\pi)$ for consecutive pair of edges
$(ij)$
(see \ref{modifier-sec}), thus $\sin\psi_{ij}<0$ and the total sign is
\begin{equation}
 (-1)^{j_1+j_3-j_2}\ (-1)^{j_1+j_2-j_3}
(-1)^{j_2+j_3-j_1}
(-1)^{j_1+j_3-j_2}=(-1)^{2j_2}
\end{equation}
As already discussed above, this is a relative sign of our invariant with respect to the invariant
described in \cite{Penrose,Penrose2}.

\subsubsection{The sign $\sum s_e$}
\label{sec:sign1}

In the stationary point we can write (see \ref{sec:stat-phi} and
\ref{sec:angle-inter} for the derivation)
\begin{equation}
U_{s(e)}^{-1}U_{t(e)}=(-1)^{\tilde{s_e}}O_{s(e)e}e^{-i\tilde{\theta}_{s(e)t(e)}
\begin{pmatrix}
                                                            1 &0\\ 0&-1
                                                           \end{pmatrix}}
\begin{pmatrix}
 0 & -1\\ 1 & 0
\end{pmatrix}O_{t(e)e}^{-1} \quad ,
\end{equation}
where we assumed that $\tilde{\theta}_{s(e)t(e)}\in
\left(-\frac{\pi}{2},\frac{\pi}{2}\right)$. It is straightforward to check that
\begin{equation}
U_{t(e)}^{-1}U_{s(e)}=(-1)^{\tilde{s_e} +1}O_{t(e)e}e^{-i\tilde{\theta}_{t(e)s(e)}
\begin{pmatrix}
                                                            1 &0\\ 0&-1
                                                           \end{pmatrix}}
\begin{pmatrix}
 0 & -1\\ 1 & 0
\end{pmatrix}O_{s(e)e}^{-1} \quad ,
\end{equation}
where $\tilde{\theta}_{t(e)s(e)}=-\tilde{\theta}_{s(e)t(e)}\in
\left(-\frac{\pi}{2},\frac{\pi}{2}\right)$. Thus in general we have
\begin{equation}
U_{f}^{-1}U_{f'}=(-1)^{\tilde{s_e} +c_e}O_{fe}e^{-i\tilde{\theta}_{ff'}
\begin{pmatrix}
                                                            1 &0\\ 0&-1
                                                           \end{pmatrix}}
\begin{pmatrix}
 0 & -1\\ 1 & 0
\end{pmatrix}O_{f'e}^{-1} \quad,
\end{equation}
where
\begin{equation}
 c_e=\left\{\begin{array}{ll}
             0 & f=s(e)\ \text{and}\ f'=t(e)\\
             1 & f=t(e)\ \text{and}\ f'=s(e)
            \end{array}\right. \quad .
\end{equation}
By a cycle we denote an assignment of a number $\{0,1\}$ to every edge such that
\begin{equation}
 \forall_f\ \sum_{e\subset f} c_e=0\ \text{mod}\ 2\ .
\end{equation}
The set of cycles is denoted by $Z_1$.
Abusing the notation, we will also say that the cycle is formed by edges with
$c_e=1$. Let us notice that such edges form a disjoint sum of loops that we will denote by $c_i$.

For every cycle $c$ holds
\begin{equation}\label{eq:s_e}
 \prod_i \left(\prod_j U_{f^i_j}^{-1}U_{f^i_{j+1}}\right)=1 \quad ,
\end{equation}
where $\{f^i_jf^i_{j+1}\}$ are consequtive pair of faces in the cycle $c_i$ (in the correctly chosen order).

Thus, we can write
\begin{equation}
 (-1)^{\tilde{s}(c)}=\prod_{e=[ff']\subset c} (-1)^{c_e}O_{fe}e^{-i\tilde{\theta}_{ff'}
\begin{pmatrix}
                                                            1 &0\\ 0&-1
                                                           \end{pmatrix}}
\begin{pmatrix}
 0 & -1\\ 1 & 0
\end{pmatrix}O_{f'e}^{-1} \quad ,
\end{equation}
where we used the same order of multiplications as before.
The equations \eqref{eq:s_e} translate into the set of equations satisfied by $\tilde{s}_e$:
\begin{equation}\label{eq:cycle}
 \forall\ c\in Z_1\quad \sum_{e\in c} \tilde{s}_ec(e)=\tilde{s}(c)\ {\text{mod}}\
2 \quad .
\end{equation}
Given a solution for the $\tilde{s}_e$, $U_f$ can be reconstructed up to a $U$ transformation. The solutions $\{\tilde{s}_e\}$ are not completely determined, but the residual symmetry is given by $\Ran \partial$ where
\begin{equation}
 \partial:C_0\rightarrow C_1,
\end{equation}
is a boundary operator\footnote{This is because $C_1=\ker\partial^*\oplus \Ran
\partial$}.
Those correspond exactly to $-U_f$ transformations.

We are interested in \ref{sec:final-form}
\begin{equation}
 \sum s_e=\sum_{e\in c} \tilde{s}_ec(e) \quad ,
\end{equation}
where $c$ is the cycle formed by all edges that are half-integer.

\subsubsection{Sign of basic cycles in spherical case}

In this section we will compute the sign factor $\tilde{s}(c)$ for cycles consisting of only a
single loop.
Every other cycle can be uniquely written as a sum (as the $\mathbb{Z}_2$ module) of such disjoint cycles.

Let us take such a cycle. The cycle is described by the sequence
of consecutive faces and edges.
The value of $(-1)^{\tilde{s}(c)}$ is thus equal to
\begin{equation}\label{eq:cycle2}
\prod_{\{ff'\}\in c} O_{fe}e^{-i\tilde{\theta}_{ff'}
\begin{pmatrix}
                                                            1 &0\\ 0&-1
                                                           \end{pmatrix}}
\begin{pmatrix}
 0 & -1\\ 1 & 0
\end{pmatrix}O_{f'e}^{-1} \quad .
\end{equation}
All parameters (i.e $\phi_{fe}$ and $\tilde{\theta}_{ff'}$) can be continuously deformed,
i.e. there exists
a map
\begin{equation}
 [0,1]\ni t\rightarrow \{\phi^t_{fe},\tilde{\theta}^t_{ff'}\}
\end{equation}
such that
\begin{equation}
 \forall_{e\subset f}\ \phi^0_{fe}=\phi_{fe},\quad \forall_{ff'}\
\tilde{\theta}^0_{ff'}=\tilde{\theta}_{ff'},
\end{equation}
that satisfies the following conditions:
\begin{itemize}
 \item the image of \eqref{eq:cycle2} in $SO(3)$ is always the identity,
\item at the end all
deformed $SU(2)$ angles $\tilde{\theta}_{ff'}^1$ are equal to $0$
\item for every face  $f$ with ordered pair of edges $e,\,e'$ (neighbours in the cycle), the
difference $\phi_{fe}^t-\phi_{fe'}^t\in
\left(\frac{\pi}{2},\frac{3\pi}{2}\right)$ modulo $2\pi$ during the whole
deformation process. In fact, it is larger
than $\pi$ if order of edges agrees with the orientation of the face and smaller if it does not.
\end{itemize}
The final stage of the deformation will be denoted by
\begin{equation}
\forall_{e\subset f}\ \tilde{\theta}'_{ff'}=\tilde{\theta}^1_{ff'},\quad
\forall_{ff'}\ \phi'_{fe}=\phi^1_{fe}\ \quad .
\end{equation}
Up to $2$-dimensional homotopies, there are two possible final stages of such deformations.
They differ by orientation of the cycle (loop) drawn on the plane. We assume that the faces are
ordered in agreement with total
orientation.

The cycle before and after the deformation is shown in figure \ref{pic-def11}.
\begin{figure}[ht]
\begin{center}
\includegraphics[scale=0.5]{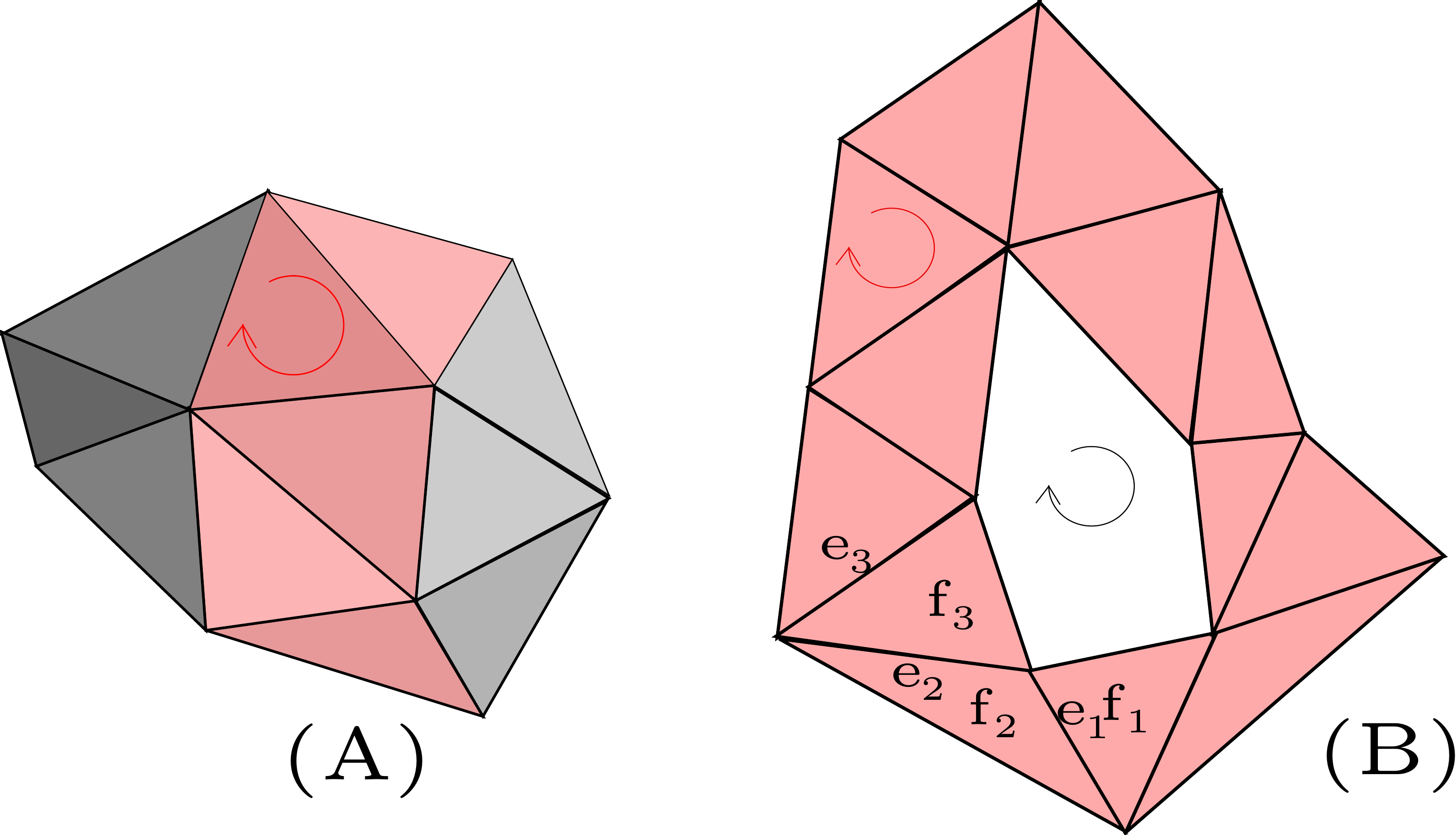}%
\caption{(A) Cycle before deformation. (B) Cycle after deformation.}
\label{pic-def11}
\end{center}
\end{figure}
The proces is shown on the figure \ref{pic-def1} on example of a single-loop cycle around the
vertex.
\begin{figure}[ht]
\begin{center}
\includegraphics[scale=0.5]{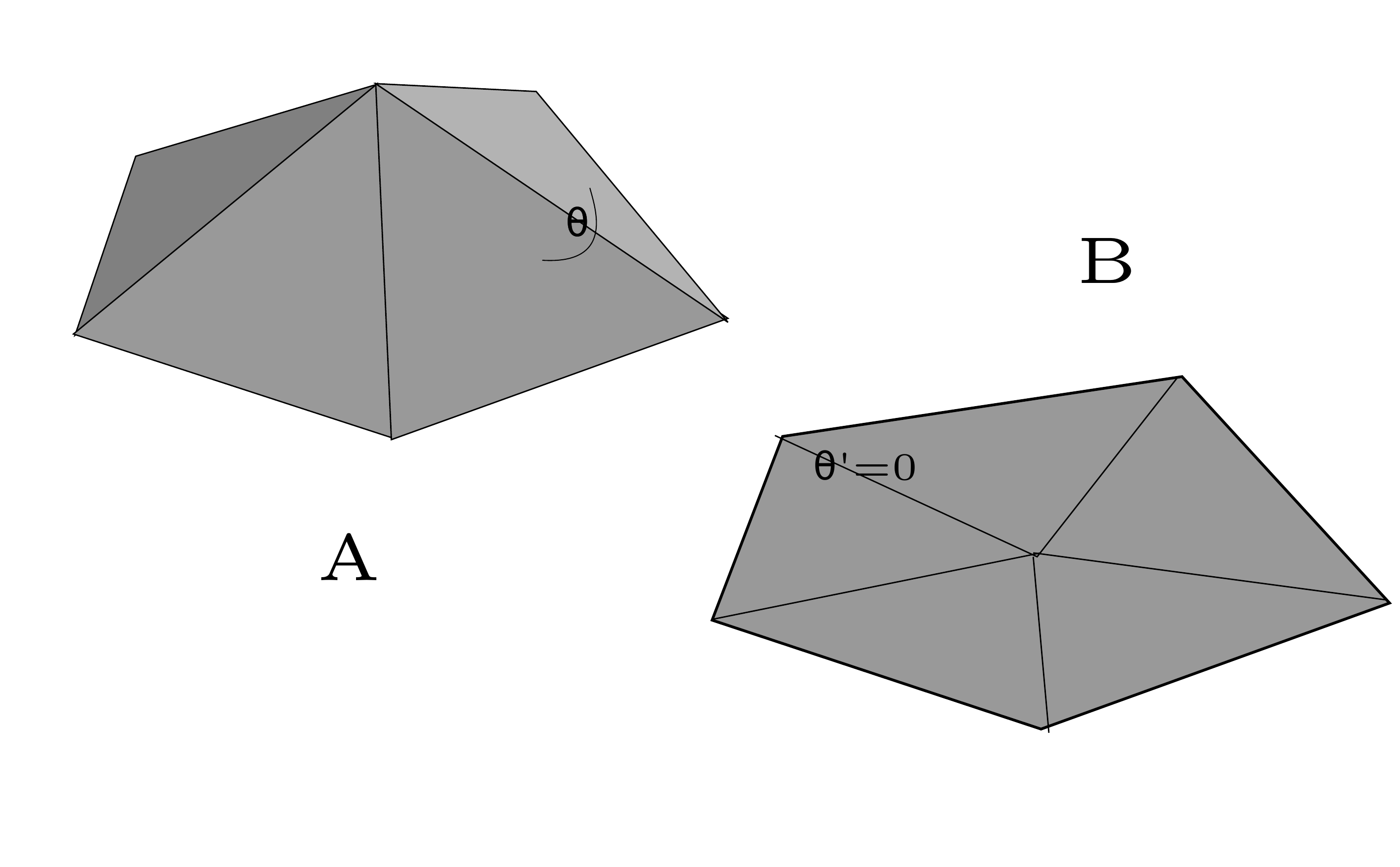}%
\caption{Example of the single-loop cycle around a vertex. (A) Cycle before deformation:
angle $\theta$ between two faces depicted. (B) Cycle after
deformation $\theta'=0$, the faces are parallel.}
\label{pic-def1}
\end{center}
\end{figure}
In the end we obtain
\begin{equation}
\prod_{\{ff'\}\in c} (-1)^{c_e}O_{fe}e^{-i\tilde{\theta}_{ff'}'
\begin{pmatrix}
                                                            1 &0\\ 0&-1
                                                           \end{pmatrix}}
\begin{pmatrix}
 0 & -1\\ 1 & 0
\end{pmatrix}O_{f'e}^{-1}
=(-1)^{C_e}\prod_{\{ff'\}} {O'}_{fe}{O'}_{f'e}^{-1}
\prod_e \begin{pmatrix}
 0 & -1\\ 1 & 0
\end{pmatrix} \quad ,
\end{equation}
where $C_e$ is the number of edges with $c_e=1$ because $O_{fe}'$ commutes with
$\begin{pmatrix}
 0 & -1\\ 1 & 0
\end{pmatrix}$.
The images of $O_{fe}'$ (and related $SO(3)$ angles $\pi(\phi_{fe}')$) satisfy (see
figure \ref{pic-def2})
\begin{equation}
 \sum_e \pi(\phi_{fe}')-\pi(\phi_{f'e}')= -\sum_{\{ee'\}\subset f}
\pi(\phi_{fe}')-\pi(\phi_{fe'}')=-(n-2)\pi \quad ,
\end{equation}
where $n$ is the
number of faces meeting in the cycle $c$.
\begin{figure}[ht]
\begin{center}
\includegraphics[scale=0.5]{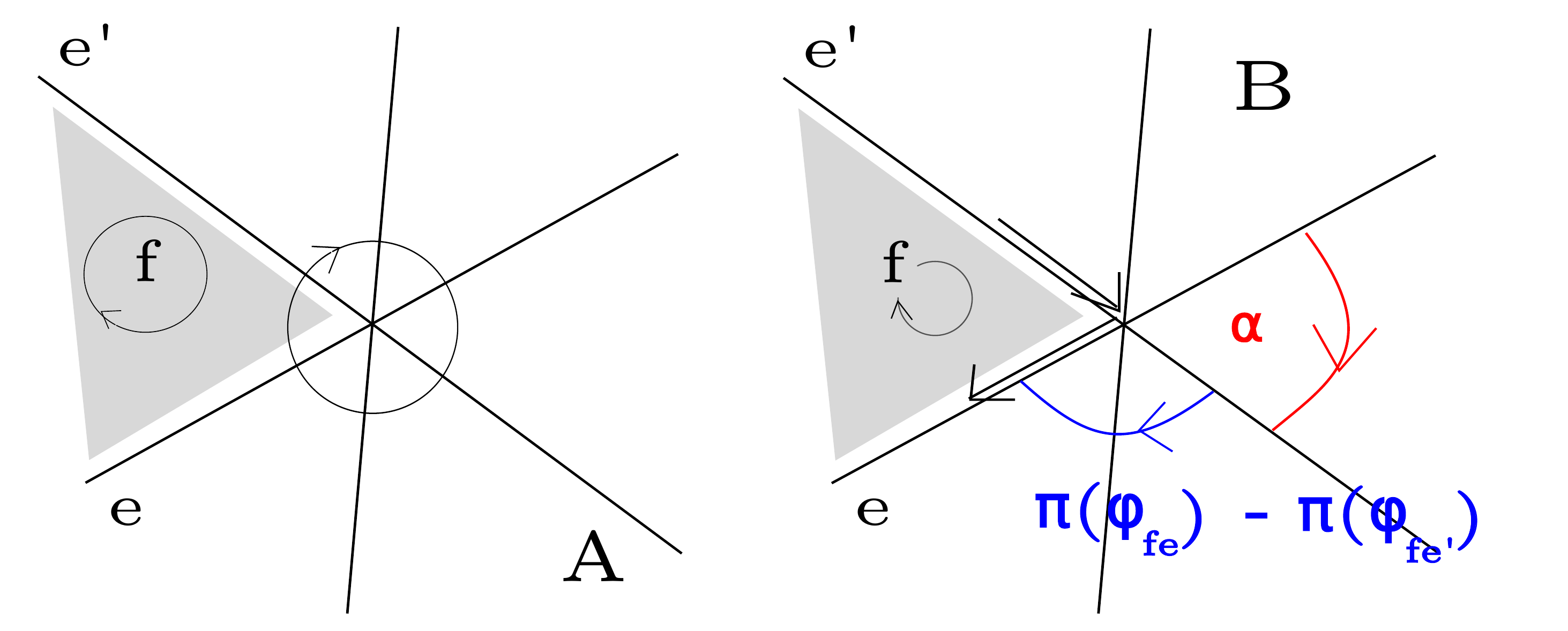}
\caption{(A) example of the cycle with orientations shown, (B)
$\pi(\phi_{fe}')-\pi(\phi_{fe'}')$ and $\alpha=\pi-(\pi(\phi_{fe}')-\pi(\phi_{fe'}'))$}
\label{pic-def2}
\end{center}
\end{figure}

Using prescription \ref{modifier-sec} for $\phi_{fe}-\phi_{fe'}$, the fact that $SU(2)$ is
the double cover of $SO(3)$ and continuity of the deformation we obtain (modulo $2\pi$)
\begin{equation}
 \sum_f \phi_{fe}'-\phi_{fe'}'=\sum_f
\frac{\pi(\phi_{fe}')-\pi(\phi_{fe'}')}{2} -\pi= \left(-n+\frac{n}{2}
+1\right)\pi \quad .
\end{equation}
Thus
\begin{equation}
\prod_{\{ff'\}} {O'}_{fe}{O'}_{f'e}^{-1}
\prod_e \begin{pmatrix}
 0 & -1\\ 1 & 0
\end{pmatrix}
=\begin{pmatrix}
 0 & 1\\ -1 & 0
\end{pmatrix}^{n-2}\begin{pmatrix}
 0 & -1\\ 1 & 0
\end{pmatrix}^{n}=-1 \quad .
\end{equation}
To sum up, we obtained for a given cycle $c$
\begin{equation}\label{eq-cycle}
 \sum_e c(e)\tilde{s}_e= C_e+1  \ \text{mod}\ 2 \quad .
\end{equation}
Since the cycle is oriented in the same way as the faces, $C_e$ is the number of edges oriented according
to the cycle.

\subsubsection{Other nethod of computation}
\label{sec:sign-other}

Let us consider an arbitrary cycle $c$. Let us draw it on the graph $G$ as in figure
\ref{pic-evaluation}. We will denote by $s(c)$ the number of crossings in the cycle. For any node (face)
$f$ we also denote
\begin{equation}
 f(c)=\left\{\begin{array}{ll}
              0 & \ \text{if the middle leg edge of $f$ does not belong to}\ c\\
              1 & \ \text{if the middle leg edge of $f$ belongs to}\ c
             \end{array}\right.
\end{equation}
In the following, we will present another method of how to compute $\sum_e
c_v(e)\tilde{s}_e$ for a basic cycle $c$.
First we will prove:
\begin{lm}
For a single loop cycle $c$ in a spherical network, the quantity
\begin{equation}
 C_e+\sum_f f(c)+s(e)\ \text{mod}\ 2
\end{equation}
does not depend on the choice of a graph $G$.
\end{lm}

\begin{proof}
 Any two graphs can be transformed into one another by a sequence of basic moves
\begin{itemize}
\item One of the Reidemeister moves \cite{Reidemeister,Reidemeister2} for the edge
(see figure \ref{pic-move1} for example). It only changes $s(c)$ by an even number.
\begin{figure}[ht]
\begin{center}
\includegraphics[scale=0.5]{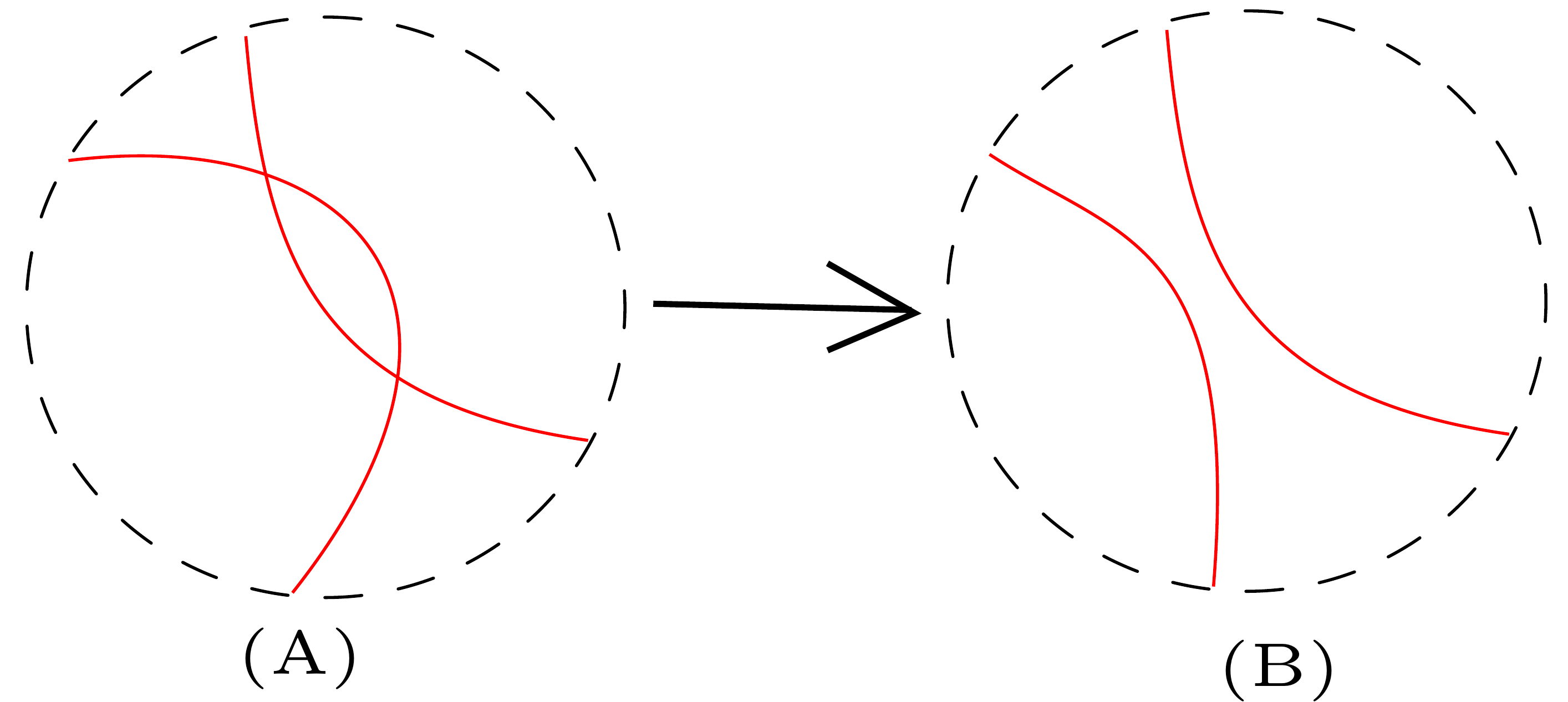}%
\caption{Part of the graph changed by the move. Example of a Reidemeister move.}
\label{pic-move1}
\end{center}
\end{figure}
 \item Transposition of two consecutive nodes belonging to the cycle (see figure \ref{pic-move2}).
In the move shown in the figure
\begin{equation}
 C_e'=C_e\pm 1,\quad  s(c)'=s(c)+3 \quad ,
\end{equation}
and all $f(c)$ remain unchanged.
\begin{figure}[ht]
\begin{center}
\includegraphics[scale=0.5]{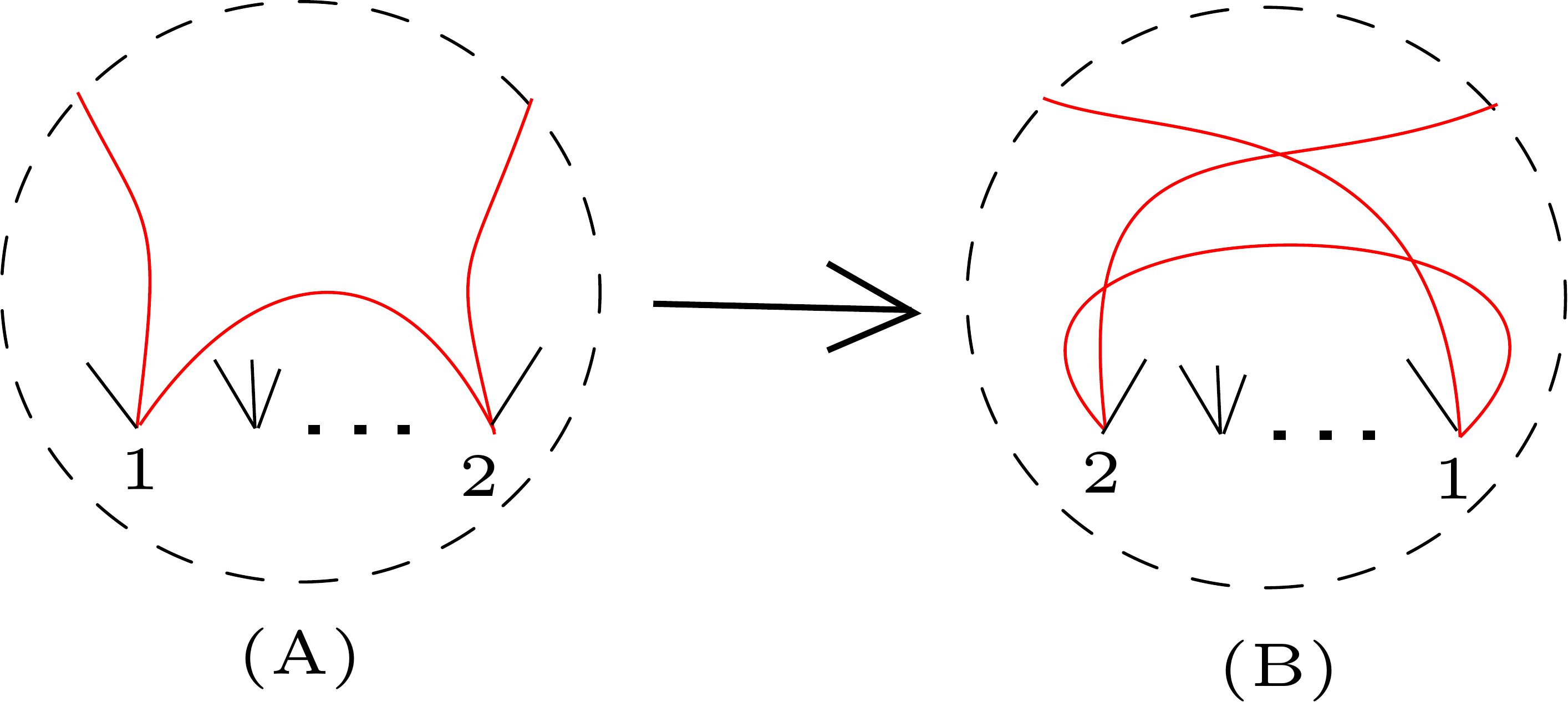}%
\caption{Part of the graph changed by the move (edges not belonging to the cycle are not drawn). Two
consecutive nodes in the cycle transposed.}
\label{pic-move2}
\end{center}
\end{figure}
\item Cyclic permutation of the legs of a node $f$ (figure \ref{pic-move3}). In this case
\begin{equation}
 f(c)+s(c)
\end{equation}
is preserved.
\begin{figure}[ht]
\begin{center}
\includegraphics[scale=0.5]{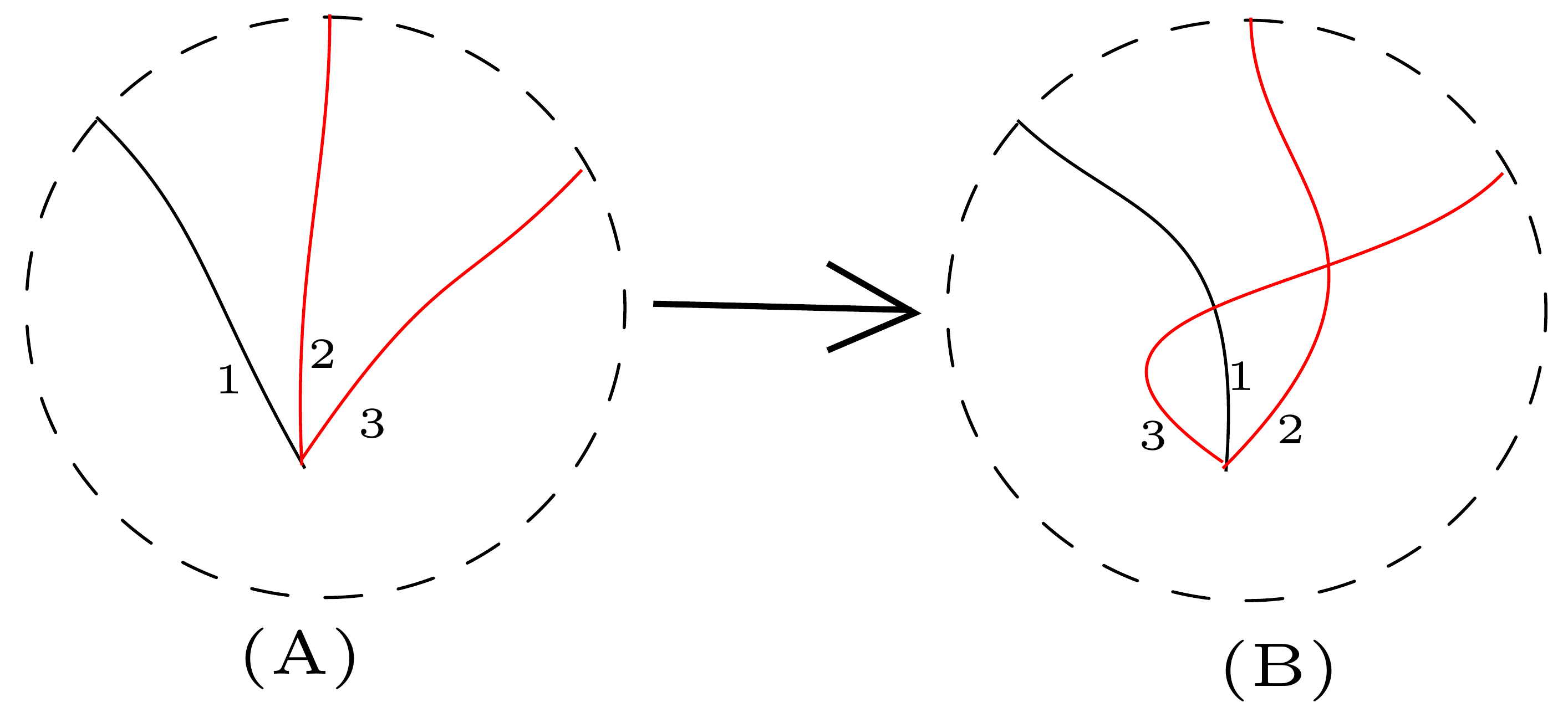}%
\caption{Part of the graph changed by the move. Cyclic change of the order of legs.}
\label{pic-move3}
\end{center}
\end{figure}
\end{itemize}
Thus $C_e+\sum_f f(c)+s(c)\ \text{mod}\ 2$ is invariant.
\end{proof}

We see that $(C_e+1)+\sum_f f(c)+s(c)$ does not depend on the chosen graph $G$, hence, we can
choose
the most convenient one (see figure \ref{pic-choice-sign}).
\begin{figure}[ht]
\begin{center}
\includegraphics[scale=0.3]{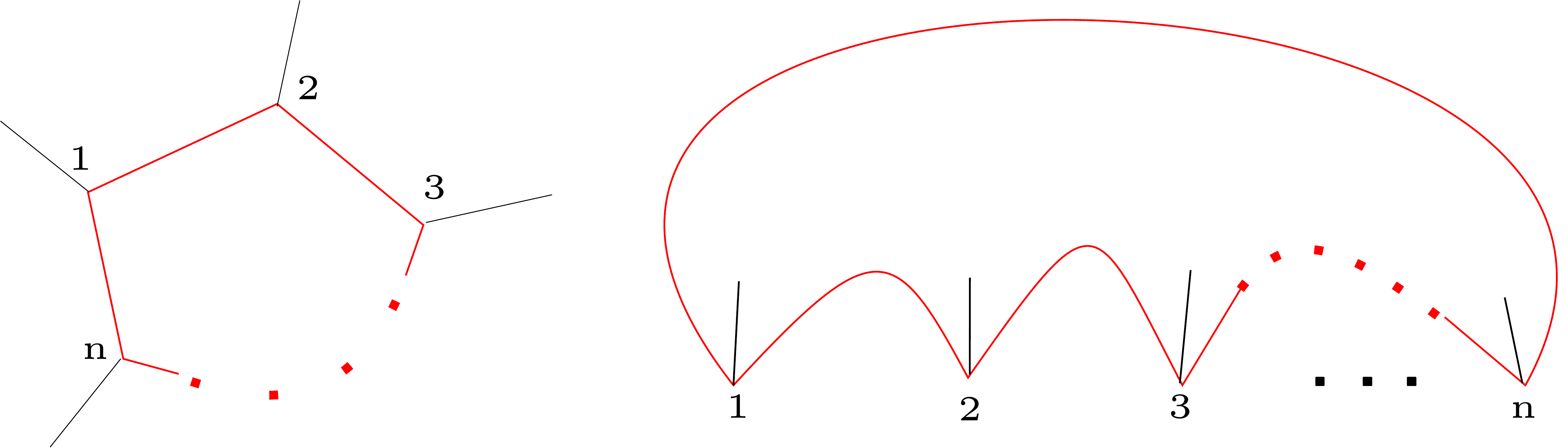}%
\caption{Convenient choice for the graph $G$}
\label{pic-choice-sign}
\end{center}
\end{figure}
For this particular choice
\begin{equation}
 C_e=1,\quad \forall_{f\subset c} f(c)= 0,\quad s(c)=0 \quad ,
\end{equation}
and thus $(C_e+1)+\sum_f f(c)+s(c)=(1+1)+0=0$ mod $2$ and
\begin{equation}
 \sum_e c(e)\tilde{s}_e=C_e+1 =\sum_f f(c)+s(c)\ \text{mod}\ 2\ .
\end{equation}

\subsubsection{Sign of the general cycle in spherical case}
\label{sign-general}

Let us state now a few properties of $f(c)$ and $s(c)$ useful in the sequel.

For two cycles $c$ and $c'$ we denote a cycle by $c+c'$ if it satisfies the following property:
\begin{equation}
 \forall_e,\ (c+c')(e)=c(e)+c'(e)\ \text{mod}\ 2\ .
\end{equation}
We have for two disjoint cycles $c$ and $c'$
\begin{align}
 s(c+c')&=s(c)+s(c')\ \text{mod}\ 2 \quad ,\\
 \forall_f \ f(c+c')&=f(c)+f(c')\ \text{mod}\ 2\ \quad .
\end{align}
We can now write every cycle $c$ in the spherical case as a sum of  disjoint single-loop
cycles
$c_\alpha$, such that $c=\sum_\alpha c_\alpha$:
\begin{equation}
\begin{split}
 \sum_e c(e)\tilde{s}_e&=\sum_e \left(\sum_\alpha c_\alpha(e)\right)\tilde{s}_e=
\sum_\alpha \left(\sum_f f(c_\alpha)+s(c_\alpha)\right)=\\
&=\sum_f f\left(\sum_\alpha c_\alpha\right) +s\left(\sum_\alpha
c_\alpha\right)=
\sum_f f(c)+s(c)
\ \text{mod}\ 2\ .
\end{split}
\end{equation}

\subsubsection{Final sign formula}

Let us notice that in the case when $c$ is the cycle of all half-integer spins we have
\begin{equation}
 s=s(c),\quad \forall_f s_f=f(c),\quad \sum s_e=\sum_e c(e)\tilde{s}_e\ ,
\end{equation}
since if we denote the spin of the middle leg edge of $f$ by $j_{f2}$ then $f(c)=2j_{f2}$ mod $2$.
Finally
\begin{equation}
s+\sum_f s_f+\sum s_e=2\left(\sum_f f(c)+s(c)\right)=0
\ \text{mod}\ 2\ .
\end{equation}

\section{Changes of variables and their Jacobians}
\label{var-change}

Let the Lie group $G$ act transitively on the manifold $S$ and let
\begin{equation}
 \chi\colon G\rightarrow {\mathbb R}
\end{equation}
be a homomorphism. There exists at most one measure $\mu$ (up to scaling) on $S$ such that
\begin{equation}
 g^*\mu=\chi(g)\mu \quad .
\end{equation}
Let
\begin{equation}
 H \circlearrowright S_1\rightarrow S_2 \quad ,
\end{equation}
where $S_1$ is a principal Lie group bundle with the structure group $H$ and the base space
$S_2$. Any (pseudo-)$k$-form $\mu_2$ on $S_2$ can be uniquely
represented by a (pseudo-) $k$-form $\mu_1$ on $S_1$ that satisfies
\begin{align}
 \label{con1mu}&h^*\mu_1=\mu_1\quad \forall h\in H\\
\label{con2mu}&\mu_1\perp\partial_\xi=0\quad \forall\ \partial_\xi\in \mathfrak {h} \quad ,
\end{align}
where ${\mathfrak h}$ is the Lie algebra of $H$ and $\perp$ is contraction of the (pseudo-)
form with the vector on the first site. Any form $\mu_1$ determines the form $\mu_2$ on $S_2$.
The integration over $S_2$ is the integration over any section of the projection map $S_1\rightarrow S_2$.

Such a form satisfying conditions \eqref{con1mu} and \eqref{con2mu} can be obtained from the $H$
invariant form $\mu$ on $S_1$ via the formula
\begin{equation}
 \mu_2= \mu \perp \bigwedge_{\xi \text{ basis } {\mathfrak h}} \partial_\xi \quad .
\end{equation}
In case of a compact group $H$ it is related to the measure obtained by integration over the fibers, called $\mu_{\int H}$,  as follows
\begin{equation}\label{norm-Haar}
 \mu_2=(\mu_H\perp\bigwedge \partial_\xi)\mu_{\int H} \quad ,
\end{equation}
where $\mu_{H}$ is the normalized Haar measure on $H$.

Let $M\subset S$ be a submanifold described locally by a set of independent equations $f_a$. For
any measure (form) $\mu$ on $S$ we can define a measure (form) $\mu_{f_a}$ on $M$ by the following
integration prescription: Let $g\in C^0(M)$ and $\tilde{g}$ be any continuous extension to
$S$, then
\begin{equation}
 \int_M\mu_{f_a} g=\int_M \prod \delta(f_a) \tilde{g}\mu \quad .
\end{equation}

Let $M$ be a section of the bundle $H\subset S_1\rightarrow S_2$ described by equations $f_a$, then
we can compare the just described measures on $M$ and $S_2$ since $M\rightarrow S_2$ is a diffeomorphism of $M$ onto $S_2$:
\begin{equation}
 \mu_{f_a}=\left(\det \partial_{\xi_i}f_a\right)^{-1}\mu\perp \bigwedge\partial_\xi \quad .
\end{equation}
Indeed, we can choose local coordinates such that $S_1=M\times H$ and the zero section is
described by $f_a=0$.
We have
\begin{equation}
 \mu=(\mu\perp \bigwedge \partial_\xi)\wedge \bigwedge d\xi \quad .
\end{equation}
By extending the function $g$ constantly along fibers from the zero section, we obtain:
\begin{equation}
\begin{split}
 \int_{M}g\mu_{f_a}=\int_{S_1} \prod \delta(f_a)\tilde{g}\mu&=
\int_{M} \tilde{g}\left(\int_H \prod \delta(f_a)\bigwedge d\xi\right)(\mu\perp \bigwedge
\partial_\xi)\\
&=\int_{S_2} g\left(\det \partial_{\xi_i}f_a\right)^{-1} (\mu\perp \bigwedge \partial_\xi) \quad .
\end{split}
\end{equation}

\subsection{Change of variables $u_i\rightarrow n_i$ (integrating out gauge)}
\label{variables-n}

Let us remind from section \ref{sec:new-form-action1} that
\begin{equation}
 S^2=SU(2)/S^1 \quad ,
\end{equation}
given by the right action of $S^1$ on $SU(2)$. The sphere $S^2$ can either be
represented by unit vectors $|n_i|^2=1$ or traceless $2 \times 2$ matrices $n_i$ with the
condition
\begin{equation}
 \frac{1}{2}\Tr n_in_i=1 \quad .
\end{equation}
Then the quotient map is given by
\begin{equation}
 n_i(u)=u H u^{-1} \quad ,
\end{equation}
where $H= \left(
\begin{matrix}
0 & i \\
-i & 0
\end{matrix} \right)$. The group $SU(2)$ acts on $S^2$ (as a left action on the quotient) by
\begin{equation}
 n_i\rightarrow u\,n_i\, u^{-1} \quad .
\end{equation}
The Haar measure from $SU(2)$ can be integrated over the fibers giving the invariant measure $\mu$
on the sphere with total volume $1$.

Another invariant measure is
\begin{equation}
 \delta(|n|^2-1)\rd n_1 \rd n_2\rd n_3 \quad .
\end{equation}
Since there is only one invariant measure up to scale, both are related by a scaling
transformation:
\begin{equation}
 \mu=c\delta(|n|^2-1)\rd n_1 \rd n_2\rd n_3 \quad .
\end{equation}
The constant is fixed by requiring:
\begin{equation}
 1\overset{!}{=}\int_{S^2}\mu=c\int_0^{2\pi}\rd\phi\int_0^\pi\sin\theta\rd\theta\int_0^\infty
r^2\delta(r^2-1)\rd r=2\pi c \quad ,
\end{equation}
thus
\begin{equation}
 \mu=\frac{1}{2\pi}\delta(|n|^2-1)\rd n_1 \rd n_2\rd n_3 \quad .
\end{equation}

\subsection{Variables $\theta$ in the flat tetrahedron}
\label{variables-theta}
Let us consider two sets of variables
\begin{equation}
 N=(\vec{n}_1,\ldots,\vec{n}_{m+1}) \quad ,
\end{equation}
where $\vec{n_i}$ are $m$ vectors with exactly one dependency, i.e. every subset of $m$
vectors forms a basis. Let
\begin{equation}
 M=N^TN,\ m_{ij}=\vec{n}_i\vec{n}_j,\ i\leq j\ ,
\end{equation}
where $M$ is a symmetric positive $(m+1)\times(m+1)$  matrix, which is degenerate with
exactly one null eigenvector, whose entries are all non vanishing.

On $N$ there exists a left action of $O(m)$, $\vec{n}_i\rightarrow
O\vec{n}_i$. The matrix $M$ is $O(m)$ invariant, so the parameters of this action can be
regarded as supplementary to $M$. The vector fields of
this action will be denoted by $L_{ab}$.
The map $N\rightarrow M$ is an $O(m)$ principal bundle.

In the following, our goal is to compare the pseudo-form
\begin{equation}
 \mu_1=\left|\bigwedge_{i=1...m+1,a=1..m} \rd n_i^a\perp\bigwedge_{a<b}
L_{ab}\right|
\end{equation}
with the form
\begin{equation}
 \mu_2=\left|\delta(\det M)\ \bigwedge_{i\leq j} \rd m_{ij}\right| \quad .
\end{equation}
Let us notice that both $\mu_1$ and $\mu_2$ are  measures on $M$.

There are additional transformations parametrized by $U\in GL(m+1)$
\begin{equation}
 \tilde{N}=NU,\ \tilde{M}=U^{T}MU\ ,
\end{equation}
which commute with the $O(m)$ action on $N$.
The measure $\mu_1$ is $\chi$ covariant with respect to this
action, where
\begin{equation}
 \chi(U)=|{\det}^{m}|(U):=|\det(U)|^m\quad \forall U\in GL(m+1)
\end{equation}
Furthermore, we have
\begin{itemize}
 \item $\mu_2$ is invariant for $U\in O(m+1)$,
\item for transformations of the form $U={\rm diag}\ {\lambda_i}$ (diagonal matrix) the
measure $\mu_2$ transforms as
\begin{equation}
 \Big|\underbrace{\frac{1}{\prod_i\lambda_i}}_{\text{from }\det M}
\underbrace{\prod_{i\leq
j}\lambda_i\lambda_j}_{=\prod_i\lambda_i^{m+1}}\Big|\mu_2 \quad ,
\end{equation}
\end{itemize}
so it is also $\chi$ covariant as rotations and scaling generate the whole
group.

Hence the two measures $\mu_1$ and $\mu_2$ differ by a constant $c$ as $GL(m+1)$ acts transitively
on $M$.
This constant can be computed for a special value of $N$:
\begin{equation}
 {n}_i^a=\left\{\begin{array}{ll}
                       \delta_{i}^a, & i\leq m\\ 0, & i=m+1
                      \end{array}\right. \quad .
\end{equation}
In this choice
\begin{equation}
 \rd {m}_{ij}=\left\{\begin{array}{ll}
                       \rd{n}_{i}^j+\rd{n}_{j}^i, & i,j\leq m\\
\rd{n}_{m+1}^i, & j=m+1\\
0, & i=j=m+1
                      \end{array}\right. \quad .
\end{equation}
Moreover
\begin{equation}
 \rd {n}_i^a\perp
L_{cd}=(L_{cd}\vec{{n}}_i)^a=\delta_{ac}\delta_{id}- \delta_{ad}\delta_{ic} \quad .
\end{equation}
The following equalities hold:
\begin{equation}
 2^m\left|\bigwedge_{i=1...m+1,a=1..m} \rd n_i^a\right|=\left|\bigwedge_{1\leq j\leq i\leq m+1,
j\not=m+1}\ \underbrace{dn_i^j+dn_j^i}_{dm_{ij}}\wedge\bigwedge_{1\leq j< i\leq m+1} dn_i^j\right| \quad ,
\end{equation}
but since $\rd m_{ij}\perp L_{ab}=0$:
\begin{equation}
 2^m\left|\bigwedge_{i=1...m+1,a=1..m} \rd n_i^a\perp\bigwedge_{a<b}
L_{ab}\right|=\left|\bigwedge_{1\leq j\leq i\leq m+1,
j\not=m+1}\ dm_{ij}\right|=\delta({m}_{m+1,m+1})\left|\bigwedge_{(i,j)\colon i\leq j} \rd
{m}_{ij}\right| \quad .
\end{equation}
Moreover in this case
\begin{equation}
 {M}_{ij}=\left\{\begin{array}{ll}
                       \delta_{ij}, & i\leq m\\ 0, & i=m+1
                      \end{array}\right.
\end{equation}
and so $\frac{\partial \det{M}}{\partial {m}_{m+1,m+1}}=1$. Eventually
\begin{equation}
2^m\left| \bigwedge_{i=1...m+1,a=1..m} \rd n_i^a\perp\bigwedge_{a<b}
L_{ab}\right|=\delta(\det M)\ \left|\bigwedge_{i\leq j} \rd m_{ij}\right| \quad .
\end{equation}

\subsubsection{Integration over the $SO(3)$ fiber}

In the case $m=3$ we are interested in integrating over the fiber, however not the
whole $O(3)$ but only over
one connected component with respect to $SO(3)$. This is due to the $u$ transformation symmetry
corresponds to $SO(3)$ not $O(3)$ (see also section \ref{sec:transformations}).

In this section we continue to compute the correct constant in front of the measure. We will
now consider a fibration with the group $SO(m)$ that can still be described locally by a projection $N\mapsto M$. This is, however, enough because we are only interested in
local variables.

In general we have
\cite{MacDonald}
\begin{equation}
 \int_{SO(m)}\left|\bigwedge_{a<b}
L^*_{ab}\right|=\prod_{k=2}^n\frac{2\pi^\frac{k}{2}}{\Gamma\left(\frac{k}{2}\right)} \quad ,
\end{equation}
so from \eqref{norm-Haar} our measure integrated over the fibre is equal to
\begin{equation}
 \int_{SO(m)}\left|\bigwedge_{i=1...m+1,a=1..m} \rd
n_i^a\right|=\frac{1}{2^m}\prod_{k=2}^{m}\frac{2\pi^\frac{k}{2}}{\Gamma\left(\frac{k}{2}\right)}
\delta(\det M)\
\left|\bigwedge_{i\leq j} \rd
m_{ij}\right| \quad .
\end{equation}
If we impose the condition $|\vec{n}_i|=1$, we integrate over a set of unit vectors. This
implies for $M$ that we have to skip $\rd m_{ii}$ in the measure and we define $m_{ij} = \cos
\theta_{ij}$. In these new angle variables the measure takes the form
\begin{equation}
 \frac{1}{2^m}\prod_{k=2}^m\frac{2\pi^{\frac{k}{2}}}{\Gamma\left(\frac{k}{2}\right)}\delta(\det
\tilde{G})\prod_{i<
j}|\sin\theta_{ij}|\ \bigwedge_{i< j} \left|\rd \theta_{ij}\right|\ ,
\end{equation}
where $\tilde{G}$ is the Gram matrix with the convention
\begin{equation}
 \tilde{G}_{ij}=\cos\theta_{ij},\quad \theta_{ii}=0 \quad.
\end{equation}
In case $n=3$ we have
\begin{equation}
 \pi^2 \delta(\det \tilde{G})\prod_{i<
j}|\sin\theta_{ij}|\ \bigwedge_{i< j} \rd \theta_{ij}\ .
\end{equation}

\subsection{Variables $\theta$/ $l$ in the spherical constantly curved tetrahedron}

Let us consider the spaces of matrices
\begin{equation}
 \tilde{N}=\{N\in M_n({\mathbb R})\colon \det N>0\}
\end{equation}
and
\begin{equation}
 \tilde{M}=\{M\in M_n({\mathbb R})\colon M>0\} \quad .
\end{equation}
We have a fibration with the group $SO(n)$ (via left action on $\tilde{N}$)
\begin{equation}
 \tilde{N}\rightarrow \tilde{M},\quad M=N^TN \quad .
\end{equation}
We can compare forms
\begin{equation}
 \begin{split}
  \mu_1&=\left|\bigwedge \rd n_i^a\perp \bigwedge\partial_\xi\right| \quad , \\
\mu_2&=(\det M)^{-\coh}\left|\bigwedge_{i<j}\rd m_{ij}\right| \quad .
 \end{split}
\end{equation}
As in section \ref{variables-theta} there is an action of $SL(n)$ by
\begin{equation}
 N\rightarrow NU,\quad M\rightarrow U^TMU \quad .
\end{equation}
We can check that both measures are $\chi=|{\det}^{n}|$ covariant. Since
$SL(n)$ acts transitively on matrices with positive determinant, we
have
\begin{equation}
 \mu_1=c\mu_2 \quad .
\end{equation}
Checking for $N={\mathbb I}$ gives $c=1$.

Let us notice that $n_i\cdot n_i=m_{ii}$. On the surface $m_{ii}=1$ we can introduce angle variables
$\cos\theta_{ij}=m_{ij}$ and obtain
\begin{equation}
 (\det\tilde{G})^{-\coh}\prod\sin\theta_{ij}\bigwedge\rd\theta_{ij}
=\pm\prod\delta(n_in_i-1)\bigwedge \rd n_i^a \quad .
\end{equation}

\subsection{Determinant $\det \frac{\partial \theta}{\partial l}$ for constantly curved simplices}
\label{bunch3}

We denote the length Gram matrix by $G$ and the angle Gram matrix by $\tilde{G}$. The
dimension is equal to $n-1$ and we are working in ${\mathbb R}^n$ on the sphere with radius~$1$.

Our goal is to prove the following formulas for the $n-1$-dimensional curved simplex. It was first proposed in \cite{DF} and checked using an algebraic manipulator. Now we are presenting the complete derivation.
\begin{lm}\label{lm:det}
The following formulas hold for a spherical $(n-1)$-simplex:
\begin{equation}
 \det \frac{\partial\theta_{ij}}{\partial l_{km}'}=(-1)^n
\frac{\prod\sin
l_{ij}'}{\prod\sin\theta_{ij}}\left(\frac{\det\tilde{G}}{\det G}\right)^{\frac{n+1}{2}}\ ,
\end{equation}
and for $n=4$
\begin{equation}
 \det \frac{\partial\theta_{ij}}{\partial l_{km}}=-\det \frac{\partial\theta_{ij}}{\partial
l_{km}'}= -\frac{\det \tilde{G}}{\det
G}\ .
\end{equation}
\end{lm}

Where we used standard convention that
the angle $\theta_{ij}$ is the angle on
the hinge obtained by leaving out indices $i$ and $j$. The length of the opposite edge, i.e. the edge connecting vertices $i$ and $j$, is denoted by
$l'_{ij}$ and in $3$D, $l_{ij}$ is the length of the edge at which the angle sits.

\subsubsection{Outline of the proof}

We compute how the measure $\bigwedge d l_{ij}'$ transforms under the the change of
variables
\begin{equation}
 \theta_{ij}\rightarrow l_{ij}'\ .
\end{equation}
In fact, introducing variables $m_{ij}=\cos\theta_{ij}$ and $m_{ij}'=\cos l_{ij}$, we have (in the
right order)
\begin{align}
 \prod \sin\theta_{ij}\bigwedge_{i<j}
d\theta_{ij}&=\prod_i \delta(m_{ii}-1)\bigwedge_{i\leq j} \rd m_{ij} \quad , \\
 \prod \sin l_{ij}'\bigwedge_{i<j}
dl_{ij}'&=\prod_i \delta(m_{ii}'-1)\bigwedge_{i\leq j} \rd m_{ij}' \quad .
\end{align}
Both measures on the left hand side are on
\begin{equation}
 \tilde{M}_1=\{M\in \tilde{M}\colon \forall_i m_{ii}=1\} \quad ,
\end{equation}
where we introduced the notation $\tilde{M}=GL_+(n)$ for simplicity.

\subsubsection{Computation}

There is an action of the group of diagonal matrices
\begin{equation}
 D=\{d\in GL_+(n)\colon d_{ij}=\lambda_i\delta_{ij},\quad \lambda_i>0\}
\end{equation}
on $\tilde{M}$ given by
\begin{equation} \label{eq:diag_trafo}
 M\rightarrow d^TMd \quad .
\end{equation}
A basis for the Lie algebra ${\mathfrak d}$ of the group $D$ is given by
$\partial_{\xi_i}=E_{ii}$ (matrices with only one nonzero entry being the $i$-th element
on the diagonal equal to $1$).

We have a fibration
\begin{equation}
 \tilde{M}\rightarrow \tilde{M}/D
\end{equation}
and let $\tilde{M}_1\subset \tilde{M}$ be a cross section given by the equations
\begin{equation}
\forall_i\ m_{ii}=1\ .
\end{equation}
 Let us introduce maps
\begin{equation}
\begin{split}
 &\psi_1\colon \tilde{M}_1\rightarrow \tilde{M}/D,\quad M\, {\mapsto}\, [M]\quad ,\\
&\psi_2\colon \tilde{M}\rightarrow \tilde{M},\quad M\,{\mapsto}\, M^{-1} \quad .
\end{split}
\end{equation}
Acting with $\psi_2$ on matrix transformed as in \eqref{eq:diag_trafo}, we have
\begin{equation}
 \psi_2(d^TMd)=(d^{-1})^T\psi_2(M)(d^{-1}) \quad ,
\end{equation}
such that there is a map
\begin{equation}
 [\psi_2]\colon \tilde{M}/D\rightarrow \tilde{M}/D \quad .
\end{equation}
Let us notice that the composition $\psi:=\psi_1^{-1}[\psi_2]\psi_1$ transforms $\tilde{M}_1$ into
$\tilde{M}_1$.

We define measures
\begin{equation}
 \begin{split}
  \mu_1&=\prod\delta(m_{ii}-1)\bigwedge_{i\leq j} \rd m_{ij}=\bigwedge_{i<j} \rd m_{ij} \quad ,\\
   \mu&=(\det M)^{-\frac{n+1}{2}}\bigwedge_{i\leq j} \rd m_{ij} \quad ,\\
\mu_{\tilde{M}/D}&=\mu\perp \bigwedge_i\partial_{\xi_i} \quad ,
 \end{split}
\end{equation}
where $\partial_{\xi_i}$ is the basis of the Lie algebra ${\mathfrak d}$
\begin{equation}
 \partial_{\xi_i}m_{kl}=(\delta_{ik}+ \delta_{il})m_{kl}\quad .
\end{equation}
Let us notice that according to section \ref{variables-theta} $\mu$ is $SL(n)$ invariant
(where it acts as $D$) thus the pullback
\begin{equation}
 \psi_2^*\mu=c\mu \quad ,
\end{equation}
since $SL(n)$ acts transitively on $\tilde{M}$. We can check that $c=1$ by computing the measures
for $M=M^{-1}={\mathbb
I}$.

We have $\psi_2\partial_{\xi_i}=-\partial_{\xi_i}$ so
\begin{equation}
 [\psi_2]^*\mu_{\tilde{M}/D}=(-1)^n\mu_{\tilde{M}/D} \quad .
\end{equation}
From basic facts explained in the appendix \ref{var-change} we know that (in the right order)
\begin{equation}
 \psi_1^*\mu_{\tilde{M}/D}=\left[\det \partial_{\xi_i}(m_{jj}-1)\right](\det
M)^{\frac{n+1}{2}}\mu_1
=2^n(\det M)^{\frac{n+1}{2}}\mu_1 \quad .
\end{equation}
Combining all transformations we obtain
\begin{equation}
 \mu_1=(-1)^n \left(\frac{\det M}{\det \psi(M)}\right)^{\frac{n+1}{2}}\psi^*\mu_1 \quad .
\end{equation}
Finally, we obtain
\begin{equation}
 \bigwedge d\theta_{ij}=(-1)^n\frac{\prod\sin
l_{ij}'}{\prod\sin\theta_{ij}}\left(\frac{\det\tilde{G}}{\det G}\right)^{\frac{n+1}{2}}\bigwedge
dl_{ij}' \quad .
\end{equation}
So eventually
\begin{equation}
 \det \frac{\partial\theta_{ij}}{\partial l_{km}'}=(-1)^n
\frac{\prod\sin
l_{ij}}{\prod\sin\theta_{ij}}\left(\frac{\det\tilde{G}}{\det G}\right)^{\frac{n+1}{2}}\ .
\end{equation}

\subsubsection{Further simplifications for $n=4$}

We can simplify the above formula using equalities from \cite{Kokkendorff}:
\begin{equation}
 \det\tilde{G}=\frac{(\det G)^{n-1}}{\prod G^*_{ii}},\ (\sin\theta_{ij})^2=\frac{\det G\det
G_{(ij)}}{G^*_{ii}G^*_{jj}}\ ,
\end{equation}
where $G^*_{ii}$ is the $ii$-element of the minor matrix, $G_{(ij)}$ is the $G$ matrix without $i$th and $j$th rows and columns. Eventually we obtain
\begin{equation}
 \det \frac{\partial\theta_{ij}}{\partial l_{km}'}=(-1)^n
\frac{\prod\sin l_{ij}'\prod (G_{kk}^*)^{\frac{n-1}{2}}}{(\det G)^{\frac{n(n-1)}{4}}\prod \sqrt{\det
G_{(ij)}}}
\frac{(\det G)^{\frac{(n-1)(n+1)}{2}}}{\prod (G^*_{kk})^{\frac{n+1}{2}}}\frac{1}{(\det
G)^{\frac{n+1}{2}}}\ .
\end{equation}
After simplification it is equal to
\begin{equation}
 \det \frac{\partial\theta_{ij}}{\partial l_{km}'}=
\underbrace{\left(\prod\frac{\sin l_{ij}'}{\sqrt{\det
G_{(ij)}}}\right)}_{=1} \frac{\det \tilde{G}}{\det G}\ .
\end{equation}
Note that this simplification only holds for $n=4$. Since
\begin{equation}
 \det \frac{\partial l_{ij}}{\partial l_{km}'}=-1
\end{equation}
we obtain
\begin{equation}
  \det \frac{\partial\theta_{ij}}{\partial l_{km}}= -\det \frac{\partial\theta_{ij}}{\partial
l_{km}'} \quad .
\end{equation}

\section{Technical computations of determinants} \label{app:technical}

In this section we will prove several technical results.

Let us introduce the notation
\begin{equation}
 {\det}' M=\sum_i M^*_{ii}\ .
\end{equation}
It is an invariant of the matrix and, moreover, in the case when the matrix is symmetric and has one
null eigenvector, it is the determinant of the matrix restricted to the space perpendicular to that
null eigenvector.

Let us also remind some general facts
\begin{align}
&l_{ij}\frac{\partial\theta_{ij}}{\partial
l_{kl}}=l_{kl}\frac{\partial\theta_{ij}}{\partial
l_{kl}}=0\quad , \quad
\frac{\partial\theta_{ij}}{\partial
l_{kl}}=\frac{\partial\theta_{kl}}{\partial
l_{ij}} \quad ,\\
&l_{ij}=\lambda\frac{\partial\det\tilde{G}}{\partial\theta_{ij}}\quad ,\quad \lambda=- \frac{2^2 \prod
S_i^2}{3^5 V^5} \quad ,\\
&{\det}'\tilde{G}=\frac{3^4}{2^2}(\sum_i S_i^2)\frac{
V^4}{\prod S_i^2} \quad ,
\end{align}
which are proven for completeness in appendix \ref{bunch4}. Our results are (see appendix \ref{bunch2}):

\begin{lm}
For the flat tetrahedron holds
\begin{equation}
 {\det}'\frac{\partial\theta_{ij}}{\partial
l_{kl}}= \frac{3^3}{2^5} \frac{|l|^2}{\prod S_i^2} V^3 \quad .
\end{equation}
\end{lm}
Moreover, in appendix \ref{bunch} we prove:
\begin{lm}\label{lm:bunch}
Let
 \begin{equation}
 \lambda=-\frac{2^2 \prod S_i^2}{3^5 V^5} \quad .
\end{equation}
For the flat tetrahedron holds
\begin{align}
&\frac{\partial\lambda}{\partial l_{ij}}
\frac{\partial\det\tilde{G}}{\partial\theta_{kl}}+
\frac{\partial\theta_{ij}}{\partial
l_{mn}}\lambda\frac{\partial^2\det\tilde{G}}{\partial\theta_{mn}\partial\theta_{kl}}=
\delta_{(ij),(kl)} \quad ,\\
&\exists_c\
\frac{\partial\det\tilde{G}}{\partial\theta_{kl}}c+\lambda\frac{\partial\lambda}{\partial
l_{ij}}\frac{\partial^2\det\tilde{G}}{\partial\theta_{ij}\partial\theta_{kl}}=0 \quad ,\\
&\frac{\partial\lambda}{\partial
l_{mn}}\frac{\partial\det\tilde{G}}{\partial\theta_{mn}}=1 \quad .
\end{align}
\end{lm}

\subsection{General knowledge}
\label{bunch4}

We know that $\det\tilde{G}=0$ for a geometric set of $\theta$'s. Moreover, the null eigenvector is
given by
\begin{equation}
 (S_1,S_2,S_3,S_4)\quad ,
\end{equation}
where $S_i$ denote the areas of the triangles of the tetrahedron. The computation of
${\det}'\tilde{G}$ can be found in appendix \ref{bunch41}.
We have
\begin{equation}
 \frac{\partial\det\tilde{G}}{\partial\theta_{ij}}=-2{\det}'\tilde{G}\frac{S_iS_j\sin\theta_{ij}}{
\sum
_k S_k^2}=- \frac{3^5}{2^2}(\sum_k S_k^2)\frac{
V^4}{\prod S_k^2} \frac{Vl_{ij}}{\sum
_k S_k^2}=- \frac{3^5}{2^2}\frac{V^5}{\prod S_k^2}\ l_{ij}=\lambda^{-1} l_{ij} \quad .
\end{equation}
In addition to that, we also have
\begin{equation}
 0=\frac{\partial\det\tilde{G}}{\partial l_{kl}}=\frac{\partial\det\tilde{G}}{\partial
\theta_{ij}}\frac{\partial\theta_{ij}}{\partial l_{kl}}=\lambda
l_{ij}\frac{\partial\theta_{ij}}{\partial l_{kl}} \quad .
\end{equation}
We know that $\theta$ has scaling dimension $0$, thus
\begin{equation}
 l_{kl}\frac{\partial\theta_{ij}}{\partial l_{kl}}=0 \quad .
\end{equation}

\subsubsection{Expressing $\frac{\partial\theta_{ij}}{\partial l_{kl}}$ in terms of $l_{ij}$}
\label{app-theta-l}

Here we recall several well-known facts for flat simplices of arbitrary dimension using the notation of $l_{ij}'$ from appendix \ref{bunch3}, see also \cite{Sorkin,DF} for more details.
Let $M$ be the following matrix
\begin{equation}
 M=\left[\begin{array}{cccc}
          0 & 1&\ldots& 1\\
          1 & {l'}_{11}^2& \ldots & {l'}_{1n}^2\\
          \vdots & \vdots &\ddots &\vdots\\
          1 & {l'}_{n1}^2& \ldots & {l'}_{nn}^2
         \end{array}\right] \quad ,
\end{equation}
where ${l'}_{11} = \hdots = {l'}_{nn} = 0$ and ${l'}_{ij} = {l'}_{ji}$. Then we have:
\begin{align}
V^2&=\frac{(-1)^{n-1}}{2^{n}(n-1)!^2}\det M \quad ,\quad
S_i^2=\frac{(-1)^{n-2}}{2^{n-1}(n-2)!^2}\ M_{ii}^* \quad ,\\
 \cos\theta_{ij}&=\frac{M_{ij}^*}{\sqrt{M_{ii}^*M_{jj}^*}} \quad .
\end{align}
In three dimension we also have
\begin{equation}
 \sin^2\theta_{ij}= \left(\frac{3}{2}\right)^2 \frac{V^2 l_{ij}^2}{S_i^2S_j^2} \quad ,
\end{equation}
so in the case $\theta_{ij}\in(0,\pi)$ we can write:
\begin{equation}
 \frac{\partial\theta_{ij}}{\partial
l_{kl}}=-\frac{1}{\sin\theta_{ij}}\frac{\partial\cos\theta_{ij}}{\partial l_{kl}}=
- \frac{2}{3}\frac{S_iS_j}{V l_{ij}}\frac{\partial}{\partial l_{kl}}\frac{M_{ij}^*}{\sqrt{M_{ii}^*M_{jj}^*}} \quad .
\end{equation}
This, in principle, allows us to compute $\frac{\partial\theta_{ij}}{\partial
l_{kl}}$ and all other derivatives in terms of lengths.

\subsubsection{Computation of ${\det}'\tilde{G}$}
\label{bunch41}

Let us start with the spherical case, i.e. a tetrahedron with constant non vanishing (positive)
curvature -- a curved tetrahedron on the unit sphere. In this case
we define $\tilde{l}_{ij}:=\epsilon l_{ij}$ and
$\theta^\epsilon_{ij}:=\theta(\epsilon l_{kl})$ and take the limit
$\epsilon\rightarrow 0$ in order to reobtain the flat case. The angles ($\theta_{ij}^\epsilon$) have a limit as the angles of
flat tetrahedron ($\theta_{ij}$) with lengths $l_{ij}$.

First, let us notice that
\begin{align}
 \det G&=\det \left(\begin{array}{ccc}
               1 & 0 &\cdots\\
               1 & G &\cdots\\
               \vdots &\cdots &\cdots
              \end{array}
\right)=\det \left(\begin{array}{ccc}
               1 & 0 &\cdots\\
               1 & 1-\frac{1}{2}\epsilon^2 l_{ij}^2 + O(\epsilon^4) &\cdots\\
               \vdots &\cdots &\cdots
              \end{array}
\right)\\&=\frac{1}{8}\epsilon^{6} \det \underbrace{\left(\begin{array}{
ccc }
               0 & 1 &\cdots\\
               1 & l_{ij}^2 &\cdots\\
               \vdots &\cdots &\cdots
              \end{array}
\right)}_{C}+O(\epsilon^8) \quad .
\end{align}
We can compute ${\det}'\tilde{G}=\sum_i \tilde{G}^*_{ii}$ using the following identity from
\cite{Kokkendorff} ($n=4$, i.e. $D=3$):
\begin{equation}
 \frac{\tilde{G}_{ii}^*}{G_{ii}^*}=\frac{(\det G)^{n-2}}{\prod G_{ii}^*} \quad ,
\end{equation}
obtaining
\begin{equation}
{\det}'\tilde{G}=\sum_i \tilde{G}^*_{ii}=\left(\sum_i G_{ii}^*\right)\frac{(\det
G)^{n-2}}{\prod G_{ii}^*}
=\frac{3^4}{2^2}\left(\sum_i S_i^2\right)\frac{
V^4}{\prod S_i^2}+O(\epsilon^2) \quad .
\end{equation}
So in the flat case
\begin{equation}
 {\det}'\tilde{G}=\frac{3^4}{2^2}\left(\sum_i S_i^2\right)\frac{
V^4}{\prod S_i^2} \quad .
\end{equation}

\subsection{ Collection of results}
\label{bunch}

Let us prove the following useful formulas:
 \begin{equation}
  \frac{\partial\lambda}{\partial l_{ij}}
\frac{\partial\det\tilde{G}}{\partial\theta_{kl}}+\sum_{m<n}
\frac{\partial\theta_{ij}}{\partial
l_{mn}}\lambda\frac{\partial^2\det\tilde{G}}{\partial\theta_{mn}\partial\theta_{kl}}=\frac{\partial}{\partial
l_{ij}}\left(\lambda\frac{\partial\det\tilde{G}}{\partial\theta_{kl}}\right)=\delta_{(ij),(kl)} \quad ,
 \end{equation}
since we know, due to the Schl\"afli identity, that
$\frac{\partial\theta_{ij}}{\partial l_{kl}}=\frac{\partial\theta_{kl}}{\partial l_{ij}}$.
This also implies that
\begin{equation}
 \sum_{m<n}
\frac{\partial\theta_{mn}}{\partial
l_{ij}}\lambda\frac{\partial^2\det\tilde{G}}{\partial\theta_{mn}\partial\theta_{kl}}=
\delta_{(ij),(kl)}-\frac{\partial\lambda}{\partial l_{ij}}\frac{l_{kl}}{\lambda} \quad .
\end{equation}
We will now prove that there exists such a $c$ that
\begin{equation}
 \label{c-eq}
\frac{\partial\det\tilde{G}}{\partial\theta_{kl}}c+\lambda\frac{\partial\lambda}{\partial
l_{ij}}\frac{\partial^2\det\tilde{G}}{\partial\theta_{ij}\partial\theta_{kl}} =0 \quad .
\end{equation}
Because the range of the matrix $\frac{\partial\theta_{ij}}{\partial l_{kl}}$ is the whole space
perpendicular to the vector $\vec{l}=(l_{ij})$ and the vector
$\frac{\partial\det\tilde{G}}{\partial\theta_{kl}}$
is proportional to $\vec{l}$, it is enough to compute
\begin{equation}
\begin{split}
&\left(\frac{\partial\det\tilde{G}}{\partial\theta_{kl}}c+\lambda\frac{\partial\lambda}{\partial
l_{ij}}\frac{\partial^2\det\tilde{G}}{\partial\theta_{ij}\partial\theta_{kl}}\right)
 \frac{\partial\theta_{kl}}{\partial l_{mn}}=
\lambda\frac{\partial\lambda}{\partial
l_{ij}}\frac{\partial}{\partial
l_{mn}}\frac{\partial\det\tilde{G}}{\partial\theta_{ij}}=\\
&=\lambda\frac{\partial\lambda}{\partial
l_{ij}}\frac{\partial}{\partial
l_{mn}}\frac{l_{ij}}{\lambda}=\lambda\frac{\partial\lambda}{\partial
l_{ij}}\left(\frac{\delta_{(ij)(mn)}}{\lambda}-l_{ij}\frac{\partial\lambda}{\partial
l_{mn}}\frac{1}{(\lambda)^2}\right) \quad .
\end{split}
\end{equation}
On the other hand we know that, since $\lambda$ is of scaling dimension 1, $l_{ij}\frac{\partial\lambda}{\partial
l_{ij}}=\lambda$, and thus
\begin{equation}
 \lambda\frac{\partial\lambda}{\partial
l_{ij}}\left(\frac{\delta_{(ij)(mn)}}{\lambda}-l_{ij}\frac{\partial\lambda}{
\partial
l_{mn}}\frac{1}{(\lambda)^2}\right)=
\frac{\partial\lambda}{\partial
l_{mn}}-\frac{\partial\lambda}{\partial
l_{mn}}=0 \quad .
\end{equation}
Let us also remind that:
\begin{equation}
 \frac{\partial\lambda}{\partial
l_{mn}}\frac{\partial\det\tilde{G}}{\partial\theta_{mn}}=\frac{\partial\lambda}{\partial
l_{mn}}\frac{l_{mn}}{\lambda}=1 \quad .
\end{equation}

\subsection{Computation of $\det'\frac{\partial\theta_{ij}}{\partial l_{ij}}$}
\label{bunch2}

In this section we will prove that
\begin{equation}
 {\det}'\frac{\partial\theta_{ij}}{\partial
l_{ij}}= \frac{3^3}{2^5}  \frac{|l|^2}{\prod S_i^2} V^3 \quad .
\end{equation}
We will start from the formula valid for a spherical tetrahedron (Lemma \ref{lm:det}):
\begin{equation}
 \det\frac{\partial\theta_{ij}}{\partial
\tilde{l}_{ij}}=-\frac{\det\tilde{G}}{\det G}\ .
\end{equation}
As mentioned above we set $\tilde{l}_{ij}=\epsilon l_{ij}$ and
$\theta^\epsilon_{ij}=\theta(\epsilon l_{kl})$ and take the limit $\epsilon\rightarrow 0$ in
the end. In this limit the angles converge to the angles of a flat tetrahedron with lengths
$l_{ij}$.

Let us remind that
\begin{equation}
 \det G=\frac{1}{8}\epsilon^{6} \det \underbrace{\left(\begin{array}{
ccc }
               0 & 1 &\cdots\\
               1 & l_{ij}^2 &\cdots\\
               \vdots &\cdots &\cdots
              \end{array}
\right)}_{C}+O(\epsilon^8) \quad .
\end{equation}
Let us notice that because $G$ (in the spherical case) is a function of $\cos\epsilon l_{ij}$,
its expansion around $\epsilon=0$ is an analytic function in
$\epsilon^2$ and not only in $\epsilon$. The same holds for the matrix $\tilde{G}$ since it
is
\begin{equation}
 \tilde{G}_{ij}=\frac{1}{\sqrt{G^*_{ii}}}G^*_{ij}\frac{1}{\sqrt{G^*_{jj}}} \quad ,
\end{equation}
where $G^*_{ij}$ is the cofactor matrix of $G$ and $\sqrt{G^*_{ii}}$ is $\epsilon^3$
times an analytic
function in $\epsilon^2$.

Hence we know that for the vector
$\vec{S}=(S_1,S_2,S_3,S_4)$ (the single null eigenvector in the limit $\epsilon= 0$)
\begin{equation}
 \tilde{G}\vec{S}=O(\epsilon^2)\ ,\ (\vec{S},\tilde{G}\vec{S})=O(\epsilon^2)\ ,
\end{equation}
then also $\det \tilde{G}=O(\epsilon^2)$ and
\begin{equation}
 \det\tilde{G}={\det}'\tilde{G}\frac{(\vec{S},\tilde{G}\vec{S})}{|S|^2}
+O(\epsilon^4)\ .
\end{equation}
Moreover
\begin{equation}
 \epsilon\frac{\partial}{\partial\epsilon}\frac{(\vec{S},\tilde{G}\vec{S})}{|S|^2}
=2\frac{(\vec{S},\tilde{G}\vec{S})}{|S|^2}+O(\epsilon^3)\ .
\end{equation}
We have
\begin{align}
 \epsilon\frac{\partial}{\partial\epsilon}\frac{(\vec{S},\tilde{G}\vec{S})}{|S|^2}
&=\sum_{(ij)}l_{ij}\partial_{l_{ij}}\frac{
(\vec{S},\tilde{G}\vec{S})}{|S|^2}
=\frac{(\vec{S},\sum_{(ij)}l_{ij}\partial_{l_{ij}}\tilde{G}\vec{S})}{|S|^2}=\\
&=-2\frac{
\sum_{(km)(ij)}S_kS_m\sin\theta_{km}l_{ij}\frac{\partial\theta^\epsilon_{km}}{
\partial l_{ij}}}{|S|^2}=- 3 \frac{V}{|S|^2}\sum_{(km)(ij)}l_{km}l_{ij}\frac{
\partial\theta^\epsilon_{km}}{\partial
l_{ij}} \quad .
\end{align}
Similarly, we know that $\sum_{ij}\frac{\partial\theta^\epsilon_{km}}{\partial
l_{ij}}l_{ij}=O(\epsilon^2)$ and
$\sum_{ijkm}l_{km}\frac{\partial\theta^\epsilon_{km}}{\partial
l_{ij}}l_{ij}=O(\epsilon^2)$, so
\begin{align}
 \det\frac{\partial\theta_{km}}{\partial
\tilde{l}_{ij}}=\epsilon^{-6}\det\frac{\partial\theta^\epsilon_{km}}{\partial
l_{ij}}=\epsilon^{-6}{\det}'\frac{\partial\theta^\epsilon_{km}}{\partial
l_{ij}}\frac{\sum_{(ij)(km)}l_{km}\frac{\partial\theta^\epsilon_{km}}{\partial
l_{ij}}l_{ij}}{|l|^2} \quad .
\end{align}
Eventually, we have
\begin{equation}
 \epsilon^{-6}{\det}'\frac{\partial\theta^\epsilon_{km}}{\partial
l_{ij}}\frac{\sum_{ijkm}l_{km}\frac{\partial\theta^\epsilon_{km}}{\partial
l_{ij}}l_{ij}}{|l|^2}=12 \epsilon^{-6}\frac{{\det}'\tilde{G}}{\det
C}\frac{V}{|S|^2}\sum_{kmij}l_{km}l_{ij}\frac{
\partial\theta^\epsilon_{km}}{\partial
l_{ij}}+O(\epsilon^{-3}) \quad .
\end{equation}
and so
\begin{equation}
{\det}'\frac{\partial\theta^\epsilon_{km}}{\partial
l_{ij}}=\frac{12 |l|^2}{|S|^2}\frac{V}{\det
C}{\det}'\tilde{G}+O(\epsilon) \quad.
\end{equation}
Now we can use the identities from appendix \ref{bunch41}
\begin{equation}
 {\det}'\tilde{G}=\frac{3^4}{2^2} \left(\sum_i S_i^2\right)\frac{
V^4}{\prod S_i^2},\quad \det C=8 (3!)^2 V^2 \quad .
\end{equation}
Finally, in the limit $\epsilon\rightarrow 0$, $(\theta_{ij}=\lim
\theta_{ij}^\epsilon$):
\begin{equation}\label{eq:detprimtheta}
{\det}'\frac{\partial\theta_{km}}{\partial
l_{ij}}= \frac{3^3}{2^5} \frac{|l|^2}{\prod S_i^2}V^3 \quad .
\end{equation}

\section{Technical computations}
\label{tech:sec}

In this appendix we give some explicit computations needed in the main body of the paper.

\subsection{Weak equivalences}
\label{weak:sec}

In the following we will use the notation introduced in section \ref{equiv:sec}. We can
compute
\begin{equation}
\begin{split}
 0\equiv L_+ A_+^k &=k(L_+A_+)A_+^{k-1}+A_+^{k+1}\\
&=\left(\frac{k}{2}+1\right)A_+^{k+1}+\frac{k}{2}A_-^2A_+^{k-1}
-k\cos 2\tilde{\theta}\ A_+^{k-1} \quad ,
\end{split}
\end{equation}
such that
\begin{equation}
\label{eq:w-1}
 A_-^2A_+^{k-1}\equiv -\frac{k+2}{k}A_+^{k+1}+2\cos 2\tilde{\theta}\ A_+^{k-1} \quad .
\end{equation}
Similarly, we can derive an identity by acting on $A_+^k$ with $L^*_-$:
\begin{align}
 0\equiv L_-^*A_+^k=(k+1)A_-A_+^k-kA_+^{k-1}\quad ,\\
 \implies \,A_-A_+^k\equiv \frac{k}{k+1}A_+^{k-1} \quad .\label{eq:w-2}
\end{align}
By acting again on \eqref{eq:w-2} we obtain:
\begin{equation}
 L_-^*(A_-A_+^k)=\frac{1}{2}A_+^{k+2}+\left(k+\frac{3}{2}
\right)A_-^2A_+^k-kA_-A_+^{k-1}-\cos 2\tilde{\theta}\ A_+^k\equiv 0 \quad .
\end{equation}
Hence using \eqref{eq:w-1} and \eqref{eq:w-2} we have
\begin{equation}\label{eq:w-fin}
\begin{split}
 0&\equiv \frac{1}{2}A_+^{k+2}+\left(k+\frac{3}{2}
\right)\left(-\frac{k+3}{k+1}A_+^{k+2}+2\cos 2\tilde{\theta}\ A_+^{k}
\right)-k
\frac{k+1}{k}A_+^{k}
-\cos 2\theta\ A_+^k\\
&=-\frac{(k+2)^2}{k+1}A_+^{k+2}+2(k+1)\cos 2\tilde{\theta}\ A_+^k-(k-1)A_+^{k-2} \quad .
\end{split}
\end{equation}

\subsection{Proof of the lemma}
\label{proof:sec}

In this section we will prove the following lemma:
\begin{replm}{ergod3}
For every $m\geq 0$
\begin{equation} \label{eq:lemma1'}
 \sum_{k \leq m}(2\beta^k_{m+1-k}+\beta^k_{m-k})A_k i^{m+1-k}\sin\left(\theta -
\frac{\pi}{2}(m-k)\right)=0 \quad ,
\end{equation}
where
\begin{equation}
 \beta^k_{m}=\frac{(-k- \frac{1}{2})_{m}}{ m!}\in {\mathbb R} \quad ,
\end{equation}
and
\begin{equation}
(a)_m = a \cdot (a-1) \cdot \ldots \cdot (a - m + 1),\quad (a)_0 = 1 \quad .
\end{equation}
\end{replm}
To do so, we need:
\begin{lm}\label{ergod2}
The following equality holds:
\begin{equation}
 \frac{1}{(l\pm 1)^{k+\frac{1}{2}}}=\sum_{m\geq k}\frac{(\pm
1)^{m-k}\beta^k_{m-k}}{l^{m+\frac{1}{2}}} \quad .
\end{equation}
\end{lm}

\begin{proof}
\begin{align}
\frac{1}{(l \pm 1)^{k+\frac{1}{2}}} & =  \sum_{n=0}^\infty \binom{-k
-\frac{1}{2}}{n} l^{-k-\frac{1}{2} -n} (\pm 1)^{n} = \sum_{n=0}^\infty \frac{(-k
- \frac{1}{2})_n}{n!} \frac{1}{l^{k + n + \frac{1}{2}}} (\pm 1)^n \nonumber \\
& \overset{m:=k+n}{=}  \sum_{m \geq k} (\pm 1)^{m-k}
\frac{(-k-\frac{1}{2})_{m-k}}{(m-k)!}  \,=\, \sum_{m \geq k} (\pm 1)^{m-k}
\frac{\beta^k_{m-k}}{l^{m+\frac{1}{2}}}  \quad . \qedhere
\end{align}
\end{proof}
In the following we will use  Lemma  \ref{ergod2} to prove
Lemma \ref{ergod3}:
\begin{proof}[Proof of Lemma \ref{ergod3}]
In any stationary point
we have by Lemma \ref{ergod2}
\begin{equation}
 \tilde{P}_{l\pm 1}\equiv \sum_{k \geq 0} \frac{e^{i (l \pm 1) \theta}}{(l
\pm 1)^{k +\frac{1}{2}}} A_k(\theta) \equiv \sum_{k\geq 0} e^{il\theta}
\sum_{m\geq k} \frac{(\pm
1)^{m-k}}{l^{m + \frac{1}{2}}}\beta^k_{m-k} A_k(\theta)e^{\pm i\theta} \quad .
\end{equation}
We thus have
\begin{equation}
 l(\tilde{P}_{l+1}+\tilde{P}_{l-1}-2\cos\theta \tilde{P}_l)
\equiv \sum_{k\geq 0} \frac{e^{il\theta}}{l^{k - \frac{1}{2}}} A_k \left(\sum_{m
\geq k}\frac{\beta^{k}_{m-k}}{l^{m-k}} (e^{i\theta} + (-1)^{m-k} e^{-i\theta}) -
2\cos\theta
\right) \quad .
\end{equation}
A simple algebraic manipulation gives
\begin{equation}
e^{i\theta} + (-1)^{m-k} e^{-i \theta} = \left\{
  \begin{array}{l l}
    2 \cos \theta & \quad \text{if $m = k$}\\
    2 i^{m-k} \cos\left(\theta - \frac{\pi}{2} (m-k)\right) & \quad
\text{if $m> k$}\\
  \end{array} \right. \quad ,
\end{equation}
such that we obtain:
\begin{equation} \label{eq:rec_part1}
l(\tilde{P}_{l+1} + \tilde{P}_{l-1} - 2 \cos\theta P_l)\equiv
\sum_{k\geq 0} e^{il\theta} A_k
\sum_{m\geq
k} \frac{2 \beta^{k}_{m+1-k}}{l^{m+\frac{1}{2}}}
i^{m+1-k}\underbrace{\cos\left(\theta-\frac{\pi}{2}(m-k+1)\right)}_{\sin
\left(\theta -\frac{\pi}{2}(m-k) \right)}  \quad .
\end{equation}
We also have
\begin{equation} \label{eq:rec_part2}
 \frac{1}{2}(\tilde{P}_{l+1}-\tilde{P}_{l-1})=
\sum_{k\geq 0} e^{il\theta} A_k \sum_{m\geq
k} \frac{\beta^k_{m-k}}{l^{m+\frac{1}{2}}}
i^{m+1-k}\sin\left(\theta-\frac{\pi}{2}(m-k)\right) \quad .
\end{equation}
By combining \eqref{eq:rec_part1} and \eqref{eq:rec_part2}, we obtain for the full recursion
relation \eqref{eq:rec-Legendre}:
\begin{align}
& \sum_{k\geq 0}e^{il\theta} A_k
\left(
\sum_{m\geq
k} \frac{2\beta^k_{m+1-k} +\beta^k_{m-k}}{l^{m+\frac{1}{2}}}
i^{m+1-k} \sin\left(\theta-\frac{\pi}{2}(m-k)\right)
\right)  \nonumber \\
= \; &  \sum_{m \geq 0} \frac{e^{i l \theta}}{l^{m + \frac{1}{2}}} \left(
\sum_{k \leq m} \left(2 \beta^k_{m+1-k} + \beta^k_{m-k} \right) i^{m+1-k} \, A_k
\, \sin \left( \theta - \frac{\pi}{2} (m-k) \right) \right)
= \, O(l^{-\infty}) \quad .
\end{align}
Thus every single term must be zero. That ends the proof.
\end{proof}

\subsubsection{Expanding $C_j$}
\label{sec:C_j}

In the following we will expand the normalization factor $C_j = \frac{1}{4^j}
\binom{2j}{j}$ up to $O(\frac{1}{j})$. Therefore we use Stirling's series for
the logarithm of the factorial:
\begin{equation}
\ln n! = n \ln n - n + \frac{1}{2} \ln(2 \pi n) + O\left(\frac{1}{n}\right)
\quad .
\end{equation}
Hence
\begin{equation}
\ln (C_j) =
\ln \left(\frac{1}{\sqrt{\pi j}}\right) + O\left( \frac{1}{j} \right) \quad .
\end{equation}
Therefore we obtain:
\begin{equation}
C_j = \frac{1}{\sqrt{\pi j}} e^{O\left( \frac{1}{j} \right)}= \frac{1}{\sqrt{\pi
j}}\left(1+O\left(
\frac{1}{j} \right)\right) \quad .
\end{equation}
Moreover,
since $\ln n!$ admits a complete expansion (neglecting the first terms) in powers of
$\frac{1}{n}$, also $C_j$ can be completely expanded in powers of $\frac{1}{j}$. The same is
true for an expansion in $l$.

\subsection{Theta graph}
\label{theta:sec}

In this section we explain the result that the theta graph $(C^{j_1j_2j_3}_{000})^2$ is
equal to
\begin{equation}
 \frac{1}{2 \pi {S}}\left(1+O\left(\frac{1}{l^2}\right)\right) \quad .
\end{equation}
From \cite{Varshalovich} we have
\begin{equation}
 C^{j_1j_2j_3}_{000}=(-1)^g\frac{g!}{(g-j_1)!(g-j_2)!(g-j_3)!}\sqrt{\frac{
(2g-2j_1)!(2g-2j_2)!(2g-2j_3)!}{(2g+1)!}} \quad ,
\end{equation}
where $2g=j_1+j_2+j_3$. We compute the expansion of $\ln C^{j_1j_2j_3}_{000}$ using the Stirling's formula:
\begin{equation}
 \ln (n!)=n\ln n-n+\frac{1}{2}\ln n+\frac{1}{2}\ln 2\pi+\frac{1}{12n}+O(n^{-2}) \quad ,
\end{equation}
obtaining
\begin{equation}
\begin{split}
\ln \left( (-1)^gC^{j_1j_2j_3}_{000} \right) &= -\frac{1}{4}\ln \left(\frac{(2\pi)^2}{16}
(l_1+l_2+l_3)
(-l_1+l_2+l_3)
(l_1-l_2+l_3)
(l_1+l_2-l_3)\right)\\&+\, O(l^{-2}) \quad .
\end{split}
\end{equation}
This is exactly
\begin{equation}
-\frac{1}{4}\ln
4 \pi^2{S}^2 +O(l^{-2}) \quad ,
\end{equation}
where ${S}$ is the area of the triangle with edge lengths $l_i$.
We conclude that the theta graph $(C^{j_1j_2j_3}_{000})^2$ is equal to
\begin{equation}
 \frac{1}{2 \pi {S}}\left(1+O\left(\frac{1}{l^2}\right)\right) \quad .
\end{equation}

\subsection{Kinetic term in equilateral case}
\label{iso:sec}

Let us introduce
\begin{equation}
 M_\lambda=\left(\begin{array}{c c c c c c c}
   0 & a & a & a & a & a &a\\
   a & b-\lambda & c-\lambda &c-\lambda & c-\lambda &c-\lambda &-\lambda\\
   a & c-\lambda & b-\lambda &c-\lambda & c-\lambda &-\lambda & c-\lambda\\
   a & c-\lambda & c-\lambda &b-\lambda & -\lambda &c-\lambda & c-\lambda\\
   a & c-\lambda & c-\lambda &-\lambda & b-\lambda &c-\lambda & c-\lambda\\
   a & c-\lambda & -\lambda &c-\lambda & c-\lambda &b-\lambda & c-\lambda\\
   a & -\lambda & c-\lambda &c-\lambda & c-\lambda &c-\lambda & b-\lambda
 \end{array}
\right) \quad ,
\end{equation}
where $a=-\sqrt{2}\frac{64}{81}$, $b=\frac{\sqrt{3}}{4}$ and $c=\frac{1}{2 \sqrt{3}}$. In the
equilateral case (all $l$ equal to $1$) the kinetic term is of the form
\begin{equation}
 -iM_0 \quad .
\end{equation}
Let us note that
\begin{equation}
 \det M_\lambda=\det M_0\not=0 \quad ,
\end{equation}
and all $M_\lambda$ are symmetric. Thus all of them have the same number of positive and negative
eigenvalues.
Matrix $M_{\lambda}$ for $\lambda=c$ is similar  (have the same determinant) by simultaneous
permutation of rows and columns to
the matrix
\begin{equation}
M'= \left(\begin{array}{c c c c c c c}
  b-c & -c & 0   & 0  & 0  & 0  &  a\\
   -c & b-c& 0   & 0  & 0  & 0  &  a\\
   0  &  0 & b-c & -c & 0  & 0  &  a\\
   0  &  0 & -c  & b-c& 0  & 0  &  a\\
   0  &  0 & 0   & 0  &b-c & -c &  a\\
   0  &  0 & 0   & 0  &-c  & b-c&   a\\
   a  & a  & a   & a  & a  & a  & 0
 \end{array}
\right) \quad .
\end{equation}
The matrix $M'$ restricted to its first $6$ rows and columns has $3$ positive and $3$ negative
eigenvalues. Applying the min-max principle \cite{Reed-Simon} to $M'$
and $-M'$ shows that
$M'$ has at least three positive and three negative eigenvalues. Together with the fact that
determinant is positive it shows that there are $4$ positive and $3$ negative eigenvalues.

Hence, the matrix of kinetic term has $4$ $-i{\mathbb R}_+$ eigenvalues and $3$ $i{\mathbb R}_+$
and the same is true for matrix $(-\mathcal{H}^{-1})$.

\section{Dupuis-Livine form and stationary points}
\label{DL-app}

In this section we will prove the following lemma:

\begin{lm}
Suppose that the integral is of the form as
\begin{equation}
  \int d\theta \frac{e^{i\eta}}{l^{\lambda}}e^S \quad ,
\end{equation}
where $S(\theta_i)$ has an asymptotic expansion around the isolated stationary point
of $S_{-1}(\theta)$ of the form
\begin{equation}
 S=S_{-1}+S_0+S_1+\ldots \quad ,
\end{equation}
and $i^kS_k\in {\mathbb R}$ is a homogeneous function of order $-k$ in $l$. Then the
contribution to the expansion of the integral from this stationary point has the DL
property.
\end{lm}

\begin{proof}
Let us consider the contribution from the isolated stationary point of
$S_{-1}$. They are of the form
\begin{equation}
 \frac{1}{\sqrt{{\mathcal H}}}e^{\tilde S} \quad ,
\end{equation}
where $\tilde{S}$ is given by the contraction of all connected Feynman diagrams. They
are made up of vertices, given by the derivatives of $S_{\geq 0}$, connected by the
propagator $H$, which is the inverse to $(-1)$ times the matrix of second
derivatives of $S_{-1}$.
\begin{equation}
 H= \left(-\partial^2 S_{-1}\right)^{-1},\quad {\mathcal H}=\det \left(-\partial^2 S_{-1}\right)\ .
\end{equation}
Their contribution is computed
by contracting the vertices $V_k$ with
propagators $H$. Since vertices are obtained from derivatives of $S_m$, $m\geq 0$, the
homogeneous degree $\deg V_k$ of this vertex is thus $m$ and the matrix elements of
$i^{\deg
V_k}V_k$ are real. Similarly $iH$ is a real matrix and is of degree $1$.

To conclude, the total contraction is thus of degree
\begin{equation}
 \sum_k \deg V_k+n \quad ,
\end{equation}
where $n$ is the number of propagators in the diagram. Moreover, the complete
contraction multiplied by
\begin{equation}
i^{\sum_k \deg V_k+n}
\end{equation}
is again real as a contraction of real matrices. This proves that expansion is still of DL
form.
\end{proof}

\section{Stationary point analysis}
\label{sec:stationary-point}

In the paper we use an advanced version of the stationary point analysis. This appendix
is intended to explain the details of this method.

\begin{lm}
 Let $S(x)=i(S_{-1}+S_0)+S_0'+\sum_{i>0} S_i$ be an asymptotic expansion of the action such that
\begin{itemize}
 \item $S_i$ is of homogeneous degree $-i$ in $j$,
\item $S_0$ and $S_{-1}$ are real
\item $S_{-1}+S_0$ is homogeneous in $l=j+\frac{1}{2}$
\end{itemize}
and let $x_0$ be an isolated stationary point of $S_{-1}$. Then there is an asymptotic expansion of
the contribution to the integral
\begin{equation}
 \int dx e^{S}
\end{equation}
from the neighbourhood of $x_0$ given as follows:

We can write the asymptotic expansion of $S$ in homogeneous terms in $l$ as
\begin{equation}
 S=i\tilde{S}_{-1}+S_0'+\sum_{i>0} \tilde{S}_i \quad ,
\end{equation}
where $\tilde{S}_{-1}=S_{-1}+S_0$. Let $x_1$ be the
stationary point of $\tilde{S}_{-1}$ obtained by perturbation of $x_0$ (there is exactly one
such stationary point if the matrix of second derivatives of $S_{-1}$ is non-degenerate). The asymptotic expansion of the integral is equal to
\begin{equation}
 \frac{1}{\sqrt{\det (-H)}}e^{\sum_{i\geq -1} A_i}
\end{equation}
where $H$ is the matrix of second derivatives of $\tilde{S}_{-1}$ and $A_{-1}=\tilde{S}_{-1}$
evaluated on $x_1$. The terms $A_i$ for $i\geq 0$ are homogeneous functions of order
$-i$ in $l$ and can be obtained from the Feynman
diagram expansion with the propagator $(-H)^{-1}$ and interaction vertices given by derivatives of
$\tilde{S}_{i}$ for $i\geq 0$.
\end{lm}

The same fact applies when the isolated point is replaced by the isolated orbit of the symmetry group of the
action.

The second fact concerns with integration over only a part of the variables:

\begin{lm}
 Let $S(x,y)=iS_{-1}(x,y)+\sum_{i\geq 0} S_i$ has an isolated stationary point $(x_0,y_0)$ with a
non-degenerate matrix of second derivatives $H$ with the property
\begin{equation}
    H=\begin{pmatrix}
       H_{xx} & H_{xy}\\ H_{xy} & H_{yy}
      \end{pmatrix},\quad H_{yy}\ \text{invertible}\ .
\end{equation}
Then there exists a function $y(x)$ such
that (in the neighbourhood of stationary point)
\begin{equation}
 \nabla_x S_{-1}+\frac{\partial y}{\partial x}\nabla_y S_{-1}=0
\end{equation}
and the asymptotic expansion of $\int dxdy\ e^S$ is equal to asymptotic expansion of
\begin{equation}
\int dx\ e^{\bar{S}}\ ,
\end{equation}
where $\bar{S}$ is obtained by asymptotic expansion of the integral
$e^{\bar{S}}=\int dy\ e^S$.
\end{lm}

\section{Feynman diagrams}
\label{sec:Feyn}

In this subsection we are interested in the next to leading order in the expansion
of the $6j$ symbol.
We will derive
expressions for $S_1$ in terms of Feynman diagrams. Vertices in
this expansion consist of
derivatives of
\begin{equation}
 -\sum_i\frac{1}{2}\ln\sin\theta_{ij}, \quad \frac{-i}{8l_{ij}}\cot\theta_{ij},
\end{equation}
and higher than second derivatives of $|l|\rho\det\tilde{G}$ with respect to $\rho$
and $\theta_{ij}$. Each propagator contributes a weight $|l|^{-1}$.

We only evaluate closed diagrams,
so if the diagram is made up of vertices of valency $n_k$, i.e. the $n_k$-th derivative of a
function
with weight
$|l|^{\alpha_k}$, then the scaling behaviour of the whole diagram is as
\begin{equation}
 |l|^{\sum_k \left(\alpha_k-\frac{n_k}{2}\right)} \quad .
\end{equation}
The only vertices that can contribute up to order $|l|^{-1}$ are thus
\begin{equation*}
 \begin{array}{|c|c|c|c|c|c|c|c|c|c|}
\hline
 \text{Vertex} &
-\frac{1}{2}\ln\sin\theta_{ij} &
-\frac{1}{2}\cot\theta_{ij}&
 -\frac{1}{2}\frac{\partial}{\partial\theta_{ij}}\cot\theta_{ij}&
-\frac{i}{8l_{ij}}\cot\theta_{ij} &
i|l|\partial^3\rho\det\tilde{G} &
i|l|\partial^4\rho\det\tilde{G}\\
\hline
\text{Valency} & 0 &1_{(ij)} &2_{(ij)(ij)} &0&3&4\\
\hline
\text{Order} & |l|^{0} & |l|^{-1/2}& |l|^{-1}& |l|^{-1}& |l|^{-1/2}& |l|^{-1}\\
\hline
 \end{array}
\end{equation*}
Note that the only diagram that is real (up to the order $|l|^{-1}$) is
just the first
vertex (being of order 0). Furthermore, this is also the only contribution of order $|l|^0$.
All
other diagrams are purely imaginary and of order $|l|^{-1}$.

\section{Relation to spin-network kernel formula}
\label{rel-islas}

We will prove that (for $j\in{\mathbb Z}$)
\begin{equation} \label{eq:rel_islas}
 \int_0^\pi \frac{d\phi_1}{\pi}
\left(e^{i\theta}\cos\phi_1\cos\phi_2
+e^{-i\theta}\sin\phi_1\sin\phi_2\right)^{2j}=
\frac{1}{4^j}{\binom{2j}{j}}\left(e^{i2\theta}\cos^2\phi_2
+e^{-i2\theta}\sin^2\phi_2\right)^{j} \quad .
\end{equation}
It is straightforward to check that
\begin{equation}
 e^{i2\theta}\cos^2\phi_2
+e^{-i2\theta}\sin^2\phi_2=\cos{2\theta}
+i\sin 2\theta\cos 2\phi_2 \quad.
\end{equation}
In this way we obtain the formula from \cite{islas}.

To prove \eqref{eq:rel_islas}, we use the following formulas:
\begin{align}
  &2\int_{0}^{\pi/2}d\phi\ \sin^{2\alpha}\phi\cos^{2\beta}\phi
=\frac{\Gamma\left(\alpha+\frac12\right)\Gamma\left(\beta+\frac12\right)}
{\Gamma(\alpha+\beta+1)} \quad ,\\
&\Gamma(n+1)=n!,\ \ \Gamma\left(n+\frac{1}{2}\right)=\frac{(2n)!}{4^n
n!}\sqrt{\pi} \quad .
\end{align}
In the case that $k,l\, \in \mathbb{N}$, these formulas can be simplified to:
\begin{equation}
 \int_{0}^{\pi}d\phi\
\sin^{2k}\phi\cos^{2l}\phi=\frac{(2k)!(2l)!\pi}{4^{k+l}k!l!(k+l)!} \quad ,
\end{equation}
\begin{equation}
 \int_{0}^{\pi}d\phi\
\sin^{2k+1}\phi\cos^{2l+1}\phi=0 \quad .
\end{equation}
Expanding the left hand side of \eqref{eq:rel_islas} and
using these formulas we have:
\begin{align}
 &\int_0^\pi \frac{d\phi_1}{\pi}
\left(e^{i\theta}\cos\phi_1\cos\phi_2
+e^{-i\theta}\sin\phi_1\sin\phi_2\right)^{2j}=\\
&\sum_{n=0}^{2j} \binom{2j}{n}
e^{i(2j-2n)\theta}\cos^{2j-n}\phi_2\sin^{n}\phi_2\int_0^\pi \frac{d\phi_1}{\pi}
\cos^{2j-n}\phi\sin^{n}\phi \quad .
\end{align}
This is equal to ($k:=2n$)
\begin{equation}
 \sum_{k=0}^{j} \binom{2j}{2k}
e^{i(j-2k)2\theta}\cos^{2j-2k}\phi_2\sin^{2k}\phi_2\frac{1}{\pi}
\frac{(2k)!(2j-2k)!\pi}{4^{j}k!(j-k)!j!} \quad .
\end{equation}
The factors in $j$ and $k$ can be rewritten in terms of binomial coefficients
\begin{equation}
 \binom{2j}{2k}\frac{(2k)!(2j-2k)!}{4^{j}k!(j-k)!j!}=
\frac{(2j)!}{4^{j}k!(j-k)!j!}=\frac{1}{4^j}\binom{2j}{j}\ \binom{j}{k} \quad ,
\end{equation}
such that we obtain the final result:
\begin{equation}
\begin{split}
 \frac{1}{4^j}\binom{2j}{j}&\sum_{k=0}^j \binom{j}{k}
\left(e^{i2\theta}\cos^2\phi_2\right)^{j-k}
\left(e^{-i2\theta}\sin^2\phi_2\right)^{k}=\\&=\underbrace{\frac{1}{4^j}
\binom{2j}{j}}_{C_j}
\left(e^{i2\theta}\cos^2\phi_2
+e^{-i2\theta}\sin^2\phi_2\right)^{j} \quad .
\end{split}
\end{equation}
This explains the occurrence of $C_j$ in our formulas, which is absent in integral kernel approach \cite{islas}.

\bibliographystyle{utphys}
\bibliography{pr-ref}

\end{document}